\documentclass[11pt,reqno,twoside]{amsart}
\usepackage{amscd}
\usepackage{amsfonts}
\usepackage{amsmath}
\usepackage{amssymb}
\usepackage{amsthm}
\usepackage{amsaddr}
\usepackage{fancyhdr}
\usepackage{latexsym}
\usepackage[colorlinks=true, pdfstartview=FitV, linkcolor=blue, citecolor=blue, urlcolor=blue]{hyperref}
\usepackage{enumitem}      
\usepackage{mathtools}            
\usepackage{indentfirst} 
\usepackage{color}
\usepackage{caption}          
\usepackage[normalem]{ulem}
\newcommand{\nn}{\nonumber}
\newcommand{\p}{\partial}
\newcommand{\ve}{\varepsilon}

\newcommand{\no}[1]{\left\| #1 \right\|}
\newcommand{\what}{\widehat}
\usepackage{thmtools}
\declaretheoremstyle[headfont=\kpfonts]{normalhead}

\newtheorem{theorem}{Theorem}[section]
\newtheorem{proposition}{Proposition}[section]
\newtheorem{lemma}{Lemma}[section]
\newtheorem{corollary}{Corollary}[section]

\theoremstyle{definition}

\newtheorem{remark}{Remark}[section]

\allowdisplaybreaks
\numberwithin{equation}{section}
\numberwithin{figure}{section}
\usepackage{geometry}
\makeatletter
\@namedef{subjclassname@2020}{%
  \textup{2020} Mathematics Subject Classification}
\makeatother

\begin{document}

\title{
On the proximity between the wave dynamics of the integrable focusing nonlinear Schr\"odinger equation 
\\
and its non-integrable generalizations
}

\author{Dirk Hennig$^{\MakeLowercase{a}}$, Nikos I. Karachalios$^{\MakeLowercase{a}}$, Dionyssios Mantzavinos$^{\MakeLowercase{b}}$, 
\\
Jesus Cuevas-Maraver$^{\MakeLowercase{c}}$, and Ioannis G. Stratis$^{\MakeLowercase{d}}$}
	\address{\normalfont $^a$Department of Mathematics,  University of Thessaly, 35100, Lamia, Greece\vspace*{-2mm}}
    \email{\!dirkhennig@uth.gr, \!karan@uth.gr}
	\address{\normalfont $^b$Department of Mathematics, University of Kansas, Lawrence, KS 66045, USA\vspace*{-2mm}} \email{\!mantzavinos@ku.edu}
	\address{\normalfont $^c$Grupo de F\'{i}sica No Lineal, Departamento de F\'{i}sica Aplicada I,
	Universidad de Sevilla. Escuela Polit\'{e}cnica Superior, C/ Virgen de \'{A}frica, 7, 41011-Sevilla, Spain \\[0mm]
	Instituto de Matem\'{a}ticas de la Universidad de Sevilla (IMUS). Edificio Celestino Mutis. Avda. Reina Mercedes s/n, 41012-Sevilla, Spain \vspace*{-2mm}}
\email{\!jcuevas@us.es}
	\address{\normalfont $^d$Department of Mathematics, National and Kapodistrian University of Athens, Panepistimioupolis,
	GR-15784 Athens, Greece}
 \email{\!istratis@math.uoa.gr}

\thanks{\textit{Acknowledgements.} D.M. gratefully acknowledges support from the U.S. National Science Foundation (NSF-DMS 2206270). J.C.-M. acknowledges support from the EU (FEDER program 2014-2020) through MCIN/AEI/10.13039/501100011033 under the project PID2020-112620GB-I00.}
\date{July 28, 2023}

\begin{abstract}
The question of whether features and behaviors that are characteristic to completely integrable systems persist in the transition to non-integrable settings is a central one in the field of nonlinear dispersive equations. In this work, we investigate this topic in the context of focusing nonlinear Schr\"odinger (NLS) equations. In particular, we consider non-integrable counterparts of the (integrable) focusing cubic NLS equation, which are distinct generalizations of cubic NLS and involve a broad class of nonlinearities, with the cases of power and saturable nonlinearities serving as illustrative examples.
This is a notably different direction from the one explored in other works, where the non-integrable models considered are only small perturbations of the integrable one. 
We study the Cauchy problem on the real line for both  vanishing and non-vanishing boundary conditions at infinity and quantify the proximity of solutions between the integrable and non-integrable models via estimates in appropriate metrics as well as pointwise. 
These results establish that the distance of solutions grows at most linearly with respect to time, while the growth rate of each solution is chiefly controlled by the size of the initial data and  the nonlinearity parameters. 
A major implication of these closeness estimates is that integrable dynamics emerging from small initial conditions may persist in the non-integrable setting for significantly long times. 
In the case of zero boundary conditions at infinity, this persistence includes soliton and soliton collision dynamics, while in the case of nonzero boundary conditions at infinity, it establishes the nonlinear behavior of the non-integrable models at the early stages of the ubiquitous phenomenon of modulational instability. For this latter and more challenging type of boundary conditions, the closeness estimates are proved with the aid of new results concerning the local existence of solutions to the non-integrable models.  
In addition to the infinite line, we also consider the cubic NLS equation and its non-integrable generalizations in the context of initial-boundary value problems on a finite interval. Apart from their own independent interest and features such as global existence of solutions (which does not occur in the infinite domain setting), such problems are naturally used to numerically simulate the Cauchy problem on the real line, thereby justifying the excellent agreement between the numerical findings and the theoretical results of this work. 
\end{abstract}

\vspace*{-0.5cm}
\maketitle
\markboth
{D. Hennig, N.I. Karachalios, D. Mantzavinos, J. Cuevas-Maraver, I.G. Stratis}
{On the proximity between the wave dynamics of integrable focusing NLS and its non-integrable generalizations}

%
%
%
\section{Introduction and main results}
One of the most important notions concerning dynamical systems is that of {\em structural stability}, e.g. see \cite{book2}, Section 3.7. Consider as an example the Cauchy (initial value) problem for the following semilinear  evolution equation 
\begin{equation}
\label{gev1}
\begin{aligned}
&u_t+\mathcal{L}(u)+\mathcal{N}(u)=0,
\\
&u(x, 0)=u_0(x),
\end{aligned}
\end{equation}
along with its perturbed counterpart
\begin{equation}
\label{gev2}
\begin{aligned}
&u_t+\mathcal{L}(u)+\mathcal{N}(u)+f(u,x,t)=0,
\\
&u(x, 0)=u_0(x)+p_0(x).
\end{aligned}
\end{equation}
Here, $u=u(x, t)$, $\mathcal{L}$ is a linear differential operator and $\mathcal{N}$ represents the nonlinearity. Furthermore, the perturbative term in \eqref{gev2} may represent external forces, dissipation terms or other effects, while $p_0(x)$ is a perturbation of the initial condition $u_0(x)$ of \eqref{gev1}.

The solutions of \eqref{gev1} are defined in a suitable phase space $\mathcal{X}$. 
In the context of the structural stability theory for \eqref{gev1}, the  perturbations $f(u,x,t)$ and $p_0(x)$ are small when measured in some suitable norms relevant to  $\mathcal{X}$. Then, the main question is whether the solution of the perturbed system~\eqref{gev2} deviates too far from the solution of the original system \eqref{gev1} or not.  As underlined in \cite{book2}, if the evolution equation~\eqref{gev1} ``is very unstable, then one could doubt on its ability to accurately simulate (either numerically or theoretically) a real-life system''.  On the other hand, it is also emphasized in \cite{book2} that  ``any equation with a good well-posedness theory is also likely to have a good stability theory, by modifying the arguments used to prove well-posedness suitably''. In particular,  stability results for~\eqref{gev1} generalize the property of \textit{continuous dependence} of the solution on the initial data (one of the three components of Hadamard well-posedness, the other two being existence and uniqueness), which can be obtained as a reduction of such general results in the special case of \eqref{gev2} with $f = 0$.
Similarly, the uniqueness theory of \eqref{gev1} can be approached through the special case $f=p_0=0$. 

The above statements on stability theory emphasize the importance of this direction of study and motivate us to examine, in the present paper, a modified notion of stability for an evolution equation of the form \eqref{gev1} in the framework of Hamiltonian  nonlinear Schr\"{o}dinger  equations in one spatial dimension.  The role of the semilinear equation in \eqref{gev1} is assigned to the  \textit{focusing cubic} NLS equation
\begin{equation}\label{NLS}
i u_t + \nu u_{xx}+\mu |u|^{2}u=0, \quad \nu, \mu > 0,
\end{equation}
which is one of the fundamental  \textit{integrable nonlinear} dispersive partial differential equations with numerous applications in a broad range of areas within mathematical physics. 
Furthermore, instead of considering a ``forced'' counterpart of \eqref{NLS} as in  \eqref{gev2}, we consider focusing {\em non-integrable counterparts} of \eqref{NLS} of the general form
\begin{equation}\label{noninNLS}
iU_t+\nu\,U_{xx} + \gamma F(|U|^2)\,U=0, \quad {\nu, \gamma >0},
\end{equation}
where $F:\mathbb R\rightarrow\mathbb R$ is a sufficiently smooth function satisfying  standard conditions that are specified later. 
Therefore, instead of small perturbations of \eqref{NLS}, our investigation in the framework of stability is of different nature, as our primary goal is to compare systems of the \textit{same class} \eqref{gev1} but with \textit{different nonlinearities}: namely, the integrable NLS \eqref{NLS} with nonlinearity $\mathcal{N}_{\small{I}}(u)=\mu|u|^2u$ against its non-integrable counterparts \eqref{noninNLS} with nonlinearities $\mathcal{N}_{\small{\text{NI}}}(U)=\gamma F(|U|^2)U$.

Actually,  this stability analysis is directly relevant to the proximity between the solutions of the completely integrable NLS equation \eqref{NLS} and its non-integrable counterparts \eqref{noninNLS}.  In particular, it sheds light on the potential \textit{persistence of integrable dynamics}, as these are defined by the solutions and dynamical behavior of the integrable NLS equation \eqref{NLS}, in the setting of the general family of non-integrable NLS equations \eqref{noninNLS}.    
Studies of this generic question in the context of nonlinear lattices (discrete NLS and Ginzburg-Landau equations) are given  in \cite{DJNa,DJNb,DJNc}.
Nevertheless, the investigation of this question in the context of the NLS partial differential equations is a considerably more intricate task due to the following reasons:
\begin{enumerate}[label=(\arabic*), leftmargin=6.5mm, topsep=2mm, itemsep=2mm]
\item
\textit{Differences between the well-posedness of the integrable and the non-integrable NLS equations}. The cubic NLS equation \eqref{NLS} is \textit{globally} well-posed in all cases of significant boundary conditions. This is not generally the case for the non-integrable equations \eqref{noninNLS}, for which global existence of solutions in time can be assured under smallness conditions for the initial data in suitable norms. Therefore, the investigation of stability may be carried out under either (i) the assumption of a restricted class of small initial data, in order to ensure global existence in time for the non-integrable model, or (ii) the restriction to finite times dictated by the maximal interval of existence of the non-integrable model. 
\item
\textit{Different boundary conditions  are associated with different dynamics.} The integrable NLS equation~\eqref{NLS} exhibits a rich class of analytical solutions depending on the boundary conditions with which it is supplemented. The question of  persistence of these analytical solutions in the dynamics of the non-integrable NLS  instigates a variety of further studies in terms of  potential proximity or deviation of solutions of the integrable and non-integrable NLS, respectively.
\end{enumerate}

On account of the above considerations, the main result of this work is quite general and  can be stated as follows.  
\begin{theorem}
\label{Theorem:stability}
Consider the integrable NLS equation \eqref{NLS} and the family of its non-integrable counterparts \eqref{noninNLS}, where the nonlinearity function $F$ satisfies the standard conditions
\begin{equation}\label{F-prop}
|F(x)-F(y)| \leq K(|x|^{p-1}+|y|^{p-1})|x-y|, 
\quad F(0)=0, \quad |F'(x)| \leq K |x|^{p-1},
\end{equation} 
for some $p\geq 1$, a constant $K>0$ and any $x,y \geq 0$. 
Let $[0, T_{\max})$ define a common maximal interval of existence for the solutions $u(x, t)$ and $U(x, t)$ to \eqref{NLS} and \eqref{noninNLS}, respectively, associated with  initial data $u_0(x)$ and $U_0(x)$ in a suitable Banach space $\mathcal{X}$ governed by the supplemented common boundary conditions for both equations. Furthermore, assume that  for every  $0<\varepsilon<1$ the distance in $\mathcal{X}$ between these initial data is of order $\mathcal{O}(\varepsilon^3)$ while  the individual norms of these data in $\mathcal{X}$ are of order $\mathcal{O}(\varepsilon)$, i.e.
\begin{gather}
\label{eq:distance0}
\left\| u_0 - U_0 \right\|_{\mathcal{X}}\leq C \varepsilon^3,
\\
\label{eq:distance01}
\left\| u_0 \right\|_{\mathcal{X}}\leq C \varepsilon,
\   
\left\| U_0 \right\|_{\mathcal{X}}\leq C \varepsilon,
\end{gather}
for some generic  constant $C>0$. 	
Then, for arbitrary finite $0<T_{\small{f}}<T_{\max}$, there exists a constant $\widehat{C}=\widehat C(\gamma,\mu, \nu, C,T_{\small{f}})$ such that the corresponding solutions $u(x, t)$ and $U(x, t)$ satisfy the estimate 	
\begin{equation}\label{eq:boundy}
\sup_{t\in [0, T_f]}\left\|u(t)-U(t)\right\|_{\mathcal{X}}\le \widehat{C} \varepsilon^3.
\end{equation}
That is, under the assumptions \eqref{eq:distance0} and \eqref{eq:distance01}, the distance of solutions $\Delta:=u-U$ measured in $\mathcal{X}$ is of order $\mathcal{O}(\varepsilon^3)$ for all $t\in [0, T_f]$.
\end{theorem}

In view of the condition \eqref{eq:distance0}, the distance inequality \eqref{eq:boundy} provides a generalization of the notion of continuous dependence of solutions on the initial data for the  equations \eqref{NLS} and \eqref{noninNLS}, at least for small initial data in the sense of the condition \eqref{eq:distance01}. 

It should be noted that the class of nonlinearities satisfying the standard conditions~\eqref{F-prop} is  quite broad and includes, among others, the following important cases that correspond to non-integrable NLS models: 
\begin{enumerate}[label=(\roman*), leftmargin=6mm,itemsep=2mm]
\item The general \textit{power nonlinearity} $F(x) = x^p$, which gives rise to the \textit{semilinear Schr\"odinger equation} 
\begin{equation}\label{NLSP}
	i U_t + \nu U_{xx}+\gamma|U|^{2p}U=0.
\end{equation}
Indeed, by the Mean Value Theorem, $x^p-y^p = pc^{p-1} (x-y)$ for some $c\in(x, y)$, thus for any $p\geq 1$ and $x, y \in \mathbb R$ we have
$|x^p-y^p| = p |c|^{p-1} |x-y| \leq p \left(|x|^{p-1} + |y|^{p-1}\right) |x-y|$ 
and the conditions \eqref{F-prop} are satisfied with $K=p$.
\item The rational nonlinearity $F(x) = \dfrac{x}{\kappa(1+x)}$, which corresponds to what is known as the \textit{saturable NLS equation} (e.g. see \cite{gh1991,Satn1,BorisSat})
\begin{equation}\label{Sat2}
iU_t + \nu U_{xx}+\frac{\gamma|U|^2 U}{1+\kappa|U|^2}=0.
\end{equation}	
In this case, for any $x, y \geq 0$ we have $\left|F'(x)\right| = \frac{1}{|\kappa| (1+x)^2} \leq \frac{1}{|\kappa|}$ and  
$
\left|F(x) - F(y)\right|
=
\frac{|x-y|}{|\kappa| (1+x) (1+y)} 
\leq
\frac{1}{|\kappa|} \, |x-y|,
$
thus the conditions \eqref{F-prop} are satisfied with $p=1$ and $K=\frac{1}{|\kappa|}$.
Note that the invertible transformation  $\widetilde U(x, t) = \sqrt \kappa \, e^{-i\Gamma t} \, U(x, t)$  turns equation \eqref{Sat2} into the alternative form
\begin{equation}\label{Sat1}
i \widetilde U_t + \nu \widetilde U_{xx}-\frac{\Gamma \widetilde U}{1+|\widetilde U|^2}=0, \quad \Gamma=\frac{\gamma}{\kappa}.
\end{equation}
\end{enumerate}

Concerning the plethora of analytical solutions to the integrable NLS equation \eqref{NLS}, such as solitons, multi-soliton solutions, bound states, rational solutions and others  \cite{ACbook,MSJMP,JYBook}, an important implication of the stability result of Theorem \ref{Theorem:stability} is  that \textit{small-amplitude localized analytical solutions of the integrable equation \eqref{NLS} persist, in the sense of the distance inequality~\eqref{eq:boundy}, in the non-integrable setting of equation \eqref{noninNLS}}.  
More precisely, the non-integrable equation \eqref{noninNLS} admits small-amplitude solutions of ${\mathcal{O}}(\varepsilon)$ that stay ${\mathcal{O}}(\varepsilon^3)$ close to the well-known analytical solutions of the integrable equation~\eqref{NLS} for any  $t\in [0, T_f]$ (note that one can impose the same initial condition $u(x, 0)=U(x, 0)$ on both equations). This claim is further justified by the fact that, as it turns out, estimate \eqref{eq:boundy} is also valid pointwise since it holds in $L^\infty$.

A case of specific interest is the one of $T_{\max}=\infty$. In this regard, we recall that, in the focusing case $\mu, \nu, \gamma > 0$ with zero boundary conditions at infinity,  global existence for the non-integrable model \eqref{noninNLS} is guaranteed, in general, only for small initial data. Hence, the conditions \eqref{eq:distance01} become particularly relevant. 
Such conditions ensure, for example, the persistence of small-amplitude bright solitons of the integrable NLS equation \eqref{NLS} in the non-integrable setting of equation \eqref{noninNLS} for time intervals that can be significantly long, as well as for even more complicated dynamics like bright soliton collisions. 
Such long time persistence is illustrated by numerical experiments at the end of Section~\ref{zbc-s}. 

Another application of Theorem~\ref{Theorem:stability} concerns the robustness of the persistence of analytical solutions under perturbations in the following sense: solutions that are stable in one system remain stable in the other system as long as they persist in its dynamics. This property can be proved via a transitivity argument combining the proximity estimates of Theorem \ref{Theorem:stability} and orbital stability results for the considered solutions. The discussion leading to Corollary \ref{cazstab-c} below concerns standing waves. For the stability of standing waves, solitons and multi-soliton solutions, we refer the reader to the fundamental results in \cite{CazLi,G1,G2,Martel1,Martel2,Kapit1,Wein86}. 

Importantly, the applicability of Theorem \ref{Theorem:stability} extends beyond the framework of zero boundary conditions at infinity, as the analysis carried out in this work also concerns a broad class of \textit{nonzero  boundary conditions at infinity} in the general form of vanishing profiles on top of a finite nonzero background of constant amplitude. Boundary conditions of this type are of  particular physical relevance. Indeed, in focusing media, nonzero boundary conditions at infinity are associated with the emergence of fascinating dynamics related to the well-known phenomenon of \textit{modulational instability}. This effect, which is also known as Benjamin-Feir instability \cite{bf1967}, refers to the instability of a constant background to long wavelength perturbations and is ubiquitous across nonlinear science; e.g. see the review article~\cite{zo2009}, as well as the more recent works \cite{kffmdgad2010,srkj2007,oos2006} that link modulational instability to the formation of rogue waves in optical media and the open sea.

In the case of the integrable focusing cubic NLS equation \eqref{NLS}, the nonlinear stage of modulational instability induced by the nonzero boundary conditions specified above, was studied in a series of recent works \cite{bm2016,bm2017,blm2018,blm2021} via the inverse scattering transform and the nonlinear steepest descent method of \cite{SDMeth}. In these works, it was rigorously shown that the solution remains bounded at all times, as one would naturally expect due to the complete integrability of the equation \eqref{NLS}.
Nevertheless, in the case of the general non-integrable NLS model \eqref{noninNLS} such integrability techniques are no longer available. In fact, in terms of global existence and regularity of solutions to the non-integrable equation~\eqref{noninNLS}, the corresponding results for nonzero boundary conditions at infinity are markedly different, e.g. see~\cite{book0}  as well as \cite{book1, book2}. 
These observations motivated the study \cite{blmt2018}, which suggests the existence of a universal behavior in modulationally unstable media. In fact, by considering several non-integrable models that belong to the general NLS family \eqref{noninNLS}, evidence is provided that they exhibit the same behavior as the one of the integrable NLS equation \eqref{NLS} established in \cite{bm2016,bm2017,blm2018,blm2021}.

 The stability result of Theorem \ref{Theorem:stability} proved in this work rigorously establishes the persistence of the  nonlinear behavior that was conjectured in \cite{blmt2018}, at least at its early stages. For later times, the proximity between the dynamics depends on the type of the nonlinearity present in the non-integrable model, as illustrated by the  numerical simulations provided at the end of Section \ref{finite-s}. This fact is highlighted by the example of a saturable nonlinearity, where the dynamics exhibits a remarkable proximity to the one of the integrable model (as in the case of zero boundary conditions). Furthermore, the numerical simulations illustrate that the smallness condition \eqref{eq:distance01} for the initial data is by no means restrictive. In particular, they demonstrate that the modulational instability dynamics  emerges from initial data that do satisfy the smallness condition in the non-integrable case. 
 
We emphasize that the numerical findings of Section \ref{finite-s} are predicted by our theoretical results, since Theorem \ref{Theorem:stability} is also proved when the NLS equations \eqref{NLS} and \eqref{noninNLS} are considered on a finite interval and supplemented with (zero or nonzero) Dirichlet or periodic boundary conditions.  Problems of this type fall under the class of \textit{initial-boundary value problems}, in which the spatial domain involves an actual boundary (as opposed to extending to infinity in all directions). Such problems are significant in their own right and have been studied extensively in the literature (e.g. see \cite{f2008,fhm2017} and the references therein). However, they are also directly relevant to numerical studies related to the Cauchy problem, since such studies are performed by approximating the infinite domain by a sufficiently large finite domain, supplemented with appropriate Dirichlet or periodic boundary conditions. (Some additional theoretical implications related to the question of proximity of solutions in the case of approximation by a finite domain are discussed at the end of this introductory section.)
Indeed, the numerical results of Section \ref{finite-s} are in excellent agreement with the analytical estimate \eqref{eq:boundy}. For instance, both for Dirichlet and for periodic boundary conditions,
the analytical arguments 
establish that the growth of the distance $\|u(t)-U(t)\|_{\mathcal{X}}$ is  \textit{at most linear} for any $t\in (0, T_{\max})$, since
\begin{equation}
\label{gr1-intro}
\|u(t)-U(t)\|_{\mathcal{X}}\leq C  t \varepsilon^3
\end{equation}
with $C$ depending only on the nonlinearity parameters $\gamma,\mu$  but not on $t$ (actually, the constant $\widehat{C}$ in~\eqref{eq:boundy} is given by $\widehat{C}=CT_f$).  We emphasize that the proof provides explicit expressions for the dependence of the constant $C$ on the various parameters, suggesting that these expressions can be  adjusted to decrease the linear growth rate as desired.   Regarding the importance of estimate \eqref{gr1-intro},  we also remark that, in a similar context,  the time-growth estimates for the relevant distance function between the solutions of the complex Ginzburg-Landau and the NLS equations, when the inviscid limit of the former is considered~\cite{JWU}, can even be \textit{exponential} in time \cite{OG}.

To the best of our knowledge, among the vast volume of works on NLS-type equations, the one most relevant to the present paper is \cite{PDActa}. However, the results of \cite{PDActa} concern the  \textit{defocusing} (as opposed to focusing) NLS equation
\begin{equation}
    \label{defocU}
    i v_t + v_{xx} -2|v|^{2}v=0
\end{equation}
supplemented with \textit{zero} boundary conditions. More precisely, in \cite{PDActa} the authors consider non-integrable perturbations of \eqref{defocU} in the form
\begin{equation}
    \label{defocP}
i V_t + V_{xx} -2|V|^{2}V-\epsilon |V|^p V=0, \quad p>2, \ \epsilon>0,
\end{equation}
which correspond to $\mu=1$,  $\nu=-2$ in \eqref{NLS} and preserve the defocusing nature of \eqref{defocU}.
The main result of \cite{PDActa} can be outlined as follows: for sufficiently smooth   initial data that decay at an appropriate rate (described by a suitably defined weighted Sobolev space) and small $\epsilon>0$, the solutions of \eqref{defocP}  approach  those of \eqref{defocU} as $t\rightarrow\infty$ in the sense of the estimate 
\begin{equation}
    \label{rPDActa}
\|v(t)-V(t)\|_{L^{\infty}(\mathbb{R})}=\mathcal{O}\left(\frac{1}{t^{1/2+\kappa}}\right), \quad \kappa>0. 
\end{equation}
This result is proved by combining the inverse scattering transform method (and, in particular, by studying the behavior of the reflection coefficient that emerges in the integrable case \eqref{defocU}) with  detailed estimates in appropriate Sobolev norms derived for the perturbed (non-integrable) model~\eqref{defocP}.  

The model \eqref{defocP} considered in \cite{PDActa} falls under the general perturbative framework   \eqref{gev2}. 
On the other hand, as noted earlier, the results presented herein are of different nature, since now the integrable model \eqref{NLS} is compared against  its \textit{distinct} non-integrable counterparts \eqref{noninNLS} which, unlike \eqref{defocP}, \textit{cannot} be treated as weak perturbations. 
Another major difference between \cite{PDActa} and the present work stems from the fact that, for zero boundary conditions, while the solutions of the defocusing models \eqref{defocU} and \eqref{defocP} are known to exist \textit{globally} in time for all initial data and to decay for a certain range of nonlinearity exponents \cite{reg1} even in higher than one spatial dimensions \cite{reg2,reg3}, the solutions of the non-integrable focusing NLS equation \eqref{noninNLS} are globally defined only for small initial data (see Theorem \ref{zbc-wp-t} below). 
Thus, our proximity estimates between the solutions of two essentially different systems, namely \eqref{NLS} and \eqref{noninNLS}, concern in general finite time intervals in the spirit of the continuous dependence of solutions on the associated small initial data, and do not explore the long-time asymptotic behavior considered in \cite{PDActa}.
In addition, here we also consider the important case of \textit{nonzero} boundary conditions at infinity, which was not investigated in \cite{PDActa}.

For the proof of Theorem \ref{Theorem:stability}, we employ the Fourier transform for the evolution equation satisfied by the difference of solutions $\Delta(t)=u(t)-U(t)$. We also remark on an alternative approach through energy estimates and interpolation inequalities. This second approach is applicable in all the cases of boundary conditions, albeit with  distinct implications for each specific case. For the case of zero boundary conditions, we take advantage of the global existence results for the non-integrable NLS equation in the case of small initial data and also of the regularity of solutions when this data belong to a suitable class. 

The case of  nonzero boundary conditions at infinity turns out to be more challenging since, even for small initial data, it is not known in general whether the non-integrable focusing NLS equation is globally well-posed \cite{Def4b,Def4}. Indeed, global existence is only guaranteed in the defocusing case~\cite{Def1,Def2,Def3,Defoc}. 
On the real line, we prove  local existence in $H^1(\mathbb{R})$ and then use a continuity argument for small data in the sense of \eqref{eq:distance01} in order to establish a closeness estimate between the solutions of the integrable and non-integrable systems, at least for a short time period. 

Furthermore, motivated by the numerical results of \cite{blmt2018}, which indicate the universality of the modulational instability dynamics beyond the integrability barrier, we analytically study the finite interval problem for equations \eqref{NLS} and \eqref{noninNLS} supplemented with the relevant nonzero Dirichlet boundary conditions. This finite domain problem is different from the problem on the real line in that, although both problems possess a conservation law for the $L^2$ norm that involves the amplitude $q_0$ of the wave background, global existence at the $L^2$ level can be deduced from that law only in the case of the finite domain. 
Hence, up to the critical nonlinearity $p=2$, we prove (see Theorem \ref{gex} below) that the solutions to the non-integrable NLS equation \eqref{noninNLS} on the finite interval $(-L, L)$ exist globally in time for appropriate smallness conditions on $L$  and the initial data (in the subcritical case $1\leq p < 2$, the latter condition is not necessary).  It should be noted that the upper bound on $L$ tends to infinity as the amplitude $q_0$ of the nonzero background tends to zero, i.e when the problem approaches the one with zero boundary conditions. Moreover, the proximity estimates of Theorem~\ref{Theorem:stability} under the conditions \eqref{eq:distance0} and for $q_0=\mathcal{O}(\varepsilon)$ are valid for $L=\mathcal{O}\left(1/\varepsilon\right)$. 
Therefore,  the accuracy of the closeness estimate between the solutions of the integrable and non-integrable NLS equations improves when the relevant norms are evaluated over the interval $(-1/\ve, 1/\ve)$ around the core of the respective modulational instability pattern. This fact is further discussed in Section \ref{finite-s}, where it is also illustrated numerically by the simulations of Figures \ref{fig7} and \ref{fig9}.

Finally, we comment on the case of   periodic boundary conditions, which can also be used for approximating the problem on the real line supplemented with zero or nonzero boundary conditions. The proofs in the periodic case are similar to the ones in the finite interval case. Note, in particular, that the assumptions for global existence of solutions to the non-integrable model are also similar, as the main conservation laws for the energy and power are the standard ones. Regarding numerical simulations, however, if the periodic problem is used in order to approximate the infinite line problem with nonzero boundary conditions by taking the parameter $L$ to be large, then a complication arises: the periodic boundary conditions define problems with finite energies in Sobolev spaces of periodic functions~\cite{Def4c}, which is not the case for the problem on the real line. Thus, finite domain approximations may not capture essential effects and implications associated with the infinite domain problem. This issue is especially highlighted by  the results of the present paper through the contrast between the local existence for the problem on the real line and the global existence for the finite domain approximation with  nonzero Dirichlet boundary conditions.
\\[3mm]
\noindent
\textbf{Structure of the paper.} In  Section \ref{zbc-s}, we prove Theorem \ref{Theorem:stability} in the case of vanishing boundary conditions, namely we establish Theorem \ref{TH1}. In addition, we present numerical studies for the concrete examples of bright solitons and soliton collisions.  In Section \ref{nzbc-s}, we establish the version of Theorem~\ref{Theorem:stability} associated with the case of non-vanishing boundary conditions emerging from a constant nonzero background, namely Theorem \ref{nzbc-t}. A key result, which is used for the derivation of the proximity estimates, is the proof of Theorem \ref{nzbc-lwp-t} for the local existence of solutions to the non-integrable NLS models in the class  $C([0,T],H^1(\mathbb{R}))$. In Section \ref{finite-s}, we turn out attention to the case of a finite domain. First, we prove the version of Theorem \ref{Theorem:stability} for the problem on a finite interval with nonzero Dirichlet boundary conditions, namely Theorem \ref{gex2}, then we establish the corresponding result (Theorem \ref{THpbcfinite}) for periodic boundary conditions, and finally we conclude with a numerical study simulating the problem with  nonzero boundary conditions for a variety of localized initial data on the top of a finite background.

%
%
%
\section{Zero boundary conditions at infinity}
\label{zbc-s}
In this section, we establish the version of Theorem \ref{Theorem:stability} that corresponds to the scenario in which both the integrable NLS equation \eqref{NLS} and its non-integrable counterpart \eqref{noninNLS} satisfy zero boundary conditions at infinity, namely
\begin{equation}\label{vbc}
	\lim_{|x|\rightarrow\infty}u(x,t)=\lim_{|x|\rightarrow\infty}U(x,t)=0.
\end{equation}
We begin by recalling the following well-known global existence and regularity results for the general non-integrable NLS equation \eqref{noninNLS}.
\begin{theorem}[Existing results on well-posedness]
	\label{zbc-wp-t}
Consider the Cauchy problem on the line for the non-integrable NLS equation \eqref{noninNLS} with zero conditions at infinity.
\begin{enumerate}[label=\textnormal{(\roman*)}, leftmargin=8mm, topsep=2mm, itemsep=1mm]
\item In the focusing case, if $1 \leq p<2$ then for any initial datum $U_0\in H^1(\mathbb R)$ there exists a global in time solution $U(x, t)$ which is uniformly bounded in $H^1(\mathbb R)$, i.e. there is a constant $M>0$ independent of $t$ such that 
\begin{equation}\label{boundg}
\sup_{t\geq 0}\left\| U(t) \right\|_{H^1(\mathbb R)}\leq M\left\| U_0 \right\|_{H^1(\mathbb R)}.
\end{equation}
Furthermore, if $U_0\in H^m(\mathbb R)$ with $m>1$  then the solution exists  globally in $H^m(\mathbb R)$ and is uniformly bounded in $H^m(\mathbb R)$.
\item In the focusing case, if $p\geq 2$ then there exists a constant $\delta>0$ such that if $\left\| U_0 \right\|_{H^1(\mathbb R)}\leq \delta$ then the solution $U(x, t)$ exists globally in time and is uniformly bounded in $H^1(\mathbb R)$, satisfying the estimate~\eqref{boundg} for some constant $M>0$.  Furthermore, for initial data $U_0\in H^m(\mathbb R)$ with $m>1$ the solution possesses the same regularity  as in case \textnormal{(i)}.
\end{enumerate}
\end{theorem}
For the proofs of the results stated in Theorem \ref{zbc-wp-t}, we refer the reader to the monographs \cite{book0,book1,book2}. Specifically for the further regularity properties of solutions, we also refer to \cite{reg1,reg2}. 
We now proceed to the main result of this section, which provides the counterpart of Theorem \ref{Theorem:stability} in the case of zero boundary conditions at infinity. 
\begin{theorem}[Theorem \ref{Theorem:stability} for zero boundary conditions at infinity]\label{TH1}
Let $p\geq 1$ and suppose that the integrable and non-integrable focusing NLS equations \eqref{NLS} and \eqref{noninNLS} are supplemented with the initial conditions $u(x, 0)=u_0(x)$ and $U(x, 0)=U_0(x)$, respectively.
\begin{enumerate}[label=\textnormal{(\roman*)}, leftmargin=6.5mm, topsep=2mm, itemsep=1mm]
\item \underline{$L^2$ closeness}: 
Given $0<\varepsilon<1$, suppose that the initial data satisfy
\begin{gather}
\label{d0}
\left\| u_0-U_0 \right\|_{L^2(\mathbb R)} \leq  C \varepsilon^3,
\\
\label{d1}
\left\| u_0 \right\|_{H^1(\mathbb R)} \leq c_0 \, \varepsilon,
\ \left\| U_0 \right\|_{H^1(\mathbb R)}\leq  C_0 \, \varepsilon,
\end{gather}
for some constants $c_0, C_0, C>0$. Then, for arbitrary finite time $0<T_f<\infty$, there exists a constant $\widetilde C=\widetilde C(\mu,\gamma, c_0, C_0, C, T_f)$ such that the corresponding solutions $u(x, t)$ and $U(x, t)$ satisfy the estimate 	
\begin{equation}\label{bound1}	
\sup_{t\in [0, T_f]}\left\| u(t)-U(t) \right\|_{L^2(\mathbb R)}\le \widetilde C \varepsilon^3.
\end{equation}
\item\underline{$H^1$ and $L^{\infty}$ closeness}: If the initial data $u_0, U_0$ satisfy \eqref{d1} along with the stronger condition (in place of \eqref{d0})  
\begin{equation}\label{dH1}
\left\| u_0-U_0 \right\|_{H^1(\mathbb R)} \le C_1 \varepsilon^3
\end{equation}
for some constant $C_1>0$, then there exists a constant $\widetilde C_1$ depending on $C_1$ and with a similar dependency on $T_f$ and $\mu, \gamma, c_0, C_0$ as the constant $\widetilde C$ in \eqref{bound1}  such that 
\begin{equation}\label{bound2}
\sup_{t\in [0, T_f]}\left\| u(t)-U(t) \right\|_{H^{1}(\mathbb R)}\le \widetilde C_1 \varepsilon^3.
\end{equation}
Consequently, there exists a constant  $\widetilde C_2$ with similar dependencies as $\widetilde C_1$ such that
\begin{equation}\label{boundB}
\sup_{t\in [0, T_f]}\left\| u(t)-U(t) \right\|_{L^{\infty}(\mathbb R)}\le \widetilde C_2 \varepsilon^3.
\end{equation}
\end{enumerate}
\end{theorem}
 \begin{proof}
The $L^{\infty}$ estimate \eqref{boundB} is a direct consequence of the $H^{1}$ estimate \eqref{bound2} via the Sobolev embedding theorem. The $L^2$ estimate  \eqref{bound1} and the $H^1$ estimate \eqref{bound2} are proved in a  similar way, by forming the equation
\begin{equation}\label{eqd1}
i\Delta_t+\nu\Delta_{xx}= -\mu|u|^2u + \gamma F(|U|^2)U =: N(x,t)
\end{equation}
satisfied by the difference
\begin{equation}
\Delta(x,t) := u(x,t)-U(x,t)
\end{equation}
of solutions to the integrable and non-integrable NLS equations \eqref{NLS} and~\eqref{noninNLS},
and employing the Fourier transform pair given for any $f\in L^2(\mathbb R)$ by
\begin{equation}\label{ft-def}
\begin{aligned}
\what f(\xi) &= \int_{\mathbb R} e^{-i\xi x} f(x) dx, \quad \xi \in \mathbb R,
\\
f(x) &= \frac{1}{2\pi} \int_{\mathbb R} e^{i\xi x} \what f(\xi) d\xi, \quad x \in \mathbb R.
\end{aligned}
\end{equation}

In particular, taking the Fourier transform of equation \eqref{eqd1} and integrating in $t$, we obtain 
\begin{equation}\label{eq:FT}
\widehat{\Delta}(\xi,t) = e^{-i\xi^2t} \widehat{\Delta}(\xi,0)-i\int_{0}^t e^{-i\xi^2(t-\tau)} \widehat{N}(\xi,\tau) d\tau.
\end{equation}
Starting from this expression and employing successively Plancherel's theorem, Minkowski's integral inequality, and the fact that $e^{-i\xi^2 t}$ is unitary, we find%
\begin{equation}\label{D-planch}
\begin{aligned}
\left\| \Delta(t) \right\|_{L^2(\mathbb R)}
&\leq
\frac{1}{\sqrt{2\pi}} \, \big\| e^{-i\xi^2t}\widehat{\Delta}(\xi,0)\big\|_{L^2(\mathbb R)}
+
\frac{1}{\sqrt{2\pi}} \int_{0}^t \big\| e^{-i\xi^2(t-\tau)} \widehat{N}(\xi,\tau) \big\|_{L^2(\mathbb R)} d\tau
\\
&=
\left\| \Delta(0)\right\|_{L^2(\mathbb R)}
+
\int_{0}^t \left\| N(\tau)\right\|_{L^2(\mathbb R)} d\tau.
\end{aligned}
\end{equation}
Hence, we need to estimate the $L^2$ norm of the nonlinearity $N$. By the  inequality $|a+b|^2 \leq 2|a|^2 + 2|b|^2$ and the first of the assumptions \eqref{F-prop}, we have
\begin{align}
\left\|N(t)\right\|_{L^2(\mathbb R)}^2
&\leq
2\mu^2 \int_{\mathbb R} \left|u(x, t)\right|^6 dx
+
2 \gamma^2 \int_{\mathbb R} \left|F(|U(x,t)|^2)\right|^2 \left|U(x,t)\right|^2 dx
\nn\\
&\leq
2\mu^2 \left\| u(t) \right\|_{L^6(\mathbb R)}^6 + 2\gamma^2 K^2 \left\| U(t) \right\|_{L^{2(2p+1)}(\mathbb R)}^{2(2p+1)}.
\end{align}
Thus, invoking
the Sobolev embedding (see Theorem 4.12, last part of Case A on page 85 of \cite{af2003})
\begin{equation}\label{eqd3}
H^1(\mathbb R)\subset L^{q}(\mathbb R)
\ \Rightarrow \
\left\|f\right\|_{L^q(\mathbb R)} \leq c \left\|f\right\|_{H^1(\mathbb R)},
\quad
2 \leq q\leq\infty,
\end{equation}
for $q=6$ and also $q=2(2p+1)$, we obtain
\begin{equation}\label{F-est-0}
\left\|N(t)\right\|_{L^2(\mathbb R)}^2
\leq
2\mu^2 \left\| u(t) \right\|_{H^1(\mathbb R)}^6 + 2 \gamma^2 K^2 \left\| U(t) \right\|_{H^1(\mathbb R)}^{2(2p+1)}.
\end{equation}
In turn, since $\sqrt{a+b} \leq  \sqrt a + \sqrt b$, we have
\begin{align}
\int_0^t  \left\|N(\tau)\right\|_{L^2(\mathbb R)}  d\tau
&\leq
 \sqrt 2 \, |\mu| \int_0^t \left\| u(\tau) \right\|_{H^1(\mathbb R)}^3 d\tau + \sqrt 2 \, |\gamma| K \int_0^t \left\| U(\tau) \right\|_{H^1(\mathbb R)}^{2p+1} d\tau
\nn\\
&\leq
A \sup_{\tau \in [0, t]} \left(\left\| u(\tau) \right\|_{H^1(\mathbb R)}^3 +  \left\| U(\tau) \right\|_{H^1(\mathbb R)}^{2p+1} \right) t
\label{F-est}
\end{align}
where $A = \sqrt 2\max\{|\mu|, |\gamma| K\}$.
In view of the estimate \eqref{F-est}, inequality \eqref{D-planch} yields
\begin{equation}
\left\| \Delta(t) \right\|_{L^2(\mathbb R)}
\leq
\left\| \Delta(0) \right\|_{L^2(\mathbb R)} 
+
A \sup_{\tau \in [0, t]} \left(\left\| u(\tau) \right\|_{H^1(\mathbb R)}^3 +  \left\| U(\tau) \right\|_{H^1(\mathbb R)}^{2p+1} \right) t
\end{equation}
which along with the solution bound \eqref{boundg} and the conditions \eqref{d1} and  \eqref{d0} gives rise to the desired estimate \eqref{bound1} with constant $\widetilde C = C + A \, \big(c_0^3 + C_0^{2p+1}\big) \, T_f$. 

We proceed to the proof of the $H^1$ estimate \eqref{bound2}. Expression \eqref{eq:FT} combined with the inequality $|a+b|^2 \leq 2|a|^2+2|b|^2$ yields
\begin{align}
\left\|\Delta(t)\right\|_{H^1(\mathbb R)}^2
&\leq
2\int_{\xi\in\mathbb R} \left(1+\xi^2\right) \big|e^{-i\xi^2t} \widehat{\Delta}(\xi,0)\big|^2 d\xi
+
2\int_{\xi\in\mathbb R} \left(1+\xi^2\right) \left|\int_{0}^t e^{-i\xi^2(t-\tau)} \widehat{N}(\xi,\tau)  d\tau\right|^2 d\xi
\nn
\end{align}
so using Minkowski's integral inequality in the second integral  we find
\begin{align}
\left\|\Delta(t)\right\|_{H^1(\mathbb R)}^2 
&\leq
2\int_{\xi\in\mathbb R} \left(1+\xi^2\right) \big|\widehat{\Delta}(\xi,0)\big|^2 d\xi
+
2\left(\int_0^t \left(\int_{\xi\in\mathbb R} \left(1+\xi^2\right) \big|\widehat{N}(\xi,\tau)\big|^2 d\xi\right)^{\frac 12} d\tau \right)^2
\nonumber\\
&=
2 \left\|\Delta(0)\right\|_{H^1(\mathbb R)}^2 + 2 \left(\int_0^t \left\|N(\tau)\right\|_{H^1(\mathbb R)} d\tau\right)^2.
\label{D-sup}
\end{align}
Note that $\left\|N\right\|_{H^1(\mathbb R)}=\left\|N\right\|_{L^2(\mathbb R)}+\left\|\p_x N\right\|_{L^2(\mathbb R)}$. The $L^2$ norm of $N$ has already been estimated by~\eqref{F-est}. Differentiating the right-hand side of \eqref{eqd1}, we have
\begin{equation*}
\begin{aligned}
\left\|\p_xN(t)\right\|_{L^2(\mathbb R)}^2
&\leq
2\mu^2 \int_{\mathbb R} \left|2|u|^2u_x + u^2 \bar u_x\right|^2 dx
\\
&\quad
+
2\gamma^2 \int_{\mathbb R} \Big|F'(|U|^2) |U|^2U_x + F'(|U|^2) U^2 \bar U_x + F(|U|^2) U_x\Big|^2 dx
\\
&\leq
18\mu^2 \int_{\mathbb R}  |u|^4 |u_x|^2 dx
+ 
4 \gamma^2 \int_{\mathbb R} 
\left(
4\left| F'(|U|^2) \right|^2 |U|^4 |U_x|^2
+ \left| F(|U|^2) \right|^2  |U_x|^2
\right) dx.
\end{aligned}
\end{equation*}
Then, due to the first and third assumption in \eqref{F-prop}, 
\begin{align}
\left\|\p_xN(t)\right\|_{L^2(\mathbb R)}^2
&\leq
18 \mu^2 \int_{\mathbb R}  |u|^4 |u_x|^2  dx
+ 20 \gamma^2 K^2 \int_{\mathbb R}  |U|^{4p}  |U_x|^2 dx
\nn\\
&\leq
18 \mu^2 \left\| u(t) \right\|_{L^\infty(\mathbb R)}^4 \left\| u_x(t) \right\|_{L^2(\mathbb R)}^2
+ 20 \gamma^2 K^2 \left\| U(t) \right\|_{L^\infty(\mathbb R)}^{4p} \left\| U_x(t) \right\|_{L^2(\mathbb R)}^2
\nn
\end{align}
and so, by the Sobolev embedding theorem,
\begin{equation}
\left\|\p_xN(t)\right\|_{L^2(\mathbb R)}^2
\leq
18 \mu^2 \left\| u(t) \right\|_{H^1(\mathbb R)}^6 + 20 \gamma^2 K^2 \left\| U(t) \right\|_{H^1(\mathbb R)}^{2(2p+1)}.
\end{equation}
In turn, since $\sqrt{a+b} \leq \sqrt a + \sqrt b$, we obtain
\begin{equation}
\int_0^t  \left\|\p_xN(\tau)\right\|_{L^2(\mathbb R)}  d\tau
\leq
A' \sup_{\tau \in [0, t]} \left( \left\| u(\tau) \right\|_{H^1(\mathbb R)}^3 + \left\| U(\tau) \right\|_{H^1(\mathbb R)}^{2p+1}\right) t
\label{F'-est}
\end{equation}
with $A' = \max\left\{3\sqrt 2 |\mu|,  2\sqrt 5 |\gamma| K\right\}$. 
Overall, estimates \eqref{F-est} and \eqref{F'-est} combined with inequality \eqref{D-sup} imply (note that $A \leq A'$)
\begin{equation}\label{D-sup-2}
\left\| \Delta(t) \right\|_{H^1(\mathbb R)} 
\leq
\sqrt 2 \left\|\Delta(0)\right\|_{H^1(\mathbb R)}
+
\sqrt 2 A' \sup_{\tau \in [0, t]} \left( \left\| u(\tau) \right\|_{H^1(\mathbb R)}^3 + \left\| U(\tau) \right\|_{H^1(\mathbb R)}^{2p+1}\right) t.
\end{equation}
Hence, thanks to the solution estimate \eqref{boundg} and the assumptions \eqref{d1} and \eqref{dH1} on the initial data and their $H^1$ distance, we infer
\begin{equation}\label{D-H1}
\begin{aligned}
\left\| \Delta(t) \right\|_{H^1(\mathbb R)} 
&\leq 
\sqrt 2 \, C_1 \varepsilon^3
+
\sqrt 2 A' \left(M^3 \left\| u_0 \right\|_{H^1(\mathbb R)}^3 + M_1^{2p+1} \left\| U_0 \right\|_{H^1(\mathbb R)}^{2p+1}\right) t
\\
&\leq 
\sqrt 2 \, C_1 \varepsilon^3
+
\sqrt 2 A' \left(M^3 c_0^3 \varepsilon^3 + M_1^{2p+1} C_0^{2p+1} \varepsilon^{2p+1}\right) t.
\end{aligned}
\end{equation}
Then, noting that $\varepsilon^{2p+1} \leq \varepsilon^3$ as $\varepsilon<1$ and $2p+1 \geq 3$ (recall that $p\geq 1$), we arrive at the $H^1$ estimate~\eqref{bound2} with constant $\widetilde C_1 = \sqrt 2 \, C_1 + \sqrt 2 \, A' \max\big\{M^3 C_{\mu,0}^3, M_1^{2p+1} C_{\gamma,0}^{2p+1}\big\} T_f$. 

The proof of Theorem \ref{TH1} is complete.
\end{proof}
\vskip 3mm
\noindent
\textbf{Remarks on Theorem \ref{TH1}.}
\\[2mm]
\noindent
\textit{Persistence of the analytical localized solutions of the integrable NLS in the non-integrable setting}. An important implication of Theorem \ref{TH1} is that it rigorously justifies that (at least) small-amplitude localized structures satisfying the integrable NLS equation along with the zero boundary conditions~\eqref{vbc}  may persist in the non-integrable setting of equation \eqref{noninNLS} for significant times. In particular, the non-integrable  NLS equation \eqref{noninNLS} admits small-amplitude solutions of $\mathcal O(\varepsilon)$ that remain $\mathcal O(\varepsilon^3)$ close to the analytical solutions of the integrable NLS \eqref{NLS} in the 
 $H^1$ and $L^{\infty}$ norms.  
 In this regard, in the case where $u_0\equiv U_0$ (i.e. $C=0$ in \eqref{d0}),  the analytical estimates of Theorem \ref{TH1} show that the distance of solutions grows at most linearly for any $t\in (0,\infty)$, since
\begin{equation}\label{gr1}
\left\|\Delta(t)\right\|_{\mathcal{X}} \leq M \varepsilon^3 t,
\quad
\mathcal{X}=H^1(\mathbb{R}) \text{ or } L^{\infty}(\mathbb{R}),	
\end{equation}
where $M$ is one of the constants $\widetilde C_1, \widetilde C_2$.  For example, for times $t\sim \mathcal{O}\left(1/\varepsilon^2\right)$ the distance function $\left\|\Delta(t)\right\|_{\mathcal{X}}\sim \mathcal{O}(\varepsilon)$. 

The time growth of the bound for the distance function in estimate \eqref{gr1} can be juxtaposed against the time growth in the corresponding estimates for the distance function between the solutions of the complex Ginzburg-Landau and the NLS equations when the inviscid limit of the former is considered~\cite{JWU}. These latter estimates can grow at an even exponential rate \cite{OG}, which has to be distinguished from the linear growth result of \eqref{gr1} proved here.
\\[2mm]
\noindent
\textit{Alternative proof of Theorem \ref{TH1} via energy estimates and interpolation}.  An alternative proof of Theorem \ref{TH1} can be provided through an energy argument combined with interpolation estimates by usingthe one-dimensional Gagliardo-Nirenberg inequality and the regularity of the initial data. This method, however, results in weaker stability estimates under stronger conditions on the initial data. 

Indeed, via the energy method, \eqref{eqd1} yields the differential inequality
\begin{equation}
	\label{eqd8}
\frac{d}{dt}\left\| \Delta(t) \right\|_{L^2(\mathbb R)}\leq M_2 \left( \left\| u_0 \right\|^3_{H^1(\mathbb R)}+ \left\| U_0 \right\|_{H^1(\mathbb R)}^{2p+1}\right)	
\end{equation}
for some constant $M_2=M_2(M, M_1, c, \gamma, \mu)$. Integrating  \eqref{eqd8} for any $t$ in the arbitrary interval $[0, T_f]$ and using the assumptions \eqref{d1} and \eqref{d0} on the initial data and their $L^2$ distance along with the fact that $0<\varepsilon<1$, we obtain
estimate \eqref{bound1} with constant $\widetilde C = C + M_2 \, \big(c_0^3+ C_0^{2p+1}\big) \, T_f$. 
Moreover, the $H^1$ estimate for $\Delta$ can be derived via interpolation by using the Gagliardo-Nirenberg inequality
(e.g. see Theorem 1.3.4 in \cite{ch1998}), namely
\begin{equation}\label{GN}
\left\| \partial^j f \right\|_{L^p(\mathbb R)} 
\leq 
C_{\text{GN}} \left\| f \right\|^{1-\theta}_{L^q(\mathbb R)}
\left\| \partial^m f \right\|_{L^r(\mathbb R)}^{\theta},
\quad j, m \in \mathbb N, \ 0\leq j < m, \ \frac{j}{m} \leq \theta \leq 1,
\end{equation}
where $1 \leq q, r \leq \infty$, $p$ is given by $\frac{1}{p}=j+\theta\left(\frac{1}{r}-m\right)+\frac{1-\theta}{q}$, and $C_{\text{GN}} = C_{\text{GN}}(q, r, j, m, \theta)$.
However, this step requires sufficient regularity of the initial data. In particular, assume that $u_0,U_0\in H^m(\mathbb R)$ with $m>1$.  Then, according to the global existence results of Theorem \ref{zbc-wp-t}, the solutions of the integrable and non-integrable NLS equations \eqref{NLS} and  \eqref{noninNLS} satisfy uniform in time estimates, 
\begin{equation*}
	\sup_{t\geq 0}\left\| u(t) \right\|_{H^m(\mathbb R)}\leq R,\quad 	\sup_{t\geq 0}\left\| U(t) \right\|_{H^m(\mathbb R)}\leq R,
\end{equation*}
for some general constant $R$ which depends only on the norm of the initial data $u_0$ and $U_0$ but is independent of $t\geq 0$.
Hence, by the triangle inequality, the distance $\Delta$ also admits such a uniform bound as
\begin{equation}\label{eqd10}
\sup_{t\geq 0}\left\|\Delta(t)\right\|_{H^m(\mathbb R)}
\leq 
\sup_{t\geq 0}\left\| u(t) \right\|_{H^m(\mathbb R)}+\sup_{t\geq 0}\left\| U(t) \right\|_{H^m(\mathbb R)}\leq 2R.
\end{equation}
Employing the Gagliardo-Nirenberg inequality \eqref{GN} for $f=\Delta$, $j=1$, $p=q=r=2$ and any $m>1$ (these choices imply $\theta = \frac 1m$), we have
\begin{equation}\label{eqd11}
\left\|\p_x \Delta(t)\right\|_{L^2(\mathbb R)} \leq C_{\text{GN}} \left\|\Delta(t)\right\|_{L^2(\mathbb R)}^{\frac{m-1}{m}} \left\|\Delta(t)\right\|_{H^m(\mathbb R)}^{\frac{1}{m}}, \quad t\geq 0.
\end{equation}
The right-hand side of \eqref{eqd11} can be estimated via the $L^2$ closeness estimate \eqref{bound1} and the uniform bound \eqref{eqd10}. In particular, there exists a constant $c_2 =  c_2(C_{\text{GN}}, \widetilde C, R, m)$ such that
\begin{equation}\label{eqd11a}
\left\|\p_x \Delta(t)\right\|_{L^2(\mathbb R)} \leq c_2 \varepsilon^\sigma t^{\frac{m-1}{m}}, \quad \sigma = \frac{3(m-1)}{m}, \quad t\geq 0.
\end{equation}
Observe that, although the time growth rate $\frac{m-1}{m}$ is sublinear (as opposed to the linear growth of the bounds in Theorem \ref{TH1}), the proximity exponent $\sigma$ is smaller as it belongs to $(0, 3)$. In the limit of infinitely smooth initial data $m\rightarrow \infty$, which is a much stronger assumption than the one of Theorem~\ref{TH1}, we recover \eqref{gr1} after adding \eqref{eqd11a} to the $L^2$ closeness estimate \eqref{bound1}. 
That is, the energy method yields the cubic proximity exponent of the Fourier transform method of Theorem~\ref{TH1} only in the special case of infinitely smooth initial data.
\\[3mm]
\noindent
\textbf{Robustness of stable solutions.}
An interesting application of Theorem \ref{TH1} (and of the general statement in Theorem \ref{Theorem:stability}) is related to the robustness of stable solutions of one system in the dynamics of the other. As an illustrative example, we consider the case of standing wave solutions to the general focusing semilinear Schr\"odinger equation \eqref{NLSP}, namely solutions of the form
\begin{equation}
\label{sw1}
 U(x,t) = e^{i\omega t} \, W(x), \quad \omega>0,
\end{equation}
with $W$ satisfying the stationary equation
\begin{equation}
\label{sw2}
 -W''+\omega W+\mu|W|^{2p}W=0, \quad \mu > 0.
\end{equation}
Associated with the solutions of \eqref{sw2} is the functional
\begin{equation}
\label{Fun1}
S_p(W)=\frac{1}{2}\int_{\mathbb{R}}|W'|^2dx+\frac{\omega}{2}\int_{\mathbb{R}}|W|^2dx-\frac{\mu}{2p+2}\int_{\mathbb{R}}|W|^{2p+2}dx,
\end{equation}
and the set
\begin{equation}
\mathcal{A}_p=\left\{W\in H^{1}(\mathbb{R}): W\neq 0 \text{ and }
-W''+\omega W+\mu|W|^{2p}W=0\right\}.
\end{equation}
The orbital stability of the solutions \eqref{sw1} is discussed in \cite{book1}. In particular, the stability result reads as follows.
\begin{theorem}[Theorem 8.3.1 in \cite{book1}]
\label{cazstab}
 Let $p<2$. If $W\in\mathcal{A}_p$, then \eqref{sw1} is a stable solution of equation \eqref{NLSP} in the following sense. For every $\epsilon>0$, there exists a $\delta(\epsilon)$ such that if $U_0\in H^{1}(\mathbb{R})$ satisfies $\|U_0-W\|_{H^1(\mathbb{R})}\leq \delta (\epsilon)$ then the  maximal solution $U(x,t)$ of~\eqref{NLSP}  associated with $U_0(x)$ satisfies 
 \begin{equation}
 \label{stabre}
 \sup_{t>0 }\inf_{\theta\in\mathbb{R}}\inf_{y\in\mathbb{R}}
 \left\|U(\cdot,t) - e^{i\theta} W(\cdot-y) \right\|_{H^1(\mathbb{R})}\leq \epsilon.
 \end{equation}
 In other words, if $U_0$ is close to $W$ in $H^1(\mathbb{R})$, then the corresponding solution $U$ remains close to the orbit of $W$, up to space translations and rotations.
\end{theorem}

It is crucial to recall that, in the subcritical case $0<p<2$,  any  standing wave solution can be mapped via the transformation
\begin{equation}\label{scale1}
U(x,t)\mapsto U_{\varrho}(x,t):=\varrho^{\frac{1}{p}} U(\varrho x, \varrho^2t), \quad \varrho>0,
\end{equation}
to one with an $L^2$ norm of arbitrary size \cite{Wein99}, as 
\begin{equation}\label{scale2}
\left\|U_{\varrho}(t)\right\|_{L^2(\mathbb{R})}=\varrho^{\frac{2-p}{2p}} \left\|U(\varrho^2 t)\right\|_{L^2(\mathbb{R})}
\end{equation}
and so for $p<2$ one may construct standing waves of arbitrary $L^2$ norm for a suitable choice of $\varrho$.  Moreover, for $X:=\varrho x$ we have $\partial_x U_{\varrho}(x,t)=\varrho^{\frac{p+1}{p}}\partial_X U(X, \varrho^2t)$ and so  
\begin{equation}\label{scale3}
\left\|\partial_x U_{\varrho}(t)\right\|_{L^2(\mathbb{R})}=\varrho^{\frac{p+2}{2p}} \left\|\partial_x U(\varrho^2 t)\right\|_{L^2(\mathbb{R})}.
\end{equation} 
Therefore, the scaling \eqref{scale1} preserves the stability of standing waves. Indeed, suppose $W \in \mathcal{A}_p$ corresponds to a stable standing wave. Then, given $\epsilon > 0$ and initial data $U_0$ satisfying the hypothesis of Theorem \ref{cazstab}, for $W_{\varrho}(x) := \varrho^{\frac{1}{p}}W(\varrho x)$ and ${U_0}_{\varrho}(x)$ defined analogously we have, in view of \eqref{scale2} and~\eqref{scale3},
\begin{equation}
\left\|{U_0}_\varrho-W_{\varrho}\right\|_{H^1(\mathbb{R})}
\leq
\max\left\{\varrho^{\frac{2-p}{2p}}, \varrho^{\frac{2+p}{2p}}\right\} \delta
\end{equation}
and 
\begin{equation}
\sup_{t>0 }\inf_{\theta\in\mathbb{R}}\inf_{y\in\mathbb{R}}
 \left\|U_\varrho(\cdot, \varrho^2 t)-e^{i\theta} W_\varrho(\cdot-y)\right\|_{H^1(\mathbb{R})}
\leq 
\max\left\{\varrho^{\frac{2-p}{2p}}, \varrho^{\frac{2+p}{2p}}\right\} \epsilon,
\end{equation}
implying the stability of the standing wave associated with $W_\varrho$. 

With the above preparations, a transitivity argument that combines Theorem \ref{TH1} proved in this work  with Theorem \ref{cazstab} from  \cite{book1} implies the following result.
\begin{corollary}\label{cazstab-c}
Suppose $1<p<2$ and let $U_{\textnormal{NI}}(x,t) = e^{i\omega t} \, W_{\textnormal{NI}}(x)$, $W_{\textnormal{NI}} \in \mathcal{A}_p$, $\omega>0$, be a standing wave solution of the non-integrable NLS equation \eqref{NLSP} which is stable in the sense of Theorem~\ref{cazstab}.
Furthermore, given $0<\ve<1$, let $u_0, U_0 \in H^1(\mathbb R)$ be initial data for the integrable  and non-integrable NLS equations~\eqref{NLS} and \eqref{NLSP}, respectively, satisfying the conditions
\begin{align}
\label{nsw1}
\left\|U_0-W_{\textnormal{NI}}\right\|_{H^1(\mathbb{R})}&\leq \delta (\varepsilon),\\
\label{nsw2}
\left\|u_0\right\|_{H^1(\mathbb{R})}&\leq c_0\varepsilon,\\
\label{nsw3}
\left\|u_0-U_0\right\|_{H^1(\mathbb{R})}&\leq c_1 \varepsilon,
\end{align}
for some constants $c_0, c_1>0$, with $\delta(\varepsilon)$ satisfying the stability criterion of Theorem \ref{cazstab}.
Then, $\left\|u_0-W_{\textnormal{NI}}\right\|_{H^1(\mathbb{R})}\leq  c_1 \ve + \delta(\ve) =: \widehat{\delta}(\varepsilon)$ and for arbitrary $0<T_f<\infty$ there exists a constant $\widehat{K}(T_f)>0$ such that
\begin{equation}
 \label{rob2}
\sup_{t\in[0,T_f]}\inf_{\theta\in\mathbb{R}}\inf_{y\in\mathbb{R}}
\left\|u(\cdot,t) - e^{i\theta} W_{\textnormal{NI}}(\cdot-y) \right\|_{H^1(\mathbb{R})}\leq \widehat{K}\varepsilon.
\end{equation}
\end{corollary}

\begin{proof} 
Let $0<\ve<1$. First, note that the conditions \eqref{nsw1} and \eqref{nsw3} combined with the triangle inequality readily yield
$
\|u_0-W_{\text{NI}}\|_{H^1(\mathbb{R})}\leq 
\|u_0-U_0\|_{H^1(\mathbb{R})}+\|U_0-W_{\text{NI}}\|_{H^1(\mathbb{R})}\leq \widehat{\delta}(\varepsilon)$.
Next, since $W_{\text{NI}}$ and $U_0$ satisfy the hypothesis of Theorem \ref{cazstab}, the solution $U$ of the non-integrable NLS  equation \eqref{NLSP} associated with $U_0$ satisfies the bound
\begin{equation}\label{nsw0}
\sup_{t>0 }\inf_{\theta\in\mathbb{R}}\inf_{y\in\mathbb{R}}
 \left\|U(\cdot,t) - e^{i\theta} W_{\textnormal{NI}}(\cdot-y) \right\|_{H^1(\mathbb{R})}\leq \varepsilon.
 \end{equation} 
Let $u$ denote the solution of the integrable NLS  equation \eqref{NLS} associated with $u_0$. Then, in view of the conditions \eqref{nsw2} and \eqref{nsw3}, which via the triangle inequality and the fact that $\ve<1$ imply $\left\|U_0\right\|_{H^{1}(\mathbb R)}\leq \left(c_0+c_1\right) \varepsilon \leq \left(c_0+c_1\right) \ve^{\frac 13}$,  estimate \eqref{bound2} of Theorem \ref{TH1} yields  
\begin{equation}\label{nsw4}
\left\| u(t)-U(t) \right\|_{H^{1}(\mathbb R)}\le \widetilde C_1 \varepsilon,
\quad 
t\in [0,T_f], \ \widetilde C_1=\widetilde C_1(T_f).
\end{equation}
Hence,   combining the triangle inequality with \eqref{nsw0} and \eqref{nsw4}, we deduce
 \begin{equation*}
 \begin{aligned}
&\quad
\sup_{t\in[0,T_f]}\inf_{\theta\in\mathbb{R}}\inf_{y\in\mathbb{R}} \left\|u(\cdot,t)- e^{i\theta} W_{\text{NI}}(\cdot-y)\right\|_{H^1(\mathbb{R})}
\\
&\leq \sup_{t\in[0,T_f]}\
\left\|u(t)-U(t)\right\|_{H^1(\mathbb{R})}
+
\sup_{t>0}\inf_{\theta\in\mathbb{R}}\inf_{y\in\mathbb{R}} \left\|U(\cdot,t)- e^{i\theta} W_{\text{NI}}(\cdot-y)\right\|_{H^1(\mathbb{R})}
\leq
\widetilde C_1 \varepsilon +\varepsilon,
\end{aligned}
 \end{equation*}
which amounts to the desired estimate \eqref{rob2} with $\widehat{K}=\widetilde C_1+1$.
\end{proof}

\noindent
\textbf{Numerical simulations.}
We conclude this section with a  numerical study illustrating our analytical results in the case of the zero boundary conditions \eqref{vbc}.  As a first example, we consider the case of  bright solitons. The second example concerns the more intricate case of collision of bright solitons.
\\[2mm]
\noindent
\textit{Bright solitons.}
We supplement the non-integrable NLS equation \eqref{noninNLS} with the initial condition emanating from the one-soliton solution of the integrable NLS equation \eqref{NLS} (e.g. see \cite{AKbook}), namely
\begin{equation}\label{inc}
U_0(x)=u_{\text{S}}(x,0) := A\text{sech}(A x) e^{ic_s x}
\end{equation}
with the parameters $A$, $c_s$ chosen  so that the smallness conditions of Theorem~\ref{TH1}  are met.

\begin{figure}[ht!]
	\begin{center}
		\begin{tabular}{cc}
			\hspace{-0.0cm}\includegraphics[width=.42\textwidth]{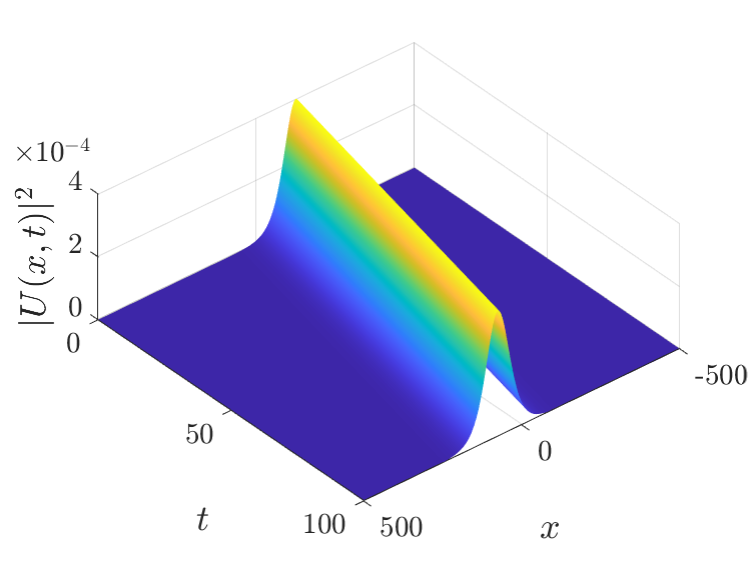}
			\hspace{-0.0cm}\includegraphics[width=.42\textwidth]{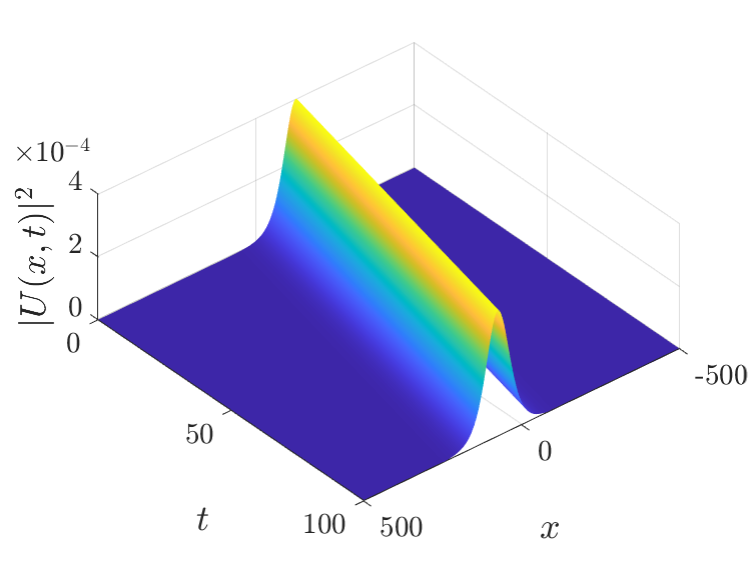}\\
			\hspace{-0.0cm}\includegraphics[width=.42\textwidth]{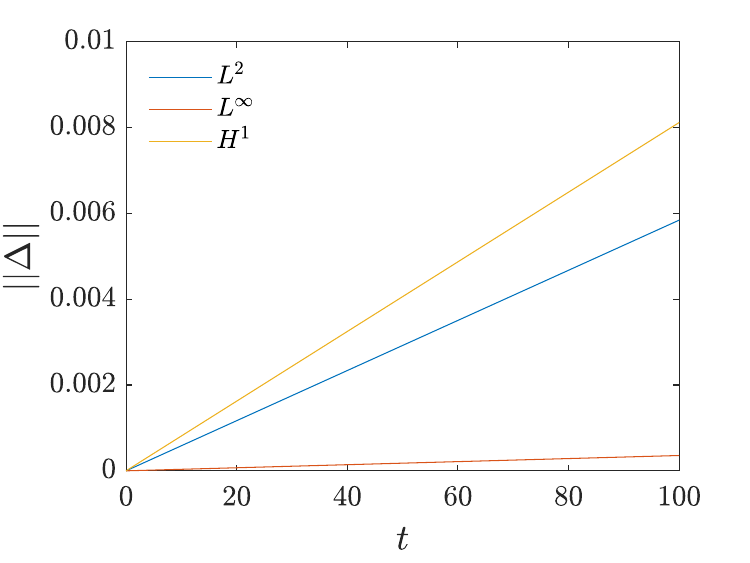}
			\hspace{0.3cm}\includegraphics[width=.42\textwidth]{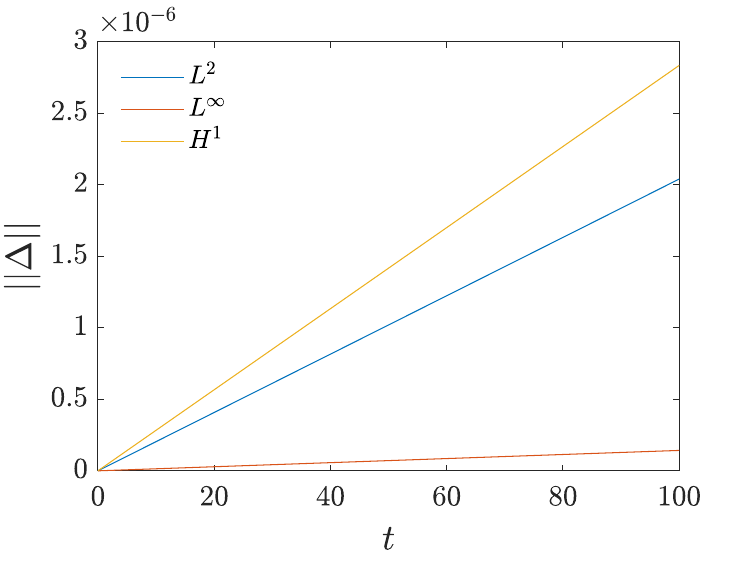}
		\end{tabular}
	\end{center}
	\caption{Top row: Spatiotemporal evolution of the density $|U(x,t)|^2$ of the  non-integrable NLS equation \eqref{noninNLS} with $\nu=1/2$, $\gamma=1$ and supplemented with the soliton initial condition  \eqref{inc}. Left: Power nonlinearity \eqref{NLSP} with $p=3$ (supercritical case). Right: Saturable nonlinearity \eqref{Sat2}. 
Bottom row: The corresponding evolution of the norms $\left\|\Delta(t)\right\|_{\mathcal{X}}$ with $\mathcal X = L^2, H^1 \text{ or } L^\infty$ for each of the two cases.
}
\label{fig1}
\end{figure}

More specifically, we trigger the dynamics of the non-integrable NLS equation \eqref{noninNLS} with $\nu=1/2$, $\gamma=1$ and the initial condition \eqref{inc} with small amplitude $A=0.02$ and velocity $c_s=1$.  We examine the case of a power nonlinearity \eqref{NLSP} with $p=3$ (top left panel in Figure \ref{fig1}) as well as the saturable nonlinearity model \eqref{Sat2} (top right panel in Figure \ref{fig1}).
As explained below, in both cases the evolution is essentially indistinguishable from the one of the integrable NLS equation. 
The bottom panels of Figure~\ref{fig1}  show the evolution of  the distance norms $\left\| \Delta \right\|_{L^2(\mathbb R)}$ (middle (blue) line), $\left\| \Delta \right\|_{H^1(\mathbb R)}$ (upper (yellow) line) and $\left\| \Delta \right\|_{L^{\infty}(\mathbb R)}$ (bottom (red) line). 
The evolution of these norms appears to be in an excellent agreement with the theoretical prediction~\eqref{gr1} for their linear growth with respect to time. 
For the aforementioned set of parameter values, we have $\left\|u_{\text{S}}(0)\right\|_{H^1(\mathbb R)}=\left\| U_0\right\|_{H^1(\mathbb R)} \simeq 0.2$, i.e. we may take  $\varepsilon=0.2$ in which case there is an excellent agreement with the theoretical predictions regarding the growth of the distance in terms of $\varepsilon$. Indeed, for  times $t\sim \mathcal{O}\left(1/\varepsilon^2\right)$ the bound \eqref{gr1} predicts that $\left\|\Delta(t)\right\|_{\mathcal{X}}\sim \mathcal{O}(\varepsilon)$, which is consistent with the evolution of  the distance norms  over the interval $t\in [0,T_f]$  with $T_f=100$ depicted in the bottom row of Figure \ref{fig1}. 

Note that in the case of the saturable nonlinearity the linear growth of $\left\|\Delta(t)\right\|_{\mathcal{X}}$ is incremental, suggesting that the dynamics of the saturable NLS equation may be even closer to those of the integrable one, compared to the power nonlinearity model.  This behavior also insinuates that in the non-integrable cases the  persistence of small-amplitude bright solitons may last for a significantly long time, which is quite remarkable.

\begin{figure}[ht!]
	\begin{center}
		\begin{tabular}{cc}
			\includegraphics[width=.45\textwidth]{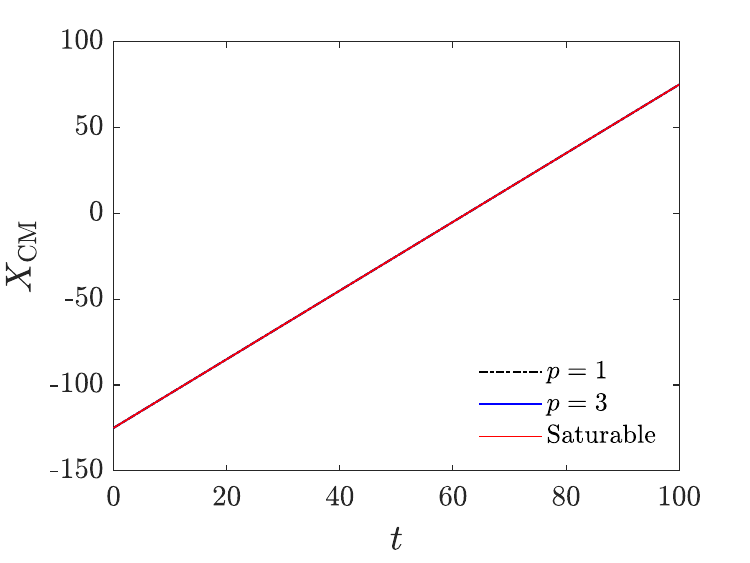}
			\includegraphics[width=.45\textwidth]{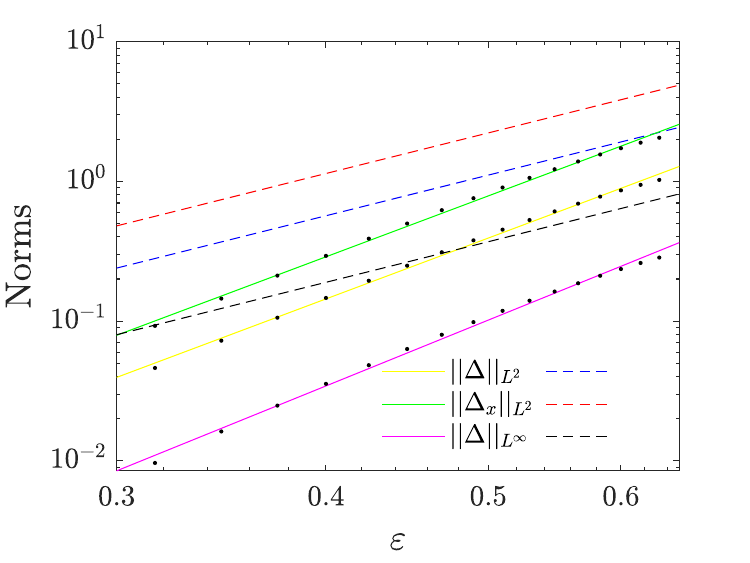}
		\end{tabular}
	\end{center}
	\caption{Left: Spatiotemporal evolution of the soliton's center of mass for the integrable and the non-integrable models.  Right: Logarithmic scale plots of the variation of the distance norm $\left\| \Delta(t) \right\|_{\mathcal X}$ with $\mathcal X = L^2, H^1 \text{ or } L^\infty$  as a function of $\varepsilon$ for fixed $T_f=600$. See text for more details. 
}
	\label{fig2}
\end{figure}

The left panel of Figure \ref{fig2}  depicts the numerical evolution of the center of mass of the soliton, presented by solid lines for the non-integrable models and by a dashed line for  the integrable model (in fact, the two solid lines for the power and saturable nonlinearities are indistinguishable). For this study, we use a spatial translation of the initial condition \eqref{inc}, namely $U_0(x)=A\mathrm{sech}[A(x-x_0)]e^{ic_s(x-x_0)}$   with $x_0=125$ and $A=0.02$, $c_s=1$ as before. The two lines corresponding to the non-integrable models are indistinguishable from the one of the integrable NLS which is extended in the figure also  to $x<0$ for visualization purposes.   The solitary waves of the non-integrable models demonstrate remarkable robustness.

Besides the case of small-amplitude initial data considered above, 
it is also important to investigate the case when the value of the parameter $\varepsilon$ is increased.  The right panel of  Figure \ref{fig2} shows  logarithmic scale plots of the variation of the distance norm $\left\|\Delta(t)\right\|_{\mathcal{X}}$ for $\mathcal X = L^2(\mathbb R),  H^1(\mathbb R) \text{ or } L^\infty(\mathbb R)$ as a function of $\varepsilon$ for fixed $T_f=600$.  The dots along the solid lines correspond to the numerically obtained rates of these variations fitted to the lines of the form $\left\| \Delta(T_f) \right\|_{L^2(\mathbb R)}$ versus $K_1\varepsilon^{\sigma_1}$, $\left\| \Delta(T_f) \right\|_{H^1(\mathbb R)}$ versus $K_2\varepsilon^{\sigma_2}$, and  $\left\| \Delta(T_f) \right\|_{L^{\infty}(\mathbb R)}$ versus $K_3\varepsilon^{\sigma_3}$. The fitting results in the following values: $K_1=8.87$, $K_2=17.77$, $K_3=2.96$ and $\sigma_1=\sigma_2=4.5$, $\sigma_3=4.86$.  
For these values of constants $K_j$, the dashed lines correspond to  the analytical estimates of Theorem \ref{TH1} of the form  $\left\| \Delta(T_f) \right\|_{\mathcal X}$ versus  $K_j\varepsilon^3$. 
The numerical results are consistent with the order of the analytical estimates and, in fact, indicate that the numerical variation of the distance norms may be of significantly lower rate, namely of $\mathcal{O}(\varepsilon^{4.5})$ or $\mathcal{O}(\varepsilon^5)$.
\\[2mm]
\noindent
\textit{Bright soliton collisions.}
For the study of the dynamics of bright soliton collisions, we choose as initial datum two incoming soliton solutions of the cubic NLS equation \eqref{NLS} of the same amplitude $A$, initially separated by a distance $x_0$ and moving against one other at speed $c_s$, namely
\begin{equation}
\label{BSCin}
U_0(x) = A \text{sech}\left[A(x-x_0)\right] e^{-ic_s(x-x_0)} + A \text{sech}\left[A(x+x_0)\right] e^{ic_s(x+x_0)}.
\end{equation}
We take $A=0.1$ and $x_0=50$. For the velocities, we considered both fast solitons with $c_s=2$ and very slow solitons with $c_s=0.02$. 

\begin{figure}[ht!]
	\begin{center}
		\begin{tabular}{cc}
\hspace{-0cm}\includegraphics[width=.45\textwidth]{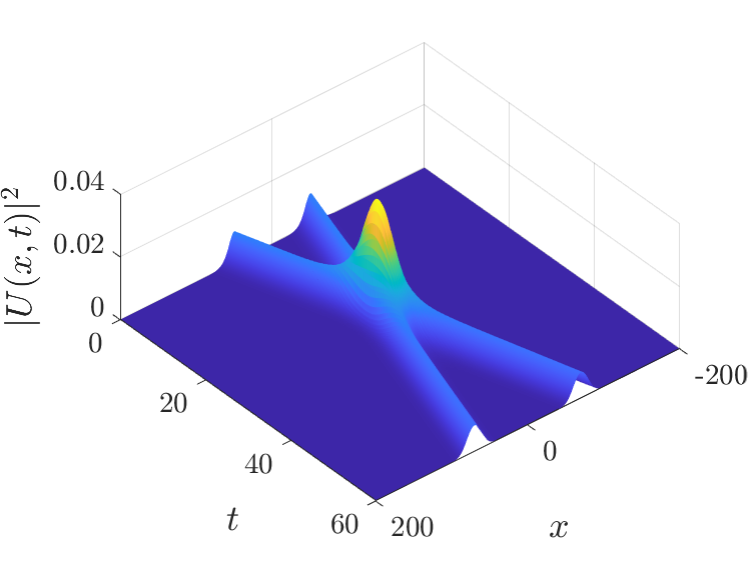}
\hspace{-0cm}\includegraphics[width=.42\textwidth]{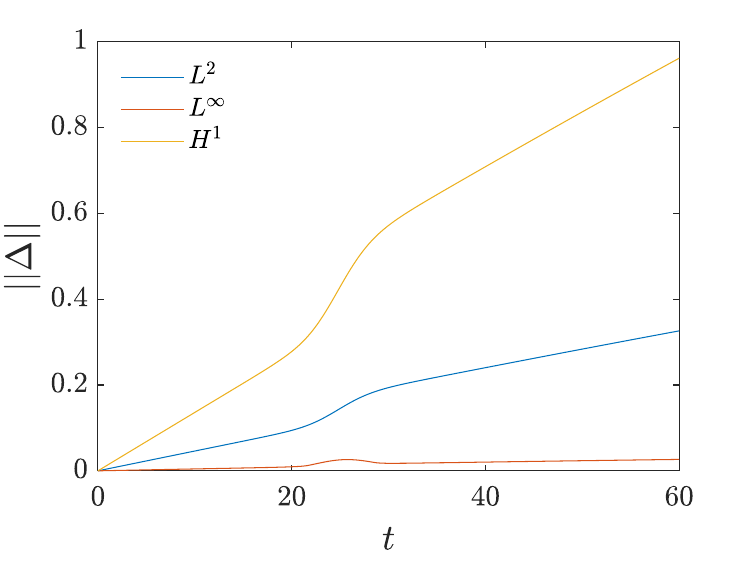}\\
\hspace{-0cm}\includegraphics[width=.45\textwidth]{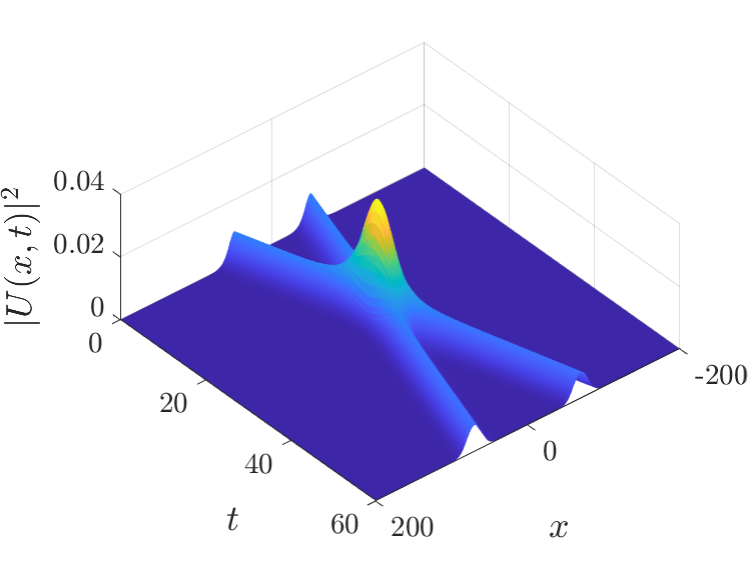}
\hspace{0cm}\includegraphics[width=.44\textwidth]{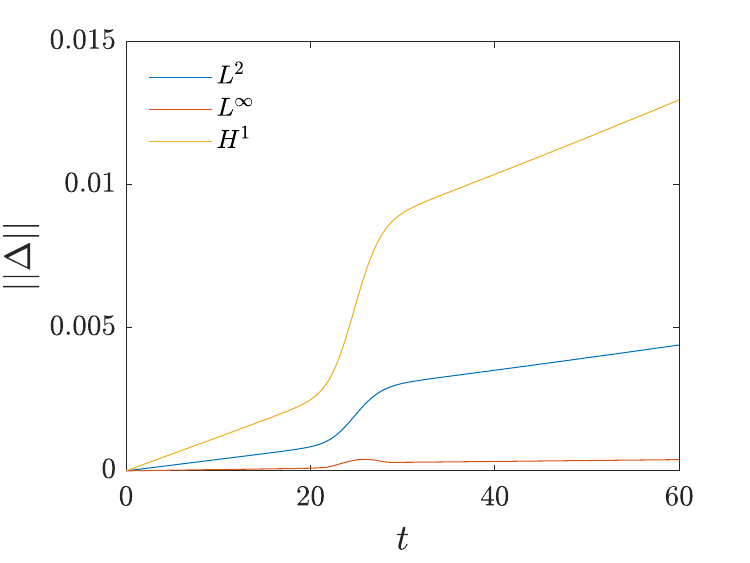}
		\end{tabular}
	\end{center}
	\caption{Fast bright soliton collisions. Left column: Spatiotemporal evolution of the density $|U(x,t)|^2$ initiated by the datum \eqref{BSCin} with $A=0.1$, $x_0=50$, $c_s=2$ for the non-integrable NLS equation with power nonlinearity \eqref{NLSP} in the supercritical case $p=3$ (top) and the non-integrable saturable NLS equation \eqref{Sat2} (bottom).  Right column: Evolution of the norms $\left\|\Delta(t)\right\|_{\mathcal{X}}$ with $\mathcal X = L^2, H^1 \text{ or } L^\infty$  for the power nonlinearity $p=3$ (top)  and the saturable nonlinearity (bottom). In both cases, $\nu=1/2$ and $\mu=\gamma=1$.
}
	\label{fig3}
\end{figure}

Figure \ref{fig3} depicts the results for the soliton collision dynamics in the case of the fast solitons. The top row corresponds to the power nonlinearity model \eqref{NLSP} in the supercritical case $p=3$ and the bottom row for the saturable NLS equation \eqref{Sat2}. The plots of the spatiotemporal evolution of the density $|U(x,t)|^2$ in the left panels reveal  that the non-integrable systems exhibit  collision dynamics that are almost identical to the integrable (cubic) case. This persistence of integrable soliton dynamics is even more evident in the plots of the evolution of the distance norm $\left\|\Delta(t)\right\|_{\mathcal{X}}$ with $\mathcal X = L^2(\mathbb R),  H^1(\mathbb R) \text{ or } L^\infty(\mathbb R)$, which are in excellent agreement with the predictions of Theorem \ref{TH1}. Note that the collision event is captured by the evolution of norms  as the sigmoid-shaped portion of the curve featuring  an interlude of increased, but still moderate, slope. 

Furthermore,  the evolution of norms highlights two features:
\begin{enumerate}[label=(\roman*), leftmargin=7mm, itemsep=1mm, topsep=2mm]
\item The saturable NLS equation seems to be a more structurally stable non-integrable model with respect to the dynamics of the integrable NLS equation, an effect which was already observed in the dynamics of the single bright soliton earlier. 
\item The relatively slow growth of $\left\|\Delta(t)\right\|_{L^{\infty}(\mathbb R)}$ 
indicates that the solutions are closer in the sense of pointwise convergence than they are in the sense of convergence in the $L^2(\mathbb R)$ or $H^1(\mathbb R)$ norms; a larger deviation in the latter norms hints at rather complex  dynamics on smaller scales that may contribute to the overall difference between the integrable and non-integrable cases. 
\end{enumerate}
These two features are further elucidated by the study of the slow solitons collisions shown in Figure \ref{fig4}.  There, the top row corresponds to the power nonlinearity with $p=3$. The breakdown of integrability in this case appears to be dramatic, as it is  illustrated by the spatiotemporal evolution of the density in the top left panel. Part of the energy is trapped at the collision site, a common feature also  present in other non-integrable models (e.g. the Klein-Gordon equation $\phi^4$). 
This is not the case for the saturable nonlinearity, as shown in the bottom row. Qualitatively, the dynamics of the saturable NLS equation is  almost identical to the one of the integrable NLS equation and, in contrast to the power nonlinearity, no energy trapping effect is observed. The vast difference between the two non-integrable models is also prominent in the corresponding evolution of the distance norms $\left\|\Delta(t)\right\|_{\mathcal{X}}$. For the septic power nonlinearity, the deviation of the $L^{\infty}$ norm is still moderate and compliant with an $L^{\infty}$ closeness of solutions. However, the larger growth and the behavior of  the $L^2$ and $H^1$ norms (particularly after the collision) gives further demonstration of the integrability breakdown induced by the power nonlinearity. On the other hand, the corresponding study of the behavior of $\left\|\Delta(t)\right\|_{\mathcal{X}}$ for all three norms in the case of the saturable NLS equation provides additional evidence of structural stability of that model in reference to the integrable one. 
 
\begin{figure}[ht!]
	\begin{center}
		\begin{tabular}{cc}
\hspace{-0cm}\includegraphics[width=.45\textwidth]{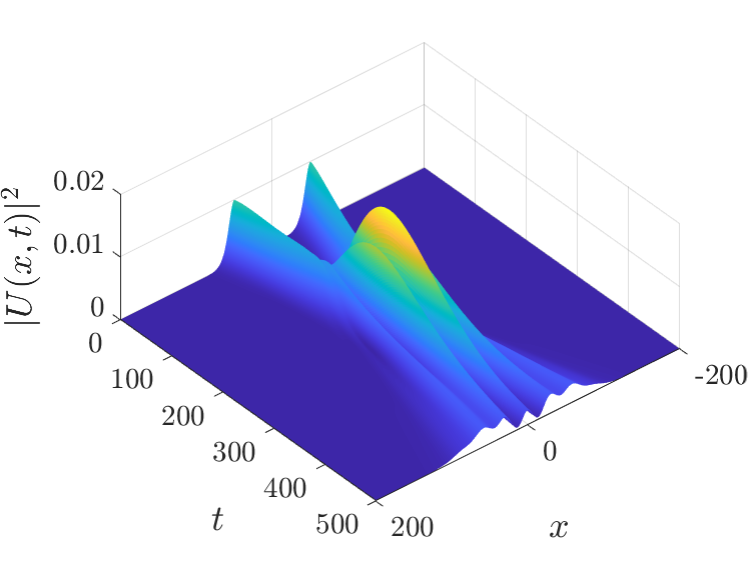}
\hspace{0cm}\includegraphics[width=.41\textwidth]{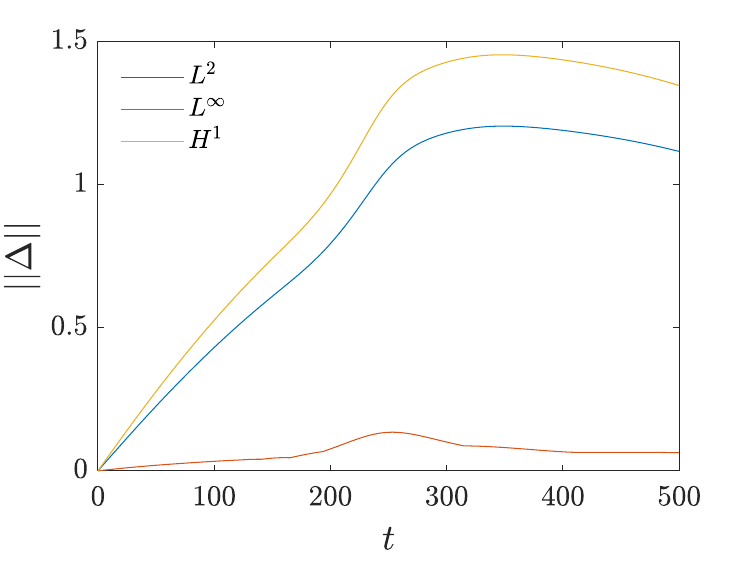}\\
\hspace{-0.cm}\includegraphics[width=.45\textwidth]{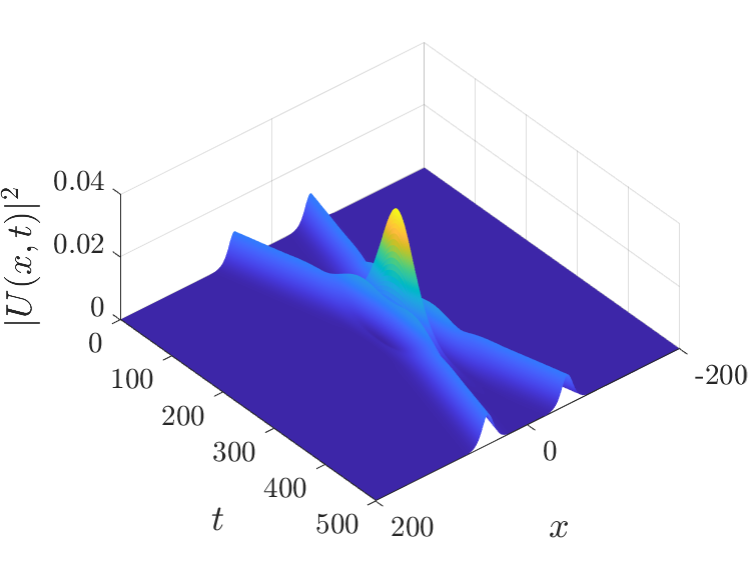}
\hspace{-0cm}\includegraphics[width=.42\textwidth]{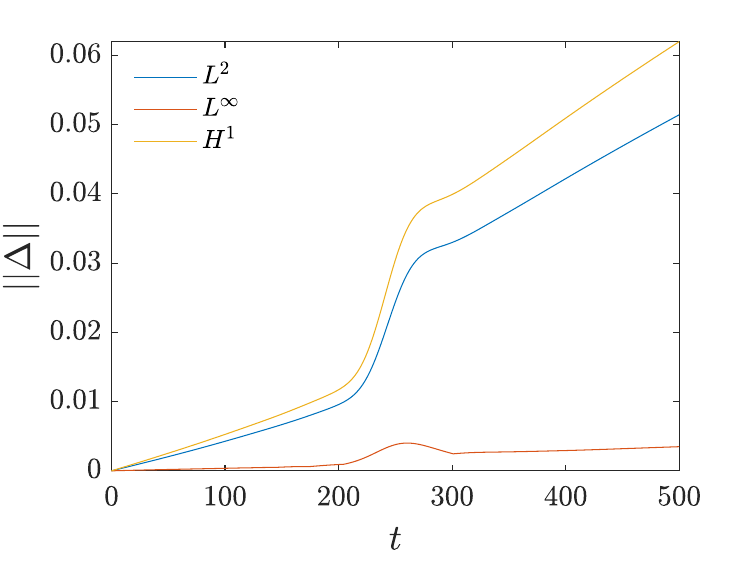}
		\end{tabular}
	\end{center}
	\caption{Slow bright soliton collisions. Same as in Figure \ref{fig3} but with $c_s=0.02$.
}
	\label{fig4}
\end{figure}

%
%
%
\section{Nonzero boundary conditions at infinity}
\label{nzbc-s}

The main objective of this section is to prove the analogue of Theorem \ref{Theorem:stability} --- namely, Theorem~\ref{nzbc-t} below --- on the proximity of solutions between the integrable and non-integrable NLS equations \eqref{NLS} and \eqref{noninNLS}  when these are supplemented with a broad class of \textit{nonzero}  boundary conditions at infinity. 
For this task, we first assume (local) existence of solutions to the non-integrable problem in $H^1(\mathbb R)$ in order to derive the relevant closeness estimates. 
In this regard, we note that the well-posedness results for the general non-integrable problem with nonzero boundary conditions are quite limited when compared to the plethora of results in the case of zero boundary conditions and, in particular, the powerful results of Theorem \ref{zbc-wp-t}. 
Thus, en route to establishing Theorem~\ref{nzbc-t}, we also prove local existence in $H^1(\mathbb R)$ for the non-integrable NLS equations \eqref{NLSP} and \eqref{Sat2} that correspond to the power and saturable nonlinearities, respectively, and as such represent the two most prominent  members of the general family of equations \eqref{noninNLS}. 

Motivated by the physically significant description given in \cite{bm2017,blmt2018,Defoc}, we supplement the integrable  and non-integrable NLS equations \eqref{NLS} and \eqref{noninNLS} with the following general class of nonzero boundary conditions at infinity:
\begin{equation}\label{uU-bc}
\lim_{|x|\rightarrow\infty} u(x,t) = \lim_{|x|\rightarrow\infty}e^{i\mu q_0^2t}\zeta(x), 
\quad 
\lim_{|x|\rightarrow\infty} U(x,t) = \lim_{|x|\rightarrow\infty} e^{i\gamma F(q_0^2)t} \zeta(x), \quad t\geq 0,
\end{equation}
where the complex-valued function $\zeta(x)$ belongs to the Zhidkov space 
\begin{equation}\label{zhi1}
X^1(\mathbb R)
:=
\left\{ \zeta\in L^\infty(\mathbb R): \zeta'\in L^2(\mathbb R)\right\}
\end{equation}
and satisfies 
\begin{equation}\label{zhi2}
\lim_{x \to \pm \infty} \zeta (x) = \zeta_\pm \in \mathbb C, \quad |\zeta_\pm| = q_0>0.	
\end{equation}
According to the above boundary conditions, the behavior of the solutions $u(x,t)$ and $U(x,t)$ as $|x| \to \infty$ is dictated by the function $\zeta (x)$, which approaches a constant background of size $q_0>0$ as described by the limit conditions~\eqref{zhi2}. 

The precise form of the conditions \eqref{uU-bc} can be motivated as follows.
Let $u(x, t)$ and $U(x, t)$ be solutions of  equations \eqref{NLS} and \eqref{noninNLS}, respectively, such that  
\begin{equation}\label{upm-mot-0}
\lim_{x\to\pm\infty}u(x, t) = u_\pm(t),
\quad
\lim_{x\to\pm\infty}U(x, t) = U_\pm(t),
\end{equation}
where $U_\pm(t)$ and $u_\pm(t)$ are temporal functions of constant modulus equal to $q_0$ but are otherwise to be determined. Then, taking the limit of \eqref{NLS} and \eqref{noninNLS} as $x\to\pm\infty$ while assuming that the nonlinearity function $F$ has sufficient smoothness so that
$
\lim_{|x| \to \infty} F(|U|^2) = F(\lim_{|x|\to\infty} |U|^2) = F(q_0^2)
$
(note, in particular, that this is true both in the semilinear case \eqref{NLSP} and in the saturable case \eqref{Sat2}), 
we have
\begin{equation}\label{upm-mot}
\begin{aligned}
&i (u_\pm)_t + \mu q_0^2 u_\pm = 0 \ \Rightarrow \ u_\pm(t) = e^{i\mu q_0^2t}u_\pm(0),
\\
&i (U_\pm)_t + \gamma F(q_0^2) U_\pm = 0 \ \Rightarrow \ U_\pm(t) = e^{i\gamma F(q_0^2)t} U_\pm(0).
\end{aligned}
\end{equation}
Therefore, if the initial data are such that $\lim_{x\to\pm\infty} u(x, 0) = \lim_{x\to\pm\infty} U(x, 0) = \zeta_\pm$, i.e. if $u_\pm(0) = U_\pm(0) = \zeta_\pm$, then for $\zeta$ satisfying \eqref{zhi2} the expressions \eqref{upm-mot-0} and \eqref{upm-mot} give rise to the conditions~\eqref{uU-bc}.

A few remarks are now in place:
\begin{enumerate}[label=(\roman*), leftmargin=7mm, itemsep=1mm, topsep=2mm]
\item The assumption for $\zeta \in X^1(\mathbb{R})$ covers the standard case $\zeta(x)=\zeta_0\in\mathbb{C}$ being a constant with $|\zeta_0|=q_0>0$.    
\item 
The pure step function
$$
\zeta_s(x)
=	
\left\{
\begin{array}{ll}
\zeta_+\in\mathbb{C}, &x>0,
\\
\zeta_-\in\mathbb{C}, &x<0,
\end{array}
\right.
$$
does not belong to $X^1(\mathbb R)$, since $\zeta_s'(x) = \left(\zeta_+-\zeta_-\right) \delta(x)$ in the sense of distributions and so $\zeta_s' \notin L^2(\mathbb R)$.
Nevertheless, although our analysis does not cover the pure step  $\zeta_s$, it does cover any continuous function approximating $\zeta_s$, as the derivative of such a function will belong to $L^2(\mathbb R)$ (and hence that function will belong to $X^1(\mathbb R)$).
\item 
Generalizations from the pure step $\zeta_s$ to a nontrivial $\zeta$ belonging to the class \eqref{zhi1}-\eqref{zhi2} covered by the present work are relevant in the context of the dynamics of NLS equations. In the case where $q_0=0$, localized waveforms in the form of Peregrine solitons may locally emerge  on a nontrivial decaying $\zeta(x)$ induced by the initial conditions for the semiclassical limit of the focusing integrable NLS equation  \cite{BT2009}. For fundamental results on the semiclassical problem for NLS, we refer to~\cite{KMPM2003,PK1998,TVZ2004}. The local emergence of localized structures reminiscent of the Peregrine soliton has been observed also in experimental setups \cite{ups,ups2} and is expected to be robust in the presence of damping and forcing effects \cite{f1,f2}.
However, herein we will restrict ourselves to the case $q_0>0$.
\end{enumerate}

The boundary conditions \eqref{uU-bc} can be made time independent (i.e. \textit{constant}) via the change of variables
\begin{equation}\label{zhi4}
u(x, t) = e^{i\mu q_0^2t} q(x, t), 
\quad 
U(x, t) = e^{i\gamma F(q_0^2)t} Q(x, t).
\end{equation}
Then, equations \eqref{NLS} and \eqref{noninNLS} take the form
\begin{align}
\label{nvNINLS}
&i q_t + \nu q_{xx}+ \mu \left(|q|^2-q_0^2\right)q=0,
\quad
\\
&iQ_t + \nu Q_{xx}+ \gamma \left[F(|Q|^2) - F(q_0^2)\right] Q = 0,
\label{nvNINLS-Q}
\end{align}
the initial conditions remain unchanged, namely $q(x, 0) = u_0(x)$ and $Q(x, 0) = U_0(x)$, and the  boundary conditions become
\begin{equation}\label{zhi5a}
\lim_{x \to \pm \infty} q(x,t) = \lim_{x \to \pm \infty} Q(x,t)
=
\zeta_\pm
\end{equation}
so that $\lim_{|x|\rightarrow\infty} |q(x,t)| = \lim_{|x|\rightarrow\infty} |Q(x,t)| = q_0>0$.
Furthermore,  the additional change of variables
\begin{equation}\label{zhi6}
q(x,t) = \phi(x,t)+\zeta(x),
\quad
Q(x,t) = \Phi (x,t)+\zeta(x)
\end{equation}
leads to the modified NLS equations
\begin{align}
	\label{cNLS}
&i\phi_t+ \nu \left(\phi+\zeta\right)_{xx} + \mu \left(|\phi+\zeta|^2-q_0^2\right)(\phi+\zeta)=0,
\\
	\label{pNLS}
&i \Phi_t+ \nu \left(\Phi+\zeta\right)_{xx} +\gamma \left[F(|\Phi+\zeta|^2) - F(q_0^2)\right] (\Phi+\zeta)=0,
\end{align}
along with the initial conditions 
\begin{equation}\label{mod-ic}
\phi(x, 0) = u_0(x) - \zeta(x) =: \phi_0(x),
\quad
\Phi(x, 0) = U_0(x) - \zeta(x) =: \Phi_0(x)
\end{equation}
and \textit{zero} boundary conditions at infinity, i.e.
\begin{equation}\label{vbcn}
\lim_{|x|\rightarrow\infty}\phi(x,t)
=
\lim_{|x|\rightarrow\infty}\Phi(x,t) = 0,\quad t\geq 0.
\end{equation}

The rest of this section is organized as follows. 
First, we establish $L^2$ and $H^1$ closeness estimates for the modified NLS equations \eqref{cNLS} and \eqref{pNLS}. Through the transformations \eqref{zhi4} and \eqref{zhi6}, these closeness results imply corresponding estimates for the difference 
$$
e^{-i\mu q_0^2t} u(x, t) - e^{-i\gamma F(q_0^2) t} U(x, t)
$$ 
involving the solutions of the original integrable and non-integrable NLS equations \eqref{NLS} and \eqref{noninNLS}. 
The derivation of the closeness estimates is accomplished under the assumption of (at least) local $H^1$ existence for the non-integrable equation \eqref{noninNLS}. 
In this connection, we emphasize that the case of nonzero boundary conditions is substantially different from the one of zero boundary conditions, even at the fundamental level of well-posedness. In particular, to the best of our knowledge, the well-posedness of the non-integrable NLS equation \eqref{noninNLS} with general nonlinearities is  much less understood in the former case. In the focusing regime, one of the few known results in Sobolev spaces is due to \cite{Def4}, who established local existence for the semilinear Schr\"odinger equation \eqref{NLSP} in the case of $H^s(\mathbb R)$ perturbations of the background wave with $s>1/2$. It should be noted that the proofs in \cite{Def4} are given only for $p=1$ corresponding to the integrable cubic NLS equation, although it is remarked that the arguments should go through also in the case of general $p>1$. 
In this regard, below we also prove local existence in $H^1(\mathbb R)$ for the modified NLS equation \eqref{pNLS} (and hence for the non-integrable NLS equation \eqref{noninNLS}) in the cases of a general power nonlinearity as well as of a saturable nonlinearity, which are respectively associated with the models \eqref{NLSP} and \eqref{Sat2}.
This local existence result is crucial as (i) it removes the relevant assumption used for the derivation of the closeness estimates, and (ii) via a continuity argument with respect to time (since our local solutions end up in the class $C([0,T],H^1(\mathbb{R}))$), it allows us to obtain the analogue of Theorem \ref{Theorem:stability} in the case of the nonzero boundary conditions \eqref{uU-bc}, namely Theorem \ref{nzbc-t}.
\\[3mm]
\noindent
\textbf{Closeness estimates.}
As noted above, in order to investigate the proximity between the solutions of equations \eqref{cNLS} and \eqref{pNLS}, we first assume local existence of these solutions in $H^1(\mathbb R)$. Local estimates validating this assumption in the case of a general power nonlinearity and the saturable nonlinearity are obtained after the proof of the closeness estimates.

The difference of solutions to equations \eqref{cNLS} and \eqref{pNLS} is the same as the one of solutions to equations \eqref{nvNINLS} and \eqref{nvNINLS-Q}, because it is unaffected by the change of variables \eqref{zhi6}. On the other hand, the change of variables \eqref{zhi4} does not preserve that difference, as it results in different phase factors multiplying each of the solutions to the original equations \eqref{NLS} and \eqref{noninNLS}. That is, 
$$
\Delta(x, t) := \phi(x, t) - \Phi(x, t) = q(x, t) - Q(x, t) = e^{-i\mu q_0^2 t} \, u(x, t) - e^{-i\gamma F(q_0^2) t} \, U(x, t)
$$ 
with $\Delta$ satisfying the equation
\begin{equation}\label{eq:Deltanv}
i\Delta_t + \nu \Delta_{xx} = \gamma \, G_{F, \Phi}(x, t)  - \mu \,  G_{1,\phi}(x, t), 
\end{equation}
where 
\begin{equation}\label{G-def}
G_{F, \Phi}(x, t) 
= 
\left[F(|\Phi+\zeta|^2) - F(q_0^2)\right] (\Phi+\zeta)
\end{equation}
and  the subscript ``1'' denotes the identity function, so that $G_{1,\phi}(x, t) = \left(|\phi+\zeta|^2-q_0^2\right)(\phi+\zeta)$. 
Starting from equation \eqref{eq:Deltanv}, we first establish closeness estimates in $L^2(\mathbb R)$ and then in $H^1(\mathbb R)$, the latter also implying the result in $L^\infty(\mathbb R)$ via Sobolev embedding.
\\[2mm]
\textit{$L^2$ closeness.} 
Taking the Fourier transform \eqref{ft-def} of  equation \eqref{eq:Deltanv} and integrating over $t$, we find
\begin{equation}\label{D-hat-nzbc}
\widehat \Delta(\xi, t) 
=
e^{-i \nu \xi^2 t} \widehat \Delta(\xi, 0) - i\int_0^t e^{-i \nu \xi^2(t-\tau)} \left[\gamma \, \widehat G_{F, \Phi}(\xi, \tau) - \mu \, \widehat G_{1,\phi}(\xi, \tau)\right] d\tau.
\end{equation}
Hence, by Plancherel's theorem, Minkowski's integral inequality  and the triangle inequality, 
\begin{equation}\label{D-L2}
\begin{aligned}
\left\| \Delta(t) \right\|_{L^2(\mathbb R)}
&\leq
\frac{1}{2\pi} \big\|\widehat \Delta(0) \big\|_{L^2(\mathbb R)}
+
\frac{1}{2\pi} \int_0^t \left\| e^{-i \nu \xi^2(t-\tau)} \left[\gamma \, \widehat G_{F, \Phi}(\xi, \tau) - \mu \, \widehat G_{1,\phi}(\xi, \tau)\right] \right\|_{L^2(\mathbb R)} d\tau
\\
&=
\big\| \Delta(0) \big\|_{L^2(\mathbb R)}
+
\int_0^t \left[ |\gamma| \, \big\| G_{F, \Phi}(\tau) \big\|_{L^2(\mathbb R)} + |\mu| \, \big\| G_{1, \phi}(\tau)\big\|_{L^2(\mathbb R)} \right] d\tau
\end{aligned}
\end{equation}
and we need to estimate the spatial $L^2$ norms of $G_{F, \Phi}$ and $G_{1, \phi}$. Note that the conditions \eqref{zhi2} and~\eqref{vbcn} already imply $\lim_{|x|\to\infty} G_{F, \Phi}(x, t) = \lim_{|x|\to\infty} G_{1, \phi}(x, t) = 0$. More precisely, we have the following estimate:
\begin{lemma}\label{G-L2-l}
Let $F$ satisfy the properties \eqref{F-prop}. Then, for any $p\geq 1$ and each $t\geq 0$, the function $G_{F, \Phi}$ defined by \eqref{G-def} admits the bound
\begin{equation}\label{G-L2-est}
\left\| G_{F, \Phi}(t) \right\|_{L^2(\mathbb R)} 
\leq
2\sqrt 2 \, K  \left( \left\| \Phi(t) \right\|_{L^\infty(\mathbb R)} + \left\| \zeta \right\|_{L^\infty(\mathbb R)} + q_0\right)^{2p}
\!
\left(\left\| \Phi(t) \right\|_{L^2(\mathbb R)} + \big\| |\zeta| - q_0 \big\|_{L^2(\mathbb R)}\right)
\end{equation}
where $K$ is the constant associated with $F$ via \eqref{F-prop}.
\end{lemma}
\begin{proof}
By the first of the properties \eqref{F-prop}, we have
$$
\left|F(|\Phi+\zeta|^2) - F(q_0^2)\right|
\leq
K \big(|\Phi+\zeta|^{2(p-1)} + q_0^{2(p-1)}\big) \left\|\Phi+\zeta|^2 - q_0^2\right|, \quad p\geq 1.
$$
Then, noting that
$$
\left\|\Phi+\zeta|^2 - q_0^2\right|
\leq
\left\{
\def\arraystretch{1.2}
\begin{array}{ll}
\left(|\Phi|+|\zeta|\right)^2 - q_0^2, &|\Phi+\zeta|\geq q_0,
\\
q_0^2 - \left\|\Phi|-|\zeta|\right|^2, &|\Phi+\zeta|\leq q_0,
\end{array}
\right.
$$
we obtain
\begin{equation}\label{diff-ineq}
\begin{aligned}
\left\|\Phi+\zeta|^2 - q_0^2\right|
&\leq
\left|\left(|\Phi|+|\zeta|\right)^2 - q_0^2\right|
+
\left|\left(|\Phi|-|\zeta|\right)^2 - q_0^2\right|
\\
&\leq
2\left(|\Phi|+|\zeta| + q_0\right) \left(|\Phi|+\big| |\zeta| - q_0 \big|\right)
\end{aligned}
\end{equation}
and, in turn,
\begin{equation}\label{F-diff-nzbc}
\begin{aligned}
\left|F(|\Phi+\zeta|^2) - F(q_0^2)\right|
&\leq
2K \big(|\Phi+\zeta|^{2(p-1)} + q_0^{2(p-1)}\big)
\left(|\Phi|+|\zeta| + q_0\right) \left(|\Phi|+\big| |\zeta| - q_0 \big|\right)
\\
&\leq
2K \left(|\Phi| + |\zeta| + q_0 \right)^{2p-1}
\left(|\Phi|+\big| |\zeta| - q_0 \big|\right).
\end{aligned}
\end{equation}
Using this inequality, we find
\begin{equation*}
\begin{aligned}
&\quad
\big\| \left[F(|\Phi+\zeta|^2) - F(q_0^2)\right] (\Phi+\zeta)\big\|_{L^2(\mathbb R)}
\\
&\leq
2K
\left(
\int_{\mathbb R}
\left(|\Phi| + |\zeta| + q_0\right)^{2(2p-1)} 
\left( 
|\Phi| + \big| |\zeta| - q_0 \big|
\right)^2
 \left(|\Phi|+|\zeta|\right)^2 dx
 \right)^{\frac 12}
\\
&\leq
2 \sqrt 2 K 
\sup_{x\in\mathbb R} 
\left[\left(|\Phi| + |\zeta| + q_0\right)^{2p-1} 
\left(|\Phi|+|\zeta|\right)\right]
\left(
\int_{\mathbb R}
\big( 
|\Phi|^2 + \big| |\zeta| - q_0 \big|^2
\big) dx
\right)^{\frac 12},
\end{aligned}
\end{equation*}
which readily implies the claimed estimate.
\end{proof}

In view of the $L^2$ estimate \eqref{G-L2-est} for general $G_{F, \Phi}$ as well as for $G_{1, \phi}$ with $p=K=1$,  inequality~\eqref{D-L2} yields
\begin{align}\label{D-nzbc-L^2}
\left\| \Delta(t) \right\|_{L^2(\mathbb R)}
&\leq
\big\| \Delta(0) \big\|_{L^2(\mathbb R)}
\nn\\
&\quad
+
2\sqrt 2 \, t \sup_{\tau \in [0, t]} 
\bigg[
 |\gamma| K  \left( \left\| \Phi(\tau) \right\|_{L^\infty(\mathbb R)} + \left\| \zeta \right\|_{L^\infty(\mathbb R)} + q_0\right)^{2p}
\left(\left\| \Phi(\tau) \right\|_{L^2(\mathbb R)} + \big\| |\zeta| - q_0 \big\|_{L^2(\mathbb R)}\right)
\nn\\
&\quad
+ |\mu|
 \left( \left\| \phi(\tau) \right\|_{L^\infty(\mathbb R)} + \left\| \zeta \right\|_{L^\infty(\mathbb R)} + q_0\right)^2
\left(\left\| \phi(\tau) \right\|_{L^2(\mathbb R)} + \big\| |\zeta| - q_0 \big\|_{L^2(\mathbb R)}\right)
\bigg],
\end{align}
which for each $t\geq 0$ provides an $L^2$ estimate for $\Delta(t)$ in terms of the $L^2$ norms of $\Delta(0)$, $\Phi(t)$, $\phi(t)$, $|\zeta|-q_0$, the $L^\infty$ norms of $\Phi(t)$, $\phi(t)$, $\zeta$, and the background $q_0$.
\\[2mm]
\noindent
\textit{$H^1$ closeness.}
Starting from the definition of the Sobolev norm and using expression \eqref{D-hat-nzbc} along with the unitarity of $e^{-i\nu\xi^2t}$ and Minkowski's integral inequality, we have 
\begin{equation}\label{D-sup-t1}
\begin{aligned}
\no{\Delta(t)}_{H^1(\mathbb R)}
&=
\left(\int_{\mathbb R} \left(1+\xi^2\right) \Big| e^{-i\nu \xi^2 t} \widehat \Delta(\xi, 0) - i\int_0^t e^{-i\nu \xi^2(t-\tau)} \left[\gamma \, \widehat G_{F, \Phi}(\xi, \tau) - \mu \, \widehat G_{1, \phi}(\xi, \tau) \right] d\tau \Big|^2 d\xi\right)^{\frac 12}
\\
&\leq
\sqrt 2 
\left(
\no{\Delta(0)}_{H^1(\mathbb R)} + \int_0^t \left[ |\gamma| \, \big\|G_{F, \Phi}(\tau)\big\|_{H^1(\mathbb R)} + |\mu| \, \big\|G_{1, \phi}(\tau)\big\|_{H^1(\mathbb R)} \right] d\tau
\right).
\end{aligned}
\end{equation}
Thus, for each $t\geq 0$, we need to estimate $G_{F, \Phi}(t)$ and $G_{1, \phi}(t)$ in $H^1(\mathbb R)$. 
\begin{lemma}\label{G-H1-l}
Let $F$ satisfy the properties \eqref{F-prop}. Then, for any $p\geq 1$ and each $t\geq 0$, the  function $G_{F, \Phi}$ defined by \eqref{G-def} admits the bound
\begin{align}\label{G-H1-est}
\left\| G_{F, \Phi}(t) \right\|_{H^1(\mathbb R)} 
&\leq
2\sqrt 2 \, K  \left( \left\| \Phi(t) \right\|_{L^\infty(\mathbb R)} + \left\| \zeta \right\|_{L^\infty(\mathbb R)} + q_0\right)^{2p-1}
\bigg[
2\left( \left\| \Phi(t) \right\|_{L^\infty(\mathbb R)} + \left\| \zeta \right\|_{L^\infty(\mathbb R)} + q_0\right)
\nn\\
&\quad
\cdot \left(\left\| \Phi(t) \right\|_{H^1(\mathbb R)} + \no{\zeta'}_{L^2(\mathbb R)} +  \big\| |\zeta| - q_0 \big\|_{L^2(\mathbb R)}\right)
+
\left(\no{\Phi(t)}_{L^\infty(\mathbb R)} + \no{|\zeta|-q_0}_{L^\infty(\mathbb R)} \right)
\nn\\
&\quad
\cdot
\left(\no{\Phi(t)}_{H^1(\mathbb R)} + \no{\zeta'}_{L^2(\mathbb R)}\right)
\bigg]
\end{align}
where $K$ is the constant associated with $F$ via \eqref{F-prop}.
\end{lemma}
\begin{proof}
Since $\left\| G_{F, \Phi}(t) \right\|_{H^1(\mathbb R)} = \left\| G_{F, \Phi}(t) \right\|_{L^2(\mathbb R)} + \left\| \p_x G_{F, \Phi}(t) \right\|_{L^2(\mathbb R)}$ and we have already estimated the $L^2$ norm of $G_{F, \Phi}$ via \eqref{G-L2-est}, we proceed to the $L^2$ norm of the derivative $\p_x G_{F, \Phi}$. 
Differentiating~\eqref{G-def} and applying the triangle inequality, we have
\begin{equation}\label{G'-ti}
\begin{aligned}
\no{\p_x G_{F, \Phi}(t)}_{L^2(\mathbb R)}
&\leq
2 \no{F'(|\Phi+\zeta|^2) (\Phi+\zeta)^2 (\Phi+\zeta)_x}_{L^2(\mathbb R)}
\\
&\quad
+
 \no{\left[ F(|\Phi+\zeta|^2)-F(q_0^2)\right] (\Phi+\zeta)_x}_{L^2(\mathbb R)}.
\end{aligned}
\end{equation}
For the first term on the right-hand side, by the third of the properties \eqref{F-prop} we find
\begin{equation}
\begin{aligned}
\no{F'(|\Phi+\zeta|^2) (\Phi+\zeta)^2 (\Phi+\zeta)_x}_{L^2(\mathbb R)}
\leq
K \left(\int_{\mathbb R} |\Phi+\zeta|^{4p} \left|(\Phi+\zeta)_x
\right|^2 dx\right)^{\frac 12}
&\\
\leq
\sqrt 2 \, K \left(\no{\Phi}_{L^\infty(\mathbb R)} + \no{\zeta}_{L^\infty(\mathbb R)}\right)^{2p}
\left(\no{\Phi}_{H^1(\mathbb R)} + \no{\zeta'}_{L^2(\mathbb R)}\right).&
\nn
\end{aligned}
\end{equation}
Moreover, for the second term, using inequality \eqref{F-diff-nzbc} we infer
\begin{align}
&\quad
\no{\left[ F(|\Phi+\zeta|^2)-F(q_0^2)\right] (\Phi+\zeta)_x}_{L^2(\mathbb R)}
\nn\\
&\leq
2K
\left(\int_{\mathbb R}
\left(|\Phi| + |\zeta| + q_0 \right)^{2(2p-1)}
\left(|\Phi|+\big| |\zeta| - q_0 \big|\right)^2
\left|(\Phi+\zeta)_x\right|^2 dx
\right)^{\frac 12}
\nn\\
&\leq
2 \sqrt 2 \, K \left(\no{\Phi}_{L^\infty(\mathbb R)} +\no{\zeta}_{L^\infty(\mathbb R)} + q_0 \right)^{2p-1}
\left(\no{\Phi}_{L^\infty(\mathbb R)} + \no{|\zeta|-q_0}_{L^\infty(\mathbb R)} \right)
\left(\no{\Phi}_{H^1(\mathbb R)} + \no{\zeta'}_{L^2(\mathbb R)}\right).
\nn
\end{align}
Combining the last two estimates with \eqref{G'-ti}, we deduce 
\begin{align*} 
&
\no{\p_x G_{F, \Phi}(t)}_{L^2(\mathbb R)}
\leq
2\sqrt 2 \, K \left(\no{\Phi}_{L^\infty(\mathbb R)} + \no{\zeta}_{L^\infty(\mathbb R)}\right)^{2p}
\left(\no{\Phi}_{H^1(\mathbb R)} + \no{\zeta'}_{L^2(\mathbb R)}\right)
\\
&
+
2\sqrt 2 \, K \left(\no{\Phi}_{L^\infty(\mathbb R)} +\no{\zeta}_{L^\infty(\mathbb R)} + q_0 \right)^{2p-1}
\left(\no{\Phi}_{L^\infty(\mathbb R)} + \no{|\zeta|-q_0}_{L^\infty(\mathbb R)} \right)
\left(\no{\Phi}_{H^1(\mathbb R)} + \no{\zeta'}_{L^2(\mathbb R)}\right)
\end{align*}
which alongside \eqref{G-L2-est} yields the desired estimate \eqref{G-H1-est}.
\end{proof}

In view of the $H^1$ estimate \eqref{G-H1-est} for general $G_{F, \Phi}$  as well as for $G_{1, \phi}$ with $K=p=1$, inequality~\eqref{D-sup-t1} yields
\begin{equation}
\label{nvH1}
\begin{aligned}
&\quad
\no{\Delta(t)}_{H^1(\mathbb R)}
\leq
\sqrt 2 \no{\Delta(0)}_{H^1(\mathbb R)}   
+ 4t \sup_{\tau \in [0, t]}  
\bigg\{ 
|\gamma| K  \left( \left\| \Phi(\tau) \right\|_{L^\infty(\mathbb R)} + \left\| \zeta \right\|_{L^\infty(\mathbb R)} + q_0\right)^{2p-1}
\\
&\quad
\cdot
\bigg[
2\left( \left\| \Phi(\tau) \right\|_{L^\infty(\mathbb R)} + \left\| \zeta \right\|_{L^\infty(\mathbb R)} + q_0\right)
\left(\left\| \Phi(\tau) \right\|_{H^1(\mathbb R)} + \no{\zeta'}_{L^2(\mathbb R)} +  \big\| |\zeta| - q_0 \big\|_{L^2(\mathbb R)}\right)
\\
&\qquad
+
\left(\no{\Phi(\tau)}_{L^\infty(\mathbb R)} + \no{|\zeta|-q_0}_{L^\infty(\mathbb R)} \right)
\left(\no{\Phi(\tau)}_{H^1(\mathbb R)} + \no{\zeta'}_{L^2(\mathbb R)}\right)
\bigg]
\\
&\quad
+
|\mu|  \left( \left\| \phi(\tau) \right\|_{L^\infty(\mathbb R)} + \left\| \zeta \right\|_{L^\infty(\mathbb R)} + q_0\right)
\bigg[
2\left( \left\| \phi(\tau) \right\|_{L^\infty(\mathbb R)} + \left\| \zeta \right\|_{L^\infty(\mathbb R)} + q_0\right)
\\
&\qquad
\cdot\left(\left\| \phi(\tau) \right\|_{H^1(\mathbb R)} + \no{\zeta'}_{L^2(\mathbb R)} + \big\| |\zeta| - q_0 \big\|_{L^2(\mathbb R)}\right)
\\
&\quad
+
\left(\no{\phi(\tau)}_{L^\infty(\mathbb R)} + \no{|\zeta|-q_0}_{L^\infty(\mathbb R)} \right)
\left(\no{\phi(\tau)}_{H^1(\mathbb R)} + \no{\zeta'}_{L^2(\mathbb R)}\right)
\bigg]
\bigg\},
\end{aligned}
\end{equation}
which at each time $t\geq 0$ provides a spatial $H^1$ estimate for $\Delta(t)$ in terms of the $H^1$ norms of $\Delta(0)$, $\Phi(t)$, $\phi(t)$, the $L^2$ norms of $\zeta'$, $|\zeta|-q_0$, the $L^\infty$ norms of $\Phi(t)$, $\phi(t)$, $\zeta$, and the background $q_0$.
\\[2mm]
\noindent
\textit{$L^\infty$ closeness}.
By the Sobolev embedding theorem, the bound in \eqref{nvH1} is also satisfied by the $L^\infty$ norm of $\Delta(t)$, thereby extending our closeness result to that space as well.
\\[3mm]
\noindent
\textbf{Local existence in $H^1(\mathbb R)$ for power and saturable nonlinearities.} 
The closeness estimates~\eqref{D-nzbc-L^2} and \eqref{nvH1} were derived under the assumption of existence of solution to equations~\eqref{cNLS} and \eqref{pNLS} in the function spaces involved in those estimates. 
In what follows, we establish local existence in $H^1(\mathbb R)$ for equation \eqref{pNLS} with a general power nonlinearity (this result also covers equation \eqref{cNLS}) as well as with a saturable nonlinearity.
Note that, since the $L^\infty$ norm is controlled by the $H^1$ norm  via the Sobolev embedding theorem, existence in $H^1$ suffices for removing the aforementioned assumption from the derivation of the closeness estimates.
More precisely, we prove 
\begin{theorem}[Local existence in $H^1(\mathbb R)$]\label{nzbc-lwp-t}
Consider the modified NLS equation \eqref{pNLS} with either the saturable nonlinearity   $F(x) = \dfrac{x}{\kappa(1+x)}$ or the power nonlinearity $F(x)=x^p$, $p\geq 1$. Then, the associated Cauchy problem \eqref{pNLS}-\eqref{vbcn}  possesses a unique solution $\Phi\in B(0, \rho) \subset C([0,T_f],H^1(\mathbb{R}))$, where for some fixed $T>0$ the radius $\rho>0$ is defined by \eqref{rho-def} and the lifespan $0<T_f\leq T$ satisfies~\eqref{tf-sat} in the saturable case and \eqref{estT} in the case of the power nonlinearity.
\end{theorem}
\begin{proof} 
The proof combines linear estimates with bounds for the nonlinear terms and relies on a contraction mapping argument. It consists of several steps, starting from the linear terms and moving to the handling of the nonlinearities and their differences as required in order to invoke the contraction mapping theorem in $H^1(\mathbb R)$.

Taking the spatial Fourier transform of equation \eqref{pNLS} while noting that $\zeta' \to 0$ as $|x|\to\infty$ since $\zeta' \in L^2(\mathbb R)$, we obtain
\begin{equation}\label{Phi-ode}
\widehat \Phi_t + i \nu \xi^2 \widehat \Phi = - \nu \xi \widehat{\zeta'}(\xi) + i \gamma \widehat G_{F, \Phi}
\end{equation}
with $G_{F, \Phi}$ defined by \eqref{G-def}.
Then, integrating with respect to $t$, implementing the initial condition~\eqref{mod-ic} and inverting the Fourier transform, we arrive at the integral equation
\begin{equation}\label{L-def-F-0}
\Phi(x, t) 
= \Lambda[\Phi](x, t) 
\end{equation}
where 
\begin{equation}\label{L-def-F}
\begin{aligned}
\Lambda[\Phi](x, t) 
&:= \frac{1}{2\pi} \int_{\mathbb R} e^{i\xi x - i\nu\xi^2 t} \,  \widehat \Phi_0(\xi) d\xi
\\
&\quad
+
\frac{1}{2\pi} \int_{\mathbb R} e^{i\xi x} 
\bigg[i \gamma \int_0^t e^{-i\nu\xi^2 (t-\tau)}  \widehat G_{F, \Phi}(\xi, \tau) d\tau 
+
\frac{e^{-i\nu\xi^2 t} - 1}{i\xi} \, \widehat{\zeta'}(\xi)\bigg]
 d\xi.
\end{aligned}
\end{equation}
This formulation motivates our notion of solution to the Cauchy problem \eqref{pNLS}-\eqref{vbcn}, namely we say that $\Phi$ solves \eqref{pNLS}-\eqref{vbcn} if it satisfies the integral equation \eqref{L-def-F-0}. Thus, proving local existence of solution to \eqref{pNLS}-\eqref{vbcn} amounts to showing that \eqref{L-def-F-0} has a (unique) fixed point, which in turn will be established by proving that the mapping $\Phi \mapsto \Lambda[\Phi]$ is a contraction in the claimed solution space $C([0, T_f]; H^1(\mathbb R))$ for an appropriate choice of $T_f>0$. The first step in this direction is to obtain an estimate for the $H^1$ norm of $\Lambda[\Phi](t)$ for each $t\in\mathbb R$. This will be followed by a corresponding estimate for the difference $\Lambda[\Phi](t)-\Lambda[\Psi](t)$.  These two estimates will then be combined to deduce that $\Lambda$ is a contraction in a subset $B(0, \rho)$ of $C([0, T_f]; H^1(\mathbb R))$ for a suitably chosen radius $\rho>0$ and (minimum) lifespan $T_f>0$.

By \eqref{L-def-F} and the definition of the $H^1$ norm,
\begin{align*}
\no{\Lambda[\Phi](t)}_{H^1(\mathbb R)}^2
&\leq
2\int_{\mathbb R} \left(1+\xi^2\right) \big| \widehat \Phi_0(\xi) \big|^2 d\xi
+
4|\gamma|^2\int_{\mathbb R} \left(1+\xi^2\right) \left| \int_0^t e^{i\nu \xi^2 \tau} \widehat G_{F,\Phi}(\xi, \tau) d\tau
 \right|^2 d\xi
\nn\\
&\quad
+
4 \int_{\mathbb R} \left(1+\xi^2\right) \bigg| \frac{e^{-i\nu\xi^2 t} - 1}{i\xi} \bigg|^2  \big|\widehat{\zeta'}(\xi)\big|^2 d\xi.
\end{align*}
Concerning the integral involving $\zeta'$, we note that $\left|\frac{1-\cos(k)}{k}\right| \leq 1$ for all $k\in\mathbb R$ and so $\left| \frac{e^{-i\nu\xi^2 t} - 1}{i\xi} \right|^2 = 2 \, \frac{1 - \cos(\nu\xi^2 t)}{\xi^2} = 2|\nu| t  \left|\frac{1 - \cos(\nu\xi^2 t)}{\nu \xi^2 t}\right| \leq 2|\nu| t$ for all $\xi \in \mathbb R$ and $t\geq 0$.
This inequality is useful for $\xi$ near $0$, while for $|\xi|\geq 1$ we can simply observe that $0\leq \frac{1 - \cos(\nu\xi^2 t)}{\xi^2} \leq \frac{2}{\xi^2} \leq \frac{4}{1+\xi^2}$. Thus, using also Minkowski's integral inequality for the term involving $G_{F, \Phi}$, we infer
\begin{equation*}
\begin{aligned}
\no{\Lambda[\Phi](t)}_{H^1(\mathbb R)}^2
&\leq
2\no{\Phi_0}_{H^1(\mathbb R)}^2
+
4|\gamma|^2 \left(\int_0^t \no{G_{F,\Phi}(\tau)}_{H^1(\mathbb R)} d\tau\right)^2
\\
&\quad
+
4 \int_{|\xi|\leq1} 2 \cdot 2|\nu| t \, \big|\widehat{\zeta'}(\xi)\big|^2 d\xi
 +
 4\int_{|\xi|\geq1} \left(1+\xi^2\right) \frac{8}{1+\xi^2} \, \big|\widehat{\zeta'}(\xi)\big|^2 d\xi
\end{aligned}
\end{equation*}
and, in turn, by Plancherel's theorem we obtain
\begin{equation}\label{mf1}
\no{\Lambda[\Phi](t)}_{H^1(\mathbb R)}
\leq
\sqrt 2\no{\Phi_0}_{H^1(\mathbb R)}
+
4 \left(\sqrt{|\nu|} \sqrt t + \sqrt 2\right) \no{\zeta'}_{L^2(\mathbb R)}
+
2 |\gamma| \int_0^t \no{G_{F, \Phi}(\tau)}_{H^1(\mathbb R)} d\tau.
\end{equation}

The $H^1$ norm of $G_{F, \Phi}$ can be estimated in terms of the $H^1$ norm of $\Phi$ by combining estimate~\eqref{G-H1-est} with the Sobolev embedding theorem: 
\begin{align}\label{GF-H1-est-F}
&\quad
\left\| G_{F, \Phi}(t) \right\|_{H^1(\mathbb R)} 
\leq
4\sqrt 2 K  \left(\left\| \Phi(t) \right\|_{H^1(\mathbb R)} + \left\| \zeta \right\|_{L^\infty(\mathbb R)} + q_0\right)^{2p}
\Big(\left\| \Phi(t) \right\|_{H^1(\mathbb R)} + \no{\zeta'}_{L^2(\mathbb R)} 
\nn\\
&\quad
+  \big\| |\zeta| - q_0 \big\|_{L^2(\mathbb R)}\Big)
+
2\sqrt 2 K \left(\no{\Phi(t)}_{H^1(\mathbb R)} +\no{\zeta}_{L^\infty(\mathbb R)} + q_0 \right)^{2p-1}
\left(\no{\Phi(t)}_{H^1(\mathbb R)} + \no{|\zeta|-q_0}_{L^\infty(\mathbb R)} \right)
\nn\\
&\quad 
\cdot\left(\no{\Phi(t)}_{H^1(\mathbb R)} + \no{\zeta'}_{L^2(\mathbb R)}\right)
\nn\\
&\leq
6\sqrt 2 K   \left( \left\| \Phi(t) \right\|_{H^1(\mathbb R)} + \left\| \zeta \right\|_{L^\infty(\mathbb R)} + \no{\zeta'}_{L^2(\mathbb R)} + \big\| |\zeta| - q_0 \big\|_{L^2(\mathbb R)} + q_0\right)^{2p+1}
\nn\\
&\leq
3 \cdot 2^{2p+\frac 32} K   \left[ \left\| \Phi(t) \right\|_{H^1(\mathbb R)}^{2p+1} + \left(\left\| \zeta \right\|_{L^\infty(\mathbb R)} + \no{\zeta'}_{L^2(\mathbb R)} + \big\| |\zeta| - q_0 \big\|_{L^2(\mathbb R)} + q_0\right)^{2p+1} \right],
\end{align}
where for the last step we have used the inequality $(a+b)^\sigma \leq 2^{\sigma-1} (a^\sigma + b^\sigma)$, $a, b, \sigma \geq 1$, which follows from applying Jensen's inequality to the convex function $x^\sigma$, $\sigma\geq 1$.

By combining the bound \eqref{GF-H1-est-F} with inequality \eqref{mf1}, we obtain 
\begin{equation}\label{L-H1-est-F}
\begin{aligned}
\no{\Lambda[\Phi](t)}_{H^1(\mathbb R)}
&\leq
\sqrt 2\no{\Phi_0}_{H^1(\mathbb R)}
+
4 \big(\sqrt{|\nu|} \sqrt t+ \sqrt 2\big) \no{\zeta'}_{L^2(\mathbb R)}
\\
&\quad
+
3 \cdot 2^{2p+\frac 52} |\gamma| K \cdot t \left(\left\| \zeta \right\|_{L^\infty(\mathbb R)} + q_0 + \no{\zeta'}_{L^2(\mathbb R)} + \big\| |\zeta| - q_0 \big\|_{L^2(\mathbb R)}\right)^{2p+1}
\\
&\quad
+
3 \cdot 2^{2p+\frac 52} |\gamma| K \cdot t \sup_{\tau\in[0, t]}  \left\| \Phi(\tau) \right\|_{H^1(\mathbb R)}^{2p+1}.
\end{aligned}
\end{equation}
Motivated \eqref{L-H1-est-F}, for fixed $T>0$ we let
\begin{equation}\label{rho-def}
\begin{aligned}
\rho 
&= 
2 \cdot 3 \cdot 2^{2p+\frac 52} |\gamma| K
\bigg[
\no{\Phi_0}_{H^1(\mathbb R)}
+
\big(\sqrt{|\nu|}\sqrt T + \sqrt 2\big) \no{\zeta'}_{L^2(\mathbb R)}
\\
&\quad
+
T \left(\left\| \zeta \right\|_{L^\infty(\mathbb R)} + q_0 + \no{\zeta'}_{L^2(\mathbb R)} + \big\| |\zeta| - q_0 \big\|_{L^2(\mathbb R)}\right)^{2p+1}
\bigg] > 0
\end{aligned}
\end{equation}
and denote by $B(0, \rho)$ the open ball of radius $\rho$ centered at zero in $C([0, T_f]; H^1(\mathbb R))$, $0<T_f\leq T$. Then, in order for the map  $\Phi \mapsto \Lambda[\Phi]$ to be a contraction in $B(0, \rho)$, it must be that $\Lambda[\Phi] \in B(0, \rho)$ whenever $\Phi \in B(0, \rho)$. In view of estimate \eqref{L-H1-est-F}, this first requirement yields the following condition for the lifespan $T_f$:
\begin{equation}\label{tf-1}
\frac{\rho}{2}
+
3 \cdot 2^{2p+\frac 52} |\gamma| K T_f  \rho^{2p+1}
\leq \rho
\ \Rightarrow \
T_f \leq \frac{1}{3 \cdot 2^{2p+\frac 72} |\gamma| K \rho^{2p}}.
\end{equation}

The second requirement which alongside \eqref{tf-1} guarantees that $\Phi \mapsto \Lambda[\Phi]$ is a contraction in $B(0, \rho)$ is that $\no{\Lambda[\Phi] - \Lambda[\Psi]}_{C([0, T_f]; H^1(\mathbb R))} \leq M \no{\Phi - \Psi}_{C([0, T_f]; H^1(\mathbb R))}$ for all $\Phi, \Psi \in B(0, \rho)$ and some constant $M < 1$. Similarly to \eqref{mf1}, we have
\begin{equation}\label{L-contr}
\no{\Lambda[\Phi](t) - \Lambda[\Psi](t)}_{H^1(\mathbb R)}
\leq
2|\gamma| \int_0^t \no{G_{F, \Phi}(\tau)-G_{F, \Psi}(\tau)}_{H^1(\mathbb R)} d\tau.
\end{equation}
In order to estimate the $H^1$ norm on the right-hand side, we first manipulate the difference of nonlinearities $G_{F, \Phi}-G_{F, \Psi}$ so that a convenient factor of $\Phi-\Psi$ can be extracted. Noting that
\begin{equation}\label{g-diff-F}
G_{F, \Phi}  - G_{F, \Psi} 
=
\left[ F(\left|\Phi+\zeta\right|^2) - F(\left|\Psi+\zeta\right|^2) \right] (\Phi+\zeta)
+
\left[ F(\left|\Psi+\zeta\right|^2)-F(q_0^2)\right] (\Phi-\Psi)
\end{equation}
and recalling the conditions \eqref{F-prop} for $F$, we have
\begin{equation*}
\begin{aligned}
\left|G_{F, \Phi}  - G_{F, \Psi} \right| 
&\leq
K 
\left\{
\left(\left|\Phi+\zeta\right|^{2(p-1)} + \left|\Psi+\zeta\right|^{2(p-1)} \right) 
\left| 
\left|\Phi+\zeta\right|^2 - \left|\Psi+\zeta\right|^2
\right| 
\left|\Phi+\zeta\right|
\right.
\\
&\hskip 1.2cm
\left.+
\left(\left|\Psi+\zeta\right|^{2p} + q_0^{2p} \right) \left|\Phi-\Psi\right|
\right\}.
\end{aligned}
\end{equation*}
Thus, observing that
\begin{equation}\label{sq-diff}
\begin{aligned}
\left|\left|\Phi+\zeta\right|^2 - \left|\Psi+\zeta\right|^2\right|
&=
\big(\left|\Phi+\zeta\right|+\left|\Psi+\zeta\right|\big)
\big| \left|\Phi+\zeta\right| - \left|\Psi+\zeta\right| \big|
\\
&\leq
\big(\left|\Phi+\zeta\right|+\left|\Psi+\zeta\right|\big)
\left|\Phi-\Psi\right|,
\end{aligned}
\end{equation}
we deduce
\begin{equation*}
\begin{aligned}
\left|G_{F, \Phi}(x, t) - G_{F, \Psi} (x, t)\right| 
&\leq
K 
\left\{
\left(\left|\Phi+\zeta\right|^{2(p-1)} + \left|\Psi+\zeta\right|^{2(p-1)} \right) 
\left(
\left|\Phi+\zeta\right| + \left|\Psi+\zeta\right| 
\right)
\left|\Phi+\zeta\right|
\right.\\
&\hskip 1.2cm
\left.+
\left|\Psi+\zeta\right|^{2p} + q_0^{2p}  
\right\}
 \left|\Phi-\Psi\right|.
\end{aligned}
\end{equation*}
In turn, we can estimate the difference of nonlinearities in $L^2$ as follows: 
\begin{align}
\no{G_{F, \Phi} - G_{F, \Psi}}_{L^2(\mathbb R)}
&\leq
K
\left[
\left(\left\|\Phi+\zeta\right\|_{L^\infty(\mathbb R)}^{2(p-1)} + \left\|\Psi+\zeta\right\|_{L^\infty(\mathbb R)}^{2(p-1)} \right) 
\left(
\left\|\Phi+\zeta\right\|_{L^\infty(\mathbb R)} + \left\|\Psi+\zeta\right\|_{L^\infty(\mathbb R)}
\right)
\right.
\nn\\
&\quad
\left.\cdot
\left\|\Phi+\zeta\right\|_{L^\infty(\mathbb R)}
+
\left\|\Psi+\zeta\right\|_{L^\infty(\mathbb R)}^{2p} + q_0^{2p} 
\right]
\left\| \Phi - \Psi \right\|_{L^2(\mathbb R)}
\nn\\
&\leq
K
\left[
\left(
\left\|\Phi+\zeta\right\|_{L^\infty(\mathbb R)} + \left\|\Psi+\zeta\right\|_{L^\infty(\mathbb R)}
\right)^{2p} + q_0^{2p} 
\right]
\left\| \Phi - \Psi \right\|_{L^2(\mathbb R)}
\nn\\
&\leq K \left(\no{\Phi}_{L^\infty(\mathbb R)} + \no{\Psi}_{L^\infty(\mathbb R)} + 2\no{\zeta}_{L^\infty(\mathbb R)} + q_0\right)^{2p}
\left\| \Phi - \Psi \right\|_{L^2(\mathbb R)}
\nn\\
&\leq 
K \left(\no{\Phi}_{H^1(\mathbb R)} + \no{\Psi}_{H^1(\mathbb R)} + 2\no{\zeta}_{L^\infty(\mathbb R)} + q_0\right)^{2p}
\left\| \Phi - \Psi \right\|_{L^2(\mathbb R)}
\label{Gf-diff-l2-F}
\end{align}
with the last step due to Sobolev embedding.

It remains to also estimate the $L^2$ norm of the derivative $\p_x \left(G_{F, \Phi} - G_{F, \Psi}\right)$. This task turns out to be more involved due to the appearance of the difference $F'(|\Phi+\zeta|^2)-F'(|\Psi+\zeta|^2)$. First, differentiating~\eqref{g-diff-F} and rearranging appropriately, we have
\begin{align}\label{mf2}
\p_x \left(G_{F, \Phi} - G_{F, \Psi}\right)
&=
\Big\{ F'(\left|\Phi+\zeta\right|^2) \left[\left(\Phi+\zeta\right)_x \left(\overline{\Phi+\zeta}\right) + \left(\Phi+\zeta\right) \left(\overline{\Phi+\zeta}\right)_x\right] 
\nn\\
&\qquad
- F'(\left|\Psi+\zeta\right|^2) \left[\left(\Psi+\zeta\right)_x \left(\overline{\Psi+\zeta}\right) + \left(\Psi+\zeta\right) \left(\overline{\Psi+\zeta}\right)_x\right] \Big\} (\Phi+\zeta)
\nn\\
&\quad
+ \left[ F(\left|\Phi+\zeta\right|^2) - F(\left|\Psi+\zeta\right|^2) \right] (\Phi +\zeta)_x
\nn\\
&\quad
+
F'(\left|\Psi+\zeta\right|^2) 
\left[\left(\Psi+\zeta\right)_x \left(\overline{\Psi+\zeta}\right) + \left(\Psi+\zeta\right) \left(\overline{\Psi+\zeta}\right)_x\right] (\Phi-\Psi)
\nn\\
&\quad
+
\left[F(\left|\Psi+\zeta\right|^2) - F(q_0^2)\right] (\Phi-\Psi)_x.
\end{align}
Therefore, using also the conditions \eqref{F-prop}, the embedding $H^1(\mathbb R) \subset L^\infty(\mathbb R)$ and inequality \eqref{sq-diff}, we obtain
\begin{align}\label{g'-temp}
&\quad \no{\p_x \left(G_{F, \Phi} - G_{F, \Psi}\right)}_{L^2(\mathbb R)}
\leq
\Big\| F'(\left|\Phi+\zeta\right|^2) \left[\left(\Phi+\zeta\right)_x \left(\overline{\Phi+\zeta}\right) + \left(\Phi+\zeta\right) \left(\overline{\Phi+\zeta}\right)_x\right] 
\nn\\
&\qquad
- F'(\left|\Psi+\zeta\right|^2) \left[\left(\Psi+\zeta\right)_x \left(\overline{\Psi+\zeta}\right) + \left(\Psi+\zeta\right) \left(\overline{\Psi+\zeta}\right)_x\right] \Big\|_{L^2(\mathbb R)}
\left(\no{\Phi}_{H^1(\mathbb R)} + \no{\zeta}_{L^\infty(\mathbb R)} \right)
\nn\\
&\quad
+ 2^{2p-1} K \left(\no{\Phi}_{H^1(\mathbb R)} + \no{\Psi}_{H^1(\mathbb R)} + \no{\zeta}_{L^\infty(\mathbb R)}\right)^{2p-1}
\left(\no{\Phi}_{H^1(\mathbb R)} + \no{\zeta'}_{L^2(\mathbb R)}\right)
\no{\Phi-\Psi}_{H^1(\mathbb R)}
\nn\\
&\quad
+
2 K \left(\no{\Psi}_{H^1(\mathbb R)} + \no{\zeta}_{L^\infty(\mathbb R)}\right)^{2p-1}
\left(\no{\Psi}_{H^1(\mathbb R)} + \no{\zeta'}_{L^2(\mathbb R)} \right)
\no{\Phi-\Psi}_{H^1(\mathbb R)}
\nn\\
&\quad
+
K \left(  \no{\Psi}_{H^1(\mathbb R)} + \no{\zeta}_{L^\infty(\mathbb R)} + q_0\right)^{2p}\no{\Phi-\Psi}_{H^1(\mathbb R)}.
\end{align}
The difference $\no{\Phi-\Psi}_{H^1(\mathbb R)}$ has conveniently appeared in the last three terms on the right-hand side of the above inequality. In order to also extract it from the first term, we use the identity 
$$
u_1 v_1 \overline w_1 - u_2 v_2 \overline w_2
=
u_1 v_1 \left(\overline{w_1-w_2}\right) + u_1 \overline w_2 \left(v_1-v_2\right)   + v_2 \overline w_2 \left(u_1-u_2\right)
$$
to express the first of the differences involved in the first term of \eqref{g'-temp} as 
\begin{align*}
&\quad
F'(\left|\Phi+\zeta\right|^2) \left(\Phi+\zeta\right)_x \left(\overline{\Phi+\zeta}\right) 
- F'(\left|\Psi+\zeta\right|^2)  \left(\Psi+\zeta\right)_x \left(\overline{\Psi+\zeta}\right) 
\nn\\
&=
F'(\left|\Phi+\zeta\right|^2) \left(\Phi+\zeta\right)_x \left(\overline{\Phi-\Psi}\right)
+
F'(\left|\Phi+\zeta\right|^2) \left(\overline{\Psi+\zeta}\right) \left(\Phi-\Psi\right)_x
\nn\\
&\quad
+
\left(\Psi+\zeta\right)_x \left(\overline{\Psi+\zeta}\right) \left[F'(\left|\Phi+\zeta\right|^2) - F'(\left|\Psi+\zeta\right|^2)\right].
\end{align*}
In view of this writing and also the properties \eqref{F-prop} satisfied by $F$, we find
\begin{align}\label{F'-diff-est}
&\quad
\no{F'(\left|\Phi+\zeta\right|^2) \left(\Phi+\zeta\right)_x \left(\overline{\Phi+\zeta}\right) 
- F'(\left|\Psi+\zeta\right|^2)  \left(\Psi+\zeta\right)_x \left(\overline{\Psi+\zeta}\right)}_{L^2(\mathbb R)}
\nn\\
&\leq
\no{F'(\left|\Phi+\zeta\right|^2)}_{L^\infty(\mathbb R)} \no{\Phi_x+\zeta'}_{L^2(\mathbb R)} \no{\Phi-\Psi}_{L^\infty(\mathbb R)}
+
\no{F'(\left|\Phi+\zeta\right|^2)}_{L^\infty(\mathbb R)} \no{\Psi+\zeta}_{L^\infty(\mathbb R)} \nn\\
&\quad
\cdot \no{\left(\Phi-\Psi\right)_x}_{L^2(\mathbb R)}
+
\no{\left(\Psi+\zeta\right)_x}_{L^2(\mathbb R)} \no{\Psi+\zeta}_{L^\infty(\mathbb R)} 
\no{F'(\left|\Phi+\zeta\right|^2) - F'(\left|\Psi+\zeta\right|^2)}_{L^\infty(\mathbb R)}
\nn\\
&\leq
K \left(\no{\Phi}_{H^1(\mathbb R)} + \no{\zeta}_{L^\infty(\mathbb R)}\right)^{2(p-1)} 
\left(
\no{\Phi}_{H^1(\mathbb R)} +
\no{\Psi}_{H^1(\mathbb R)} + \no{\zeta'}_{L^2(\mathbb R)}  + \no{\zeta}_{L^\infty(\mathbb R)}
\right) 
\no{\Phi-\Psi}_{H^1(\mathbb R)}
\nn\\
&\quad
+
\left(\no{\Psi}_{H^1(\mathbb R)} + \no{\zeta'}_{L^2(\mathbb R)}\right)
\left(\no{\Psi}_{H^1(\mathbb R)} + \no{\zeta}_{L^\infty(\mathbb R)}\right)
\no{F'(\left|\Phi+\zeta\right|^2) - F'(\left|\Psi+\zeta\right|^2)}_{L^\infty(\mathbb R)}.
\end{align}
Furthermore, by symmetry with respect to complex conjugation, the exact same bound is also satisfied by the $L^2$ norm of the second difference in the first term of \eqref{g'-temp}, i.e. 
\begin{equation}\label{F'-2est}
\no{F'(\left|\Phi+\zeta\right|^2) \left(\Phi+\zeta\right) \left(\overline{\Phi+\zeta}\right)_x
- F'(\left|\Psi+\zeta\right|^2)  \left(\Psi+\zeta\right) \left(\overline{\Psi+\zeta}\right)_x}_{L^2(\mathbb R)}
\leq 
\text{RHS}_{\eqref{F'-diff-est}}.
\end{equation}

Thanks to the bounds \eqref{F'-diff-est} and \eqref{F'-2est}, it suffices to extract the $H^1$ norm of $\Phi-\Psi$ from the norm
\begin{equation}\label{chal}
\no{F'(\left|\Phi+\zeta\right|^2) - F'(\left|\Psi+\zeta\right|^2)}_{L^\infty(\mathbb R)}.
\end{equation}
This norm cannot be handled as easily as the $L^\infty$ norm of $F(\left|\Phi+\zeta\right|^2) - F(\left|\Psi+\zeta\right|^2)$ that arose earlier in the $L^2$ estimation of \eqref{g-diff-F}.
Indeed, note in particular that the first of the conditions~\eqref{F-prop} is not necessarily satisfied by $F'$. For example, in the case of power nonlinearity $F(x) = x^p$ we have $F'(x) = px^{p-1}$ and, unless $p\geq \frac 32$, it is not possible to bound the difference $|F'(x)-F'(y)|$ by a factor of $\left|\sqrt x - \sqrt y\right|$ (let alone $|x-y|$), which is the minimum requirement in order to make $\no{\Phi-\Psi}_{L^\infty(\mathbb R)}$ appear in the bound for \eqref{chal}.
On the other hand, in the case of the saturable nonlinearity  $F(x) = \frac{x}{\kappa(1+x)}$ we have  
\begin{equation}\label{F'-sat}
\begin{aligned}
\left| F'(x) - F'(y) \right| 
&= \frac{1}{|\kappa|} \left| \frac{1}{(1+x)^2} - \frac{1}{(1+y)^2} \right|
=
\frac{\left|2+x+y\right|}{|\kappa| (1+x)^2 (1+y)^2} \left|x-y\right|
\\
&\lesssim
\frac{1}{|\kappa|} \left(\frac{1}{1+x^3} + \frac{1}{1+y^3}\right) \left|x-y\right|
\leq
\frac{2}{|\kappa|} \left|x-y\right|,
\quad x, y \geq 0,
\end{aligned}
\end{equation}
and so it seems possible to control \eqref{chal} by $\no{\Phi-\Psi}_{L^\infty(\mathbb R)}$.

Following the above observations, we revisit the estimation of $\no{\p_x\left(G_{F, \Phi} - G_{F, \Psi}\right)}_{L^2(\mathbb R)}$ by treating the two cases $F(x) = \frac{x}{\kappa (1+x)}$ (saturable nonlinearity) and $F(x) = x^p$ (power nonlinearity) separately. 
\\[3mm]
\noindent
\textit{Contraction estimate for the saturable nonlinearity.}
In this case, $F(x) = \frac{x}{\kappa(1+x)}$ and $F'(x) = \frac{1}{\kappa (1+x)^2}$ are uniformly bounded by $\frac{1}{|\kappa|}$ for $x\geq 0$. In addition, the first of the conditions \eqref{F-prop} is met with $K = \frac{1}{|\kappa|}$ and $p=1$. Thus, returning to \eqref{mf2} and using these uniform bounds along with the embedding $H^1(\mathbb R) \subset L^\infty(\mathbb R)$ and inequality \eqref{sq-diff} as appropriate, we have
\begin{align}\label{temp14}
&\quad
\no{\p_x \left(G_{F, \Phi} - G_{F, \Psi}\right)}_{L^2(\mathbb R)}
\\
&\leq
\Big\| F'(\left|\Phi+\zeta\right|^2) \left[\left(\Phi+\zeta\right)_x \left(\overline{\Phi+\zeta}\right) + \left(\Phi+\zeta\right) \left(\overline{\Phi+\zeta}\right)_x\right] 
\nn\\
&\qquad
- F'(\left|\Psi+\zeta\right|^2) \left[\left(\Psi+\zeta\right)_x \left(\overline{\Psi+\zeta}\right) + \left(\Psi+\zeta\right) \left(\overline{\Psi+\zeta}\right)_x\right] \Big\|_{L^2(\mathbb R)} 
\left(\no{\Phi}_{H^1(\mathbb R)} + \no{\zeta}_{L^\infty(\mathbb R)} \right)
\nn\\
&\quad
+
\frac{2}{|\kappa|} \left(\no{\Phi}_{H^1(\mathbb R)} + \no{\Psi}_{H^1(\mathbb R)} + \no{\zeta}_{L^\infty(\mathbb R)} \right)
\no{\Phi-\Psi}_{H^1(\mathbb R)} \left(\no{\Phi}_{H^1(\mathbb R)} + \no{\zeta'}_{L^2(\mathbb R)}\right)
\nn\\
&\quad
+
\frac{2}{|\kappa|} 
\left(\no{\Psi}_{H^1(\mathbb R)} + \no{\zeta}_{L^\infty(\mathbb R)} \right)
\left(\no{\Psi}_{H^1(\mathbb R)} + \no{\zeta'}_{L^2(\mathbb R)} \right)
\no{\Phi-\Psi}_{H^1(\mathbb R)}
+
\frac{2}{|\kappa|}  \no{\Phi-\Psi}_{H^1(\mathbb R)}.
\nn
\end{align}
Notice that the bound \eqref{F'-sat} together with inequality \eqref{sq-diff} imply
\begin{equation*}
\left|F'(\left|\Phi+\zeta\right|^2) - F'(\left|\Psi+\zeta\right|^2)\right|
\leq
\frac{2}{|\kappa|} 
\big(\left|\Phi+\zeta\right| + \left|\Psi+\zeta\right|\big)
\left|\Phi-\Psi\right|.
\end{equation*}
Hence, in the saturable case, the norm \eqref{chal} satisfies the inequality
\begin{equation*}
\no{F'(\left|\Phi+\zeta\right|^2) - F'(\left|\Psi+\zeta\right|^2)}_{L^\infty(\mathbb R)}
\leq
\frac{4}{|\kappa|} \left(\no{\Phi}_{H^1(\mathbb R)} + \no{\Psi}_{H^1(\mathbb R)} + \no{\zeta}_{L^\infty(\mathbb R)}\right) \no{\Phi-\Psi}_{H^1(\mathbb R)},
\end{equation*}
which can be combined with \eqref{F'-diff-est} and the fact that $p=1$ and $K=\frac{1}{|\kappa|}$ to deduce
\begin{equation*}
\begin{aligned}
&\quad
\no{F'(\left|\Phi+\zeta\right|^2) \left(\Phi+\zeta\right)_x \left(\overline{\Phi+\zeta}\right) 
- F'(\left|\Psi+\zeta\right|^2)  \left(\Psi+\zeta\right)_x \left(\overline{\Psi+\zeta}\right)}_{L^2(\mathbb R)}
\\
&\leq
\frac{1}{|\kappa|}
\bigg[
\left(
\no{\Phi}_{H^1(\mathbb R)} +
\no{\Psi}_{H^1(\mathbb R)} + \no{\zeta'}_{L^2(\mathbb R)}  + \no{\zeta}_{L^\infty(\mathbb R)}
\right) 
+
4\left(\no{\Psi}_{H^1(\mathbb R)} + \no{\zeta'}_{L^2(\mathbb R)}\right)
\\
&\hskip 1.2cm
\cdot \left(\no{\Psi}_{H^1(\mathbb R)} + \no{\zeta}_{L^\infty(\mathbb R)}\right)
 \left(\no{\Phi}_{H^1(\mathbb R)} + \no{\Psi}_{H^1(\mathbb R)} + \no{\zeta}_{L^\infty(\mathbb R)}\right)
\bigg]\no{\Phi-\Psi}_{H^1(\mathbb R)}.
\end{aligned}
\end{equation*}
Moreover, by symmetry, the same estimate also holds for the left-hand side of \eqref{F'-2est}. Therefore, returning to \eqref{temp14}, we eventually obtain
\begin{equation}\label{temp15}
\begin{aligned}
&\quad
\no{\p_x \left(G_{F, \Phi} - G_{F, \Psi}\right)}_{L^2(\mathbb R)}
\\
&\leq
\frac{8}{|\kappa|}
\left(1+\no{\Phi}_{H^1(\mathbb R)} +
\no{\Psi}_{H^1(\mathbb R)} + \no{\zeta'}_{L^2(\mathbb R)}  + \no{\zeta}_{L^\infty(\mathbb R)}
\right)^4 
\no{\Phi-\Psi}_{H^1(\mathbb R)}.
\end{aligned}
\end{equation}

The $L^2$ estimates \eqref{Gf-diff-l2-F} and \eqref{temp15} combine to yield the $H^1$ estimate
\begin{equation*}
\no{G_{F, \Phi}-G_{F, \Psi}}_{H^1(\mathbb R)}
\leq
\frac{16}{|\kappa|} 
\left(1+
\no{\Phi}_{H^1(\mathbb R)} +
\no{\Psi}_{H^1(\mathbb R)} + \no{\zeta'}_{L^2(\mathbb R)}  + \no{\zeta}_{L^\infty(\mathbb R)} + q_0
\right)^4  \no{\Phi-\Psi}_{H^1(\mathbb R)}
\end{equation*}
which in turn implies, via inequality \eqref{L-contr},
\begin{align*}
\no{\Lambda[\Phi](t) - \Lambda[\Psi](t)}_{H^1(\mathbb R)}
&\leq
\frac{32 |\gamma|}{|\kappa|} \, t \sup_{\tau \in [0, t]} 
\Big(1+
\no{\Phi(\tau)}_{H^1(\mathbb R)} +
\no{\Psi(\tau)}_{H^1(\mathbb R)} + \no{\zeta'}_{L^2(\mathbb R)}  
\\
&\hskip 3.05cm
+ \no{\zeta}_{L^\infty(\mathbb R)} + q_0
\Big)^4
\no{(\Phi-\Psi)(\tau)}_{H^1(\mathbb R)}.
\nn
\end{align*}
Therefore, a sufficient condition for the map $\Phi \mapsto \Lambda[\Phi]$ to be a contraction in the ball $B(0, \rho) \subset C([0, T_f]; H^1(\mathbb R))$ is  
\begin{equation}\label{tf-sat}
 T_f   < 
\frac{|\kappa|}{32 |\gamma| \left(1+ 2\rho + \no{\zeta'}_{L^2(\mathbb R)}  + \no{\zeta}_{L^\infty(\mathbb R)} + q_0\right)^4}.
\end{equation}
Observe that \eqref{tf-sat} is a stronger condition than \eqref{tf-1}. 
\vskip 3mm
\noindent
\textit{Contraction estimate for the power nonlinearity.}
By the triangle inequality, 
\begin{equation}\label{temp9}
\no{\p_x \left(G_{F, \Phi}-G_{F, \Psi}\right)}_{L^2(\mathbb R)}
\leq
\no{\p_x \left[\left|\Phi+\zeta\right|^{2p} (\Phi+\zeta) -  \left|\Psi+\zeta\right|^{2p}(\Psi+\zeta) \right]}_{L^2(\mathbb R)}
+ |q_0|^{2p} \no{\Phi-\Psi}_{H^1(\mathbb R)}.
\end{equation}
We then invoke the following widely used result, whose proof we include below for completeness.  
\begin{lemma}\label{mvt-l}
For any $p\geq 1$ and any pair of complex numbers $z, z'$,
\begin{equation}
|z|^{2p} z - |z'|^{2p} z' 
=
(p+1) \left(\int_0^1  \left|Z_\lambda\right|^{2p} d\lambda\right) \left(z-z'\right)
+
p \left(\int_0^1 \left|Z_\lambda\right|^{2(p-1)}
Z_\lambda^2 \, d\lambda\right)  \overline{\left(z-z'\right)},
\end{equation}
where $Z_\lambda := \lambda z + \left(1-\lambda\right) z'$, $\lambda\in [0, 1]$.
\end{lemma}

\begin{proof}[Proof of Lemma \ref{mvt-l}]
Letting $z=x+iy$, $x, y\in\mathbb R$, we have
$|z|^{2p}z = \left(x^2 + y^2\right)^p  x 
 + i \left(x^2 + y^2\right)^p y$.
For $f(x, y) = \left(x^2 + y^2\right)^p  x$, define the function $g: [0, 1] \to \mathbb R$ by 
$g(\lambda) = f(\lambda x + (1-\lambda) x', \lambda y + (1-\lambda) y')$.
Then, by the Fundamental Theorem of Calculus,
$$
\left(x^2 + y^2\right)^p  x
-
\left(x'^2 + y'^2\right)^p  x'
\equiv
f(x, y) - f(x', y') \equiv g(1) - g(0) = \int_0^1 \frac{dg}{d\lambda} \, d\lambda.
$$
Furthermore, by the chain rule, for each $(x, y), (x',y') \in \mathbb R^2$ we have
$$
\frac{dg}{d\lambda}
=
\left(x-x'\right)\frac{\p f}{\p x}(\lambda x + (1-\lambda) x', \lambda y + (1-\lambda) y') 
+
\left(y-y'\right) \frac{\p f}{\p y}(\lambda x + (1-\lambda) x', \lambda y + (1-\lambda) y').
$$
Thus, computing 
$\frac{\p f}{\p x}(x, y) = 
2p \left(x^2+y^2\right)^{p-1} x^2 + \left(x^2+y^2\right)^p$,
$\frac{\p f}{\p y}(x, y) = 2p \left(x^2+y^2\right)^{p-1} xy$
and letting $X_\lambda = \lambda x + \left(1-\lambda\right) x'$, $Y_\lambda = \lambda y + \left(1-\lambda\right) y'$ so that $X_\lambda + i Y_\lambda = \lambda z + (1-\lambda) z' =: Z_\lambda$, we obtain
\begin{equation*}
\begin{aligned}
\left(x^2 + y^2\right)^p  x
-
\left(x'^2 + y'^2\right)^p  x'
&=
\left(\int_0^1 \left[
2p \left|Z_\lambda\right|^{2(p-1)}  X_\lambda^2
+
\left|Z_\lambda\right|^{2p}  
\right] d\lambda
\right)
\left(x-x'\right)
\\
&\quad
+
\left(
\int_0^1 
2p 
\left|Z_\lambda\right|^{2(p-1)} X_\lambda Y_\lambda
d\lambda
\right) \left(y-y'\right).
\end{aligned}
\end{equation*}
In addition, the symmetry in $x$ and $y$ readily implies
\begin{align*}
\left(x^2 + y^2\right)^p  y
-
\left(x'^2 + y'^2\right)^p  y'
&=
\left(\int_0^1 \left[
2p 
\left|Z_\lambda\right|^{2(p-1)}  Y_\lambda^2
+
\left|Z_\lambda\right|^{2p}  
\right] d\lambda
\right)
\left(y-y'\right)
\nn\\
&\quad
+
\left(\int_0^1 
2p 
\left|Z_\lambda\right|^{2(p-1)} X_\lambda Y_\lambda
d\lambda
\right) \left(x-x'\right).
\end{align*}
The last two equations can be combined to yield
\begin{equation*}
\begin{aligned}
|z|^{2p} z - |z'|^{2p} z'
=
\left(\int_0^1 \left|Z_\lambda\right|^{2p} d\lambda\right) \left(z-z'\right)
+
2p \int_0^1 &\left|Z_\lambda\right|^{2(p-1)}
\Big[
X_\lambda^2 \left(x-x'\right) + X_\lambda Y_\lambda \left(y-y'\right) 
\\
&
+ iY_\lambda \left(y-y'\right) + i X_\lambda Y_\lambda \left(x-x'\right)
\Big] d\lambda,
\end{aligned}
\end{equation*}
which is the desired expression since, upon completing the square, the quantity inside the square bracket on the right-hand side is equal to $\frac 12  Z_\lambda^2 \, \overline{\left(z-z'\right)} + \frac 12 \left|Z_\lambda\right|^2 \left(z-z'\right)$.
\end{proof}

Returning to \eqref{temp9}, we employ Lemma \ref{mvt-l} with $z=\Phi+\zeta$ and $z'=\Psi+\zeta$, which imply 
$
Z_\lambda = \lambda (\Phi+\zeta) + (1-\lambda) (\Psi+\zeta)
= \Lambda[\Phi] + (1-\lambda) \Psi + \zeta,
$ 
to obtain
\begin{align*}
&\quad
\left\| \p_x \left[ \left|\Phi+\zeta\right|^{2p} \left(\Phi+\zeta\right) - \left|\Psi+\zeta\right|^{2p} \left(\Psi+\zeta\right) \right]  \right\|_{L^2(\mathbb R)}
\nn\\
&\leq
(p+1)
\left\| \p_x \left[ \left(\int_0^1  \left|Z_\lambda\right|^{2p} d\lambda\right) \left(\Phi - \Psi\right)\right]\right\|_{L^2(\mathbb R)}
+
p
\left\| \p_x\left[ \left(\int_0^1 \left|Z_\lambda\right|^{2(p-1)}
Z_\lambda^2 \, d\lambda\right) \overline{\left(\Phi - \Psi\right)}\right]\right\|_{L^2(\mathbb R)}
\nn\\
&\leq
(p+1) \int_0^1 
\left\| \p_x \left[  \left|Z_\lambda\right|^{2p} \left(\Phi - \Psi\right)\right]  \right\|_{L^2(\mathbb R)} d\lambda
+
p \int_0^1 \left\|  \p_x\left[ \left|Z_\lambda\right|^{2(p-1)}
Z_\lambda^2 \ \overline{\left(\Phi-\Psi\right)}\right]  \right\|_{L^2(\mathbb R)} d\lambda
\nn\\
&\leq
(p+1) \sup_{\lambda\in [0, 1]} \left\| \p_x \left[  \left|Z_\lambda\right|^{2p} \left(\Phi-\Psi\right)\right]  \right\|_{L^2(\mathbb R)}  
+
p \sup_{\lambda\in [0, 1]}  \left\|  \p_x\left[ \left|Z_\lambda\right|^{2(p-1)}
Z_\lambda^2 \ \overline{\left(\Phi-\Psi\right)}\right]  \right\|_{L^2(\mathbb R)}.
\end{align*}
Thus, by the product rule,
\begin{equation}\label{diff-temp5}
\begin{aligned}
&\quad
\left\| \p_x \left[ \left|\Phi+\zeta\right|^{2p} \left(\Phi+\zeta\right) - \left|\Psi+\zeta\right|^{2p} \left(\Psi+\zeta\right) \right]  \right\|_{L^2(\mathbb R)}
\\
&\leq
(p+1) \sup_{\lambda\in [0, 1]} 
\left(
\left\| \p_x \big(\left|Z_\lambda\right|^{2p}\big) \cdot \left(\Phi-\Psi\right) \right\|_{L^2(\mathbb R)} 
+
\left\| \left|Z_\lambda\right|^{2p} \left(\Phi-\Psi\right)_x \right\|_{L^2(\mathbb R)} 
\right)
\\
&\quad
+
p \sup_{\lambda\in [0, 1]} 
\left(
\left\| \p_x\big( \left|Z_\lambda\right|^{2(p-1)}
Z_\lambda^2 \big) \cdot \overline{\left(\Phi-\Psi\right)} \right\|_{L^2(\mathbb R)}
+
\left\| \left|Z_\lambda\right|^{2(p-1)}
Z_\lambda^2 \   \overline{\left(\Phi-\Psi\right)_x} \right\|_{L^2(\mathbb R)}
\right).
\end{aligned}
\end{equation}
The second and fourth norms on the right-hand side of \eqref{diff-temp5} are easy to handle as follows:
\begin{equation}
\begin{aligned}
\left\| \left|Z_\lambda\right|^{2p} \left(\Phi-\Psi\right)_x \right\|_{L^2(\mathbb R)} 
&=
\left\| \left|Z_\lambda\right|^{2(p-1)}
Z_\lambda^2 \   \overline{\left(\Phi-\Psi\right)_x} \right\|_{L^2(\mathbb R)}
\\
&\leq
\no{Z_\lambda}_{L^\infty(\mathbb R)}^{2p}
\no{\left(\Phi-\Psi\right)_x}_{L^2(\mathbb R)} 
\\
&\leq
\left(\no{\Phi}_{H^1(\mathbb R)} + \no{\Psi}_{H^1(\mathbb R)} + \no{\zeta}_{L^\infty(\mathbb R)}\right)^{2p}
\no{\Phi-\Psi}_{H^1(\mathbb R)}.
\label{temp3}
\end{aligned}
\end{equation}
In order to estimate the first and third norms on the right-hand side of \eqref{diff-temp5}, we note that
$$
\begin{aligned}
&\p_x |Z_\lambda|^{2p}
=
2p  |Z_\lambda|^{2p-1} \p_x |Z_\lambda|
=
2p |Z_\lambda|^{2p-2} \left(\p_x |Z_\lambda|\right) |Z_\lambda|,
\\
&\p_x |Z_\lambda|^{2(p-1)}
=
2(p-1)  |Z_\lambda|^{2p-3} \p_x |Z_\lambda|
=
2(p-1) |Z_\lambda|^{2(p-2)} \left(\p_x |Z_\lambda|\right) |Z_\lambda|.
\end{aligned}
$$
Hence, it suffices to compute $\left(\p_x |Z_\lambda|\right) |Z_\lambda|$. For this, we observe that $\p_x |Z_\lambda|^2 = 2 |Z_\lambda| \p_x |Z_\lambda|$ and so we first compute
\begin{equation*}
|Z_\lambda|^2 
=
\lambda |\Phi|^2 + (1-\lambda)^2 |\Psi|^2 + q_0^2
+ \lambda (1-\lambda) \left( \Phi \overline{\Psi} + \overline \Phi \Psi \right) + \lambda \left(\Phi \overline \zeta + \overline \Phi \zeta\right) +  (1-\lambda) \left( \Psi \overline \zeta + \overline \Psi \zeta\right),
\end{equation*}
which upon differentiation yields
\begin{equation}
\begin{aligned}
2|Z_\lambda| \p_x |Z_\lambda| 
&=
\lambda^2 \left(\Phi_x \overline \Phi +  \overline \Phi_x \Phi\right) 
+ \lambda (1-\lambda) \left(\Phi_x \overline \Psi + \overline  \Phi_x \Psi\right) 
+ \lambda \left(\Phi_x \bar \zeta + \overline \Phi_x \zeta\right) 
\\
&\quad
+ \lambda (1-\lambda) \left(\overline \Phi \Psi_x + \Phi \overline \Psi_x \right) 
+ (1-\lambda)^2 \left(\overline \Psi \Psi_x + \Psi \overline \Psi_x\right) 
+ (1-\lambda) \left(\Psi_x \bar \zeta + \overline \Psi_x \zeta\right) 
\\
&\quad
+ \lambda \left(\overline \Phi \zeta' + \Phi \bar \zeta' \right)
+ (1-\lambda) \left(\overline \Psi \zeta' + \Psi \bar \zeta' \right) 
+ \left(\zeta' \bar \zeta + \bar \zeta' \zeta \right).
\end{aligned}
\end{equation}
Therefore,
\begin{equation}\label{temp}
\begin{aligned}
&\quad
\left\| \p_x \big(\left|Z_\lambda\right|^{2p}\big) \cdot \left(\Phi-\Psi\right) \right\|_{L^2(\mathbb R)} 
\leq
2p \no{\Phi-\Psi}_{L^\infty(\mathbb R)} 
\no{|Z_\lambda|^{2(p-1)} \left(\p_x |Z_\lambda|\right) |Z_\lambda|}_{L^2(\mathbb R)}
\\
&\leq
p \no{\Phi-\Psi}_{L^\infty(\mathbb R)} \no{Z_\lambda}_{L^\infty(\mathbb R)}^{2(p-1)}
\bigg[
\no{\Phi_x \overline \Phi + \overline \Phi_x \Phi}_{L^2(\mathbb R)}
+ 
\no{\Phi_x \overline \Psi + \overline  \Phi_x \Psi}_{L^2(\mathbb R)}
\\
&\quad
+ 
\no{\Phi_x \bar \zeta + \overline \Phi_x \zeta}_{L^2(\mathbb R)}
+ 
\no{\overline \Phi \Psi_x + \Phi \overline \Psi_x}_{L^2(\mathbb R)}
+
\no{\overline \Psi \Psi_x + \Psi \overline \Psi_x}_{L^2(\mathbb R)}
\\
&\quad
+
\no{\Psi_x \bar \zeta + \overline \Psi_x \zeta}_{L^2(\mathbb R)}
+ 
\no{\overline \Phi \zeta' + \Phi \bar \zeta'}_{L^2(\mathbb R)}
+
\no{\overline \Psi \zeta' + \Psi \bar \zeta'}_{L^2(\mathbb R)}
+
\no{\zeta' \bar \zeta + \bar \zeta' \zeta}_{L^2(\mathbb R)}
\bigg]
\\
&\leq
2p \no{\Phi-\Psi}_{L^\infty(\mathbb R)} \no{Z_\lambda}_{L^\infty(\mathbb R)}^{2(p-1)}
\left(
\no{\Phi}_{L^\infty(\mathbb R)}
+
\no{\Psi}_{L^\infty(\mathbb R)} 
+
\no{\zeta}_{L^\infty(\mathbb R)}
\right) 
\\
&\quad \cdot \left(\no{\Phi_x}_{L^2(\mathbb R)}
+
\no{\Psi_x}_{L^2(\mathbb R)}
+ 
\no{\zeta'}_{L^2(\mathbb R)}
\right)
\end{aligned}
\end{equation}
hence, noting that $|Z_\lambda| \leq  |\Phi| +  |\Psi| + |\zeta|$ and using the embedding $H^1(\mathbb R) \subset L^\infty(\mathbb R)$, we obtain the following bound for the first norm on the right-hand side of \eqref{diff-temp5}:
\begin{align}
\left\| \p_x \big(\left|Z_\lambda\right|^{2p}\big) \cdot \left(\Phi-\Psi\right) \right\|_{L^2(\mathbb R)} 
&\leq
2p \no{\Phi-\Psi}_{H^1(\mathbb R)} 
\left(\no{\Phi}_{H^1(\mathbb R)} + \no{\Psi}_{H^1(\mathbb R)} + \no{\zeta}_{L^\infty(\mathbb R)}  \right)^{2p-1}
\nn\\
&\quad
\cdot
\left(
\no{\Phi}_{H^1(\mathbb R)} + \no{\Psi}_{H^1(\mathbb R)}
+
\no{\zeta'}_{L^2(\mathbb R)}
\right).
\label{temp1}
\end{align}
Proceeding to the third norm on the right-hand side of \eqref{diff-temp5}, we have
\begin{align*}
&\quad
\left\| \p_x\big( \left|Z_\lambda\right|^{2(p-1)}
Z_\lambda^2 \big) \cdot \overline{\left(\Phi-\Psi\right)} \right\|_{L^2(\mathbb R)}
\leq
\no{\Phi-\Psi}_{L^\infty(\mathbb R)} \no{\p_x\big( \left|Z_\lambda\right|^{2(p-1)} Z_\lambda^2 \big)}_{L^2(\mathbb R)}
\nn\\
&\leq
\no{\Phi-\Psi}_{L^\infty(\mathbb R)} 
\left[
\no{2(p-1) \left|Z_\lambda\right|^{2(p-2)} |Z_\lambda| \left(\p_x |Z_\lambda|\right) Z_\lambda^2}_{L^2(\mathbb R)}
+
\no{\left|Z_\lambda\right|^{2(p-1)} \cdot 2 Z_\lambda \left(\p_x Z_\lambda\right)}_{L^2(\mathbb R)}
\right]
\nn\\
&\leq
\no{\Phi-\Psi}_{L^\infty(\mathbb R)} 
\left[
(p-1) \no{Z_\lambda}_{L^\infty(\mathbb R)}^{2(p-1)}
\no{2|Z_\lambda| \left(\p_x |Z_\lambda|\right)}_{L^2(\mathbb R)}
+
2\no{Z_\lambda}_{L^\infty(\mathbb R)}^{2p-1}
\no{\p_x Z_\lambda}_{L^2(\mathbb R)}
\right].
\end{align*}
The first of the $L^2$ norms  on the right-hand side can be handled along the lines of \eqref{temp}. For the second $L^2$ norm, we have
$$
\no{\p_x Z_\lambda}_{L^2(\mathbb R)}
=
\no{\lambda \Phi_x + (1-\lambda) \Psi_x + \zeta'}_{L^2(\mathbb R)}
\leq
\no{\Phi}_{H^1(\mathbb R)} + \no{\Psi}_{H^1(\mathbb R)} + \no{\zeta'}_{L^2(\mathbb R)}.
$$
Thus, using once again the embedding $H^1(\mathbb R) \subset L^\infty(\mathbb R)$, we obtain the bound
\begin{align}\label{temp2}
\left\| \p_x\big( \left|Z_\lambda\right|^{2(p-1)}
Z_\lambda^2 \big) \cdot \overline{\left(\Phi-\Psi\right)} \right\|_{L^2(\mathbb R)}
&\leq
2p \no{\Phi-\Psi}_{H^1(\mathbb R)} \left(\no{\Phi}_{H^1(\mathbb R)} + \no{\Psi}_{H^1(\mathbb R)} + \no{\zeta}_{L^\infty(\mathbb R)}  \right)^{2p-1}
\nn\\
&\quad
\cdot
\left(\no{\Phi}_{H^1(\mathbb R)} + \no{\Psi}_{H^1(\mathbb R)} + \no{\zeta'}_{L^2(\mathbb R)}\right).
\end{align}

Overall, the bounds \eqref{temp3}, \eqref{temp1} and \eqref{temp2} combine with inequality \eqref{diff-temp5} to yield 
\begin{equation*}
\begin{aligned}
&\quad
\left\| \p_x \left[ \left|\Phi+\zeta\right|^{2p} \left(\Phi+\zeta\right) - \left|\Psi+\zeta\right|^{2p} \left(\Psi+\zeta\right) \right]  \right\|_{L^2(\mathbb R)}
\\
&\leq
\left(2p+1\right)^2 \left(
\no{\Phi}_{H^1(\mathbb R)} + \no{\Psi}_{H^1(\mathbb R)}
+
\no{\zeta}_{L^\infty(\mathbb R)} + \no{\zeta'}_{L^2(\mathbb R)}
\right)^{2p}
\no{\Phi-\Psi}_{H^1(\mathbb R)}
\end{aligned}
\end{equation*}
which together with inequality \eqref{temp9} implies
\begin{equation}\label{px-G-est}
\begin{aligned}
&\quad
\no{\p_x \left(G_{F, \Phi}-G_{F, \Psi}\right)}_{L^2(\mathbb R)}
\\
&\leq
\left(2p+1\right)^2 \left(
\no{\Phi}_{H^1(\mathbb R)} + \no{\Psi}_{H^1(\mathbb R)}
+
\no{\zeta}_{L^\infty(\mathbb R)} + \no{\zeta'}_{L^2(\mathbb R)}
+
q_0
\right)^{2p}
\no{\Phi-\Psi}_{H^1(\mathbb R)}.
\end{aligned}
\end{equation}
Hence, in view of the $L^2$ estimate \eqref{Gf-diff-l2-F}, we have
\begin{align*}
&\quad
\no{G_{F, \Phi}-G_{F, \Psi}}_{H^1(\mathbb R)}
\\
&\leq
\left[2^{2p+1} p + \left(2p+1\right)^2\right] \left(\no{\Phi}_{H^1(\mathbb R)} + \no{\Psi}_{H^1(\mathbb R)} + \no{\zeta}_{L^\infty(\mathbb R)} +  \no{\zeta'}_{L^2(\mathbb R)} + q_0\right)^{2p}
\cdot \left\| \Phi - \Psi \right\|_{H^1(\mathbb R)}.
\nn
\end{align*}
In turn, inequality \eqref{L-contr} yields 
\begin{align*}
&\no{\Lambda[\Phi](t) - \Lambda[\Psi](t)}_{H^1(\mathbb R)}
\leq
2|\gamma| \left[2^{2p+1} p + \left(2p+1\right)^2\right] t 
\\
&\quad
\cdot \sup_{\tau\in [0, t]}  
\Big[
\Big(\no{\Phi(\tau)}_{H^1(\mathbb R)} + \no{\Psi(\tau)}_{H^1(\mathbb R)} 
+ \no{\zeta}_{L^\infty(\mathbb R)} +  \no{\zeta'}_{L^2(\mathbb R)} + q_0\Big)^{2p}
\left\| (\Phi - \Psi)(\tau) \right\|_{H^1(\mathbb R)}
\Big],
\nn
\end{align*}
which ensures that $\Lambda$ is a contraction in the ball $B(0, \rho) \subset C([0, T_f]; H^1(\mathbb R))$ provided that 
\begin{equation}
\label{estT}
 T_f   < 
\frac{1}{2|\gamma|\left[2^{2p+1} p + \left(2p+1\right)^2\right]\left(2\rho + \no{\zeta}_{L^\infty(\mathbb R)} +  \no{\zeta'}_{L^2(\mathbb R)} + q_0\right)^{2p}}.
\end{equation}
The proof of the local $H^1$ existence Theorem \ref{nzbc-lwp-t} is complete.
\end{proof}

\begin{remark}
The bounds \eqref{tf-sat} and \eqref{estT} for the lifespan $T_f$ get larger as the norms $\left\|\Phi_0\right\|_{H^1(\mathbb{R})}$, $\left\| \zeta \right\|_{L^\infty(\mathbb R)}$, $\no{\zeta'}_{L^2(\mathbb R)}$, $\big\| |\zeta| - q_0 \big\|_{L^2(\mathbb R)}$ and the background $q_0$ get smaller. This is consistent with the fact that, when all of these quantities are small, the problem heuristically approximates the one with zero boundary conditions, for which global existence is ensured for all $p\geq 1$ and sufficiently  small initial data according to the known results of Theorem \ref{zbc-wp-t}. In particular, when $\zeta=q_0=0$, and $1\leq p< 2$, the arguments of Theorem \ref{nzbc-lwp-t} provide the first step  towards establishing global existence of solutions satisfying the size estimate \eqref{boundg}.  In particular, this first step implies local existence of solutions satisfying \eqref{boundg}. This result is then extended to global existence by employing appropriate conservation laws, which are useful in the context of zero boundary conditions unlike the case of nonzero boundary conditions (see discussion in the next section). 
\end{remark}

\noindent
\textbf{Closeness estimates for finite times.}
Combining the closeness estimates \eqref{D-nzbc-L^2} and \eqref{nvH1} with the local existence result of Theorem \ref{nzbc-lwp-t}, we arrive at the following analogue of Theorem \ref{Theorem:stability} for the case of nonzero boundary conditions at infinity.

\begin{theorem}[Theorem \ref{Theorem:stability} for the nonzero boundary conditions \eqref{uU-bc}]
\label{nzbc-t}
Consider the modified NLS equations \eqref{cNLS} and \eqref{pNLS} with initial data $\phi_0, \Phi_0 \in H^1(\mathbb R)$ given by \eqref{mod-ic} and the boundary conditions~\eqref{vbcn}, in the case of either the saturable nonlinearity   $F(x) = \dfrac{x}{\kappa(1+x)}$ or the power nonlinearity $F(x)=x^p$ with $p\geq 1$. 
\begin{enumerate}[label=\textnormal{(\roman*)}, leftmargin=6.5mm, topsep=2mm, itemsep=1mm]
\item \underline{$L^2$ closeness}: 
Given $0<\varepsilon<1$, suppose that  the initial data satisfy
\begin{gather}
\label{d0nv}
\left\| \phi_0-\Phi_0 \right\|_{L^2(\mathbb R)} \leq  C \varepsilon^3,
\\
\label{d1nv}
\left\| \phi_0 \right\|_{H^1(\mathbb R)} \leq c_0 \, \varepsilon, \ \left\| \Phi_0 \right\|_{H^1(\mathbb R)}\leq  C_0 \, \varepsilon
\end{gather}
for some constants $C, c_0, C_0>0$. In addition, suppose that the nonzero background described by the function $\zeta \in X^1(\mathbb R)$ satisfies
\begin{equation}
\label{bgs}
q_0 \leq B\varepsilon, 
\ 
\no{\zeta}_{L^{\infty}(\mathbb{R})}\leq  B_0\varepsilon,
\ 
\no{\zeta'}_{L^{2}(\mathbb{R})}\leq B_1\varepsilon,
\ 
\big\| |\zeta| - q_0 \big\|_{L^2(\mathbb R)}\leq B_2\varepsilon
\end{equation}
for some constants $B, B_0, B_1, B_2>0$. 
Then, there exists a finite time $T_c\in \left[0, T_f\right]$, where $T_f$ is the lifespan of solutions to the non-integrable NLS equation \eqref{pNLS} from Theorem \ref{nzbc-lwp-t}, and a constant $\widetilde C=\widetilde C(\mu,\gamma, c_0, C_0, C,B, B_0, B_1, B_2, T_c)$ such that the solutions $\phi(x, t)$ and $\Phi(x, t)$ satisfy the closeness estimate 	
\begin{equation}\label{bound1nv}	
\sup_{t\in [0, T_c]}\left\| \phi(t)-\Phi(t) \right\|_{L^2(\mathbb R)}
\equiv
\no{e^{-i\mu q_0^2 t} u(t) - e^{-i\gamma F(q_0^2)t} U(t)}_{L^2(\mathbb R)} 
\leq \widetilde C \varepsilon^3.
\end{equation}
\item \underline{$H^1$ and $L^{\infty}$ closeness}: 
If the initial data $\phi_0$, $\Phi_0$ satisfy \eqref{d1nv} along with the stronger condition (in place of \eqref{d0nv})  
\begin{equation}\label{dH1nv}
\left\| \phi_0-\Phi_0 \right\|_{H^1(\mathbb R)} \le C_1 \varepsilon^3
\end{equation}
for some constant $C_1>0$ and, in addition, \eqref{bgs} holds, then there exists a constant $\widetilde C_1$ depending on $C_1$ and with a similar dependency on $T_c$ and $\mu, \gamma, c_0, C_0,B_0,B_1,B_2$ as the constant $\widetilde C$ in \eqref{bound1nv}  such that 
\begin{equation}\label{bound2nv}
\sup_{t\in [0, T_c]}\left\| \phi(t)-\Phi(t) \right\|_{H^1(\mathbb R)}
\equiv
\no{e^{-i\mu q_0^2 t} u(t) - e^{-i\gamma F(q_0^2)t} U(t)}_{H^1(\mathbb R)} 
\leq
 \widetilde C_1 \varepsilon^3.
\end{equation}
Consequently, there exists a constant  $\widetilde C_2$ with similar dependencies as $\widetilde C_1$ such that
\begin{equation}\label{boundBnv}
\sup_{t\in [0, T_c]}\left\| \phi(t)-\Phi(t) \right\|_{L^\infty(\mathbb R)}
\equiv
\no{e^{-i\mu q_0^2 t} u(t) - e^{-i\gamma F(q_0^2)t} U(t)}_{L^\infty(\mathbb R)} 
\leq \widetilde C_2 \varepsilon^3.
\end{equation}
\end{enumerate}
\end{theorem}
\begin{proof} 
We only give the details for part (ii), as part (i) can be established similarly.  The solutions to the modified cubic NLS equation \eqref{cNLS} exist globally in time with $\phi \in C([0,\infty);H^1(\mathbb{R}))$. On the other hand, thanks to Theorem \ref{nzbc-lwp-t}, the solutions to the modified general NLS equation exist at least locally with $\Phi\in C([0,T_f],H^1(\mathbb{R}))$. Consider both solutions on the time interval $[0,T_f]$ and suppose that the conditions \eqref{d1nv} are met. Then, since both solutions belong to $C([0,T_f],H^1(\mathbb{R}))$, by continuity there exist $T_{1}, T_{2} \in [0,T_f]$ associated with the solutions $\phi, \Phi$ of the integrable and non-integrable NLS equations, respectively, such that
\begin{equation}
\label{cs}
\begin{aligned}
&\left\|\phi(t)\right\|_{H^{1}(\mathbb{R})}\leq  \tilde{c}_0 \, \varepsilon \ \  \forall t\in [0, T_1],
\\
&\left\|\Phi(t)\right\|_{H^{1}(\mathbb{R})}\leq  \tilde{C}_0 \, \varepsilon \ \  \forall t\in [0, T_2],
\end{aligned} 
\end{equation}
for some constants  $\tilde{c}_0, \tilde{C}_0>0$ independent of $t\in [0, T_1]$ and $t\in [0, T_2]$, respectively. Let $T_c:=\min\left\{T_1, T_2\right\}$. Then, estimates \eqref{cs}  hold for all $t\in [0, T_c]$ and can be combined with the closeness estimate \eqref{nvH1} to imply 
\begin{equation}
 \label{nvH1b}
 \no{\Delta(t)}_{H^1(\mathbb{R})}\leq M_1T_c \, \varepsilon^{2p+1}+M_2 \, \varepsilon^{2p+1}+M_3\, \varepsilon^3,
\end{equation}
for some constants $M_j(\mu,\gamma,C_1,B_0,B_1,B_2)$, $j=1,2, 3$ and all $t\in [0,T_c]$ (recall that for the saturable nonlinearity $p=1$, while for the power nonlinearity $p\geq 1$). 
\end{proof}
\begin{remark}
The fact that the NLS solutions belong to $C([0, T_f], H^1(\mathbb{R}))$ is crucial, since it allows us to establish that solutions starting from initial data that satisfy the smallness conditions \eqref{d1nv} remain small  in the sense of the bounds \eqref{cs}, at least for short times. Such an argument could not be implemented for solutions are not continuous with respect to time, e.g. for solutions belonging in a weaker class such as $L^2([0, T_f], H^1(\mathbb{R}))$.

Moreover, estimates \eqref{bound1nv}, \eqref{bound2nv} and \eqref{boundBnv} highlight that the rotations $e^{-i\mu q_0^2 t}$ and  $ e^{-i\gamma F(q_0^2)t}$ are necessary in order to establish the closeness between the solutions of the original NLS equations \eqref{NLS} and~\eqref{NLSP} or \eqref{Sat2} in the case of the nonzero boundary conditions \eqref{uU-bc}. Indeed, these rotations are the result of converting the non-vanishing conditions \eqref{uU-bc} into vanishing ones via the changes of variables~\eqref{zhi4} and \eqref{zhi6}, which lead to the modified problems \eqref{cNLS} and \eqref{pNLS}, respectively.  For the  importance of rotations in the context of orbital stability of standing waves in the case of the vanishing boundary conditions, we refer to the reader to Remark 8.3.4 on page 274 of \cite{book1}.
\end{remark}

%
%
%
\section{The case of a finite interval}
\label{finite-s}

In this section, we examine the proximity question for initial-boundary value problems formulated on a finite interval. The impact of our results is twofold. On the one hand, they provide analytical justification for the forthcoming numerical simulations, in which the real line is  approximated by finite domains. On the other hand, they shed light on the finite domain problem, which is interesting on its own right especially in the context of global existence of solutions. In that direction, we identify major differences between the case of general nonzero boundary conditions and the case of periodic conditions.
\\[3mm]
\noindent
\textbf{Nonzero boundary conditions on a finite interval.}
Consider the integrable and non-integrable NLS equations \eqref{NLS} and \eqref{noninNLS} on the finite interval $I=(-L, L)$, $L>0$, supplemented with the boundary conditions (we use the limit notation in order to illustrate the motivation of studying this problem as a finite domain approximation to the one on $\mathbb R$)
\begin{equation}\label{nvDBC}
\lim_{x \to \pm L} u(x,t) = \lim_{x \to \pm L}e^{i\mu q_0^2t}\zeta(x), 
\quad 
\lim_{x \to \pm L} U(x,t) = \lim_{x \to \pm L} e^{i\gamma F(q_0^2)t} \zeta(x), \quad t\geq 0,
\end{equation}
where the function $\zeta$ satisfies
\begin{equation}\label{zhi2b}
 \zeta\in X^1(I),
 \quad
 \lim_{x \to \pm L} \zeta (x) = \zeta_\pm \in \mathbb C, 
 \quad 
 |\zeta_\pm| = q_0>0,
 \quad
 \lim_{x \to \pm L} \zeta' (x) = 0,
\end{equation}
with  $X^1(I)$ denoting the Zhidkov space on the interval $I$, defined analogously to \eqref{zhi1}.
In view of~\eqref{nvDBC} and \eqref{zhi2b}, $\lim_{x \to \pm L}|u(x, t)|=\lim_{x \to \pm L}|U(x,t)|=q_0$,  $t\geq 0$. Thus, making the changes of variables~\eqref{zhi4} and \eqref{zhi6}, we obtain the modified NLS equations \eqref{cNLS} and \eqref{pNLS} with zero Dirichlet boundary conditions on $I$, namely
\begin{equation}\label{Dvbcn}
\lim_{x \to \pm L} \phi(x,t)= \lim_{x \to \pm L} \Phi(x,t)=0,\quad t\geq 0.
\end{equation}

In the rest of this section, we confine our analysis to the case of the power nonlinearity 
\begin{equation}\label{Phi-p}
i \Phi_t+ \nu \left(\Phi+\zeta\right)_{xx} +\gamma \big( \left|\Phi+\zeta\right|^{2p} - q_0^{2p} \big) (\Phi+\zeta)=0, \quad p \geq 1.
\end{equation}
The case of the saturable nonlinearity can be handled similarly. 
Defining the Sobolev space $H^2_0(I)$ as the closure of $C_c^\infty(I)$ in $H^2(I)$, which can be characterized by $H^2_0(I)= \left\{f \in H^2(I): f(\pm L) = f'(\pm L) = 0\right\}$ (e.g. see \cite{lm1972}), we begin with the following result. 
\begin{lemma}\label{locex1} 
Let $\zeta$ satisfy \eqref{zhi2b} and $\Phi_0\in H^2_0(I)$. Then, there exists $T_{\max}>0$ (possibly infinite) such that the modified NLS equation \eqref{Phi-p} supplemented with the initial condition $\Phi(x, 0) = \Phi_0(x)$ and the zero Dirichlet boundary conditions~\eqref{Dvbcn}  has a unique solution $\Phi(t)\in H^2_0(I)$, $t\in [0, T_{\max})$.  
\end{lemma}

The proof of Lemma \ref{locex1} follows by adapting the proofs of the existence result given in Theorem~3.1 and the regularity result given in Theorem 4.1 of \cite{Defoc} to the focusing regime considered here. We remark that, in the defocusing regime with power nonlinearity considered in  \cite{Defoc}, global existence can be proved for the problem in higher dimensional setups for general bounded or unbounded domains $\Omega\subseteq \mathbb{R}^N$, $N\geq 1$.  In \cite{Defoc}, Galerkin approximations
are combined with an approximative domain expansion scheme for the original domain $\Omega$. This is achieved by introducing suitable extension/restriction operators
and cutoff functions; for further details and illustrative examples, we also refer the reader to~\cite{gbsjmaa}. The existence of global in time solutions in the defocusing regime is proved via suitable versions of Trudinger/Gagliardo–Nirenberg inequalities establishing existence for arbitrary time intervals through non-uniform in time estimates and continuation for all $t\in\mathbb{R}$, while regularity of solutions is proved  by estimates derived by a combination of multivariate Fa\'{a} di Bruno formulas and Gagliardo–Nirenberg type inequalities. On the other hand, in the focusing regime considered here, global existence is established at a later stage, unconditionally for $1\leq p < 2$ and with appropriate smallness conditions on the data and the size of $I$ when $p=2$ (see Remark \ref{def-r} and Theorem \ref{gex} below).

The next result provides a conservation law involving the $L^2$ norm of the solution of \eqref{pNLS}.
\begin{proposition}\label{conL2}
Suppose that the hypothesis of Lemma \ref{locex1} holds true and consider the functional
\begin{equation}\label{modL21}
\mathcal{P}[\Phi(t)] := \frac{1}{2}\left\|\Phi(t)\right\|^2_{L^2(I)}+\textnormal{Re}\int_{I}\Phi(x, t)\overline{\zeta(x)}dx.	
\end{equation}
Then, for every $t\in [0, T_{\max})$, we have the conservation law
\begin{equation}\label{modL22}
\mathcal{P}[\Phi(t)]=\mathcal{P}[\Phi_0].		
\end{equation}
\end{proposition} 

\begin{proof} 
Multiplying equation \eqref{Phi-p} by $\overline{\Phi+\zeta}$ and taking the imaginary part of the resulting expression, we have
$$
-\text{Re}\left(\Phi_t \overline \Phi\right) -  \text{Re}\left(\Phi_t \overline \zeta \right)  + \nu \text{Im}\left[(\Phi+\zeta)_{xx}(\overline{\Phi+\zeta})\right]=0.	
$$
Then, integrating over $I$ and employing integration by parts, we obtain
$$
-\frac{1}{2}\frac{d}{dt}\left\|\Phi(t)\right\|^2_{L^2(I)}
-\frac{d}{dt}\text{Re}\int_{I}\Phi(t)\overline{\zeta}dx
+ \nu \text{Im}\left[\left(\Phi+\zeta\right)_x \left(\overline{\Phi+\zeta}\right)\right]_{-L}^{L}
- \nu \text{Im} \int_I \left|\left(\Phi+\zeta\right)_x\right|^2 dx = 0.
$$
Thus, in view of the conditions \eqref{zhi2b}  and the fact that $\Phi \in H_0^2(I)$, we arrive at the desired result. 
\end{proof}

Note that, under the assumption $\zeta \in X^1(\mathbb R)$, the conservation law \eqref{modL22} is also valid on $I=\mathbb R$. 
In the case of the bounded interval $I=(-L,L)$,   \eqref{modL22} provides uniform in time $L^2$ estimates. In particular, we have
\begin{proposition}\label{globalL^2}
Suppose that the hypothesis of Lemma \ref{locex1} holds true and let $\Phi_0 \in H^2_0(I)$. Then, the unique local solution $\Phi (t) \in H^2_0(I)$ of equation \eqref{Phi-p} with the initial condition $\Phi(x, 0) = \Phi_0(x)$ and the zero boundary conditions \eqref{Dvbcn}  is uniformly bounded in $L^2(I)$ for all $t \in [0, T_{\max})$, satisfying the estimate
\begin{equation}\label{modL24A}
\sup_{t\in [0, T_{\max})} \left\|\Phi(t)\right\|^2_{L^2(I)} 
\leq 
\left\|\Phi_0\right\|^2_{L^2(I)} + L\left\|\zeta\right\|_{L^{\infty}(I)}.
\end{equation}
\end{proposition}

\begin{proof} 
By the conservation law \eqref{modL22}, for each $t\in [0, T_{\max})$ we have
\begin{equation}\label{modL25}
\frac{1}{2}\left\|\Phi(t)\right\|^2_{L^2(I)} 
+ \text{Re}\int_{I}\Phi(x, t)\overline{\zeta(x)}dx 
= 
\frac{1}{2}\left\|\Phi_0\right\|^2_{L^2(I)}+\text{Re}\int_{I}\Phi_0(x) \overline{\zeta(x)}dx.
\end{equation}
By the Cauchy-Schwarz inequality and the fact that $2ab \leq a^2 + b^2$ for any $a, b \in \mathbb R$, the second term of the left-hand side of \eqref{modL25} admits the estimate
\begin{equation}\label{modL26}
\begin{aligned}
\left|\int_{I}\Phi(x, t) \overline{\zeta(x)}dx\right|
&\leq \left\|\zeta\right\|_{L^{\infty}(I)}\int_{I} \left|\Phi(x, t)\right|dx
\\
&\leq \sqrt{2L}\left\|\zeta\right\|_{L^{\infty}(I)}\left\|\Phi(t)\right\|_{L^2(I)}
\leq \frac{1}{4}\left\|\Phi(t)\right\|^2_{L^2(I)}+2L\left\|\zeta\right\|^2_{L^{\infty}(I)}.
\end{aligned}
\end{equation}
Similarly, 
\begin{equation}
	\label{modL27}
	\left|\int_{I}\Phi_0(x) \overline{\zeta(x)}dx\right|
	\leq \frac{1}{4}\left\|\Phi_0\right\|^2_{L^2(I)}+2L\left\|\zeta\right\|^2_{L^{\infty}(I)}.
\end{equation}
Combining \eqref{modL26} and \eqref{modL27} with \eqref{modL25} and the triangle inequality yields the uniform in time estimate~\eqref{modL24A}. 
\end{proof}

Importantly, the proof of Proposition \eqref{globalL^2} given above is not valid when $I=\mathbb R$. In the case of the bounded domain $I=(-L,L)$, where that proposition is valid, it can be used to infer global existence of solutions at the level of $H^1_0(I)$. In order to prove this result, we consider the following energy functional $\mathcal{E}: H_0^1(I) \to \mathbb R$, along the lines of Section 2.2 in \cite{Defoc} but this time adapting the definition to the focusing regime:
\begin{equation}\label{Eneg1}
\mathcal{E}\left[\Phi(t); q_0,p,\zeta\right] 
:= 
\frac{1}{2} \left\|(\Phi(t)+\zeta)_x\right\|^2_{L^2(I)} - \mathcal G\left[\Phi(t);q_0,p,\zeta\right],
\end{equation}
where the functional $\mathcal G: H_0^1(I)\rightarrow\mathbb R$ is defined by
\begin{equation}\label{defG}
\begin{aligned}
\mathcal G\left[\Phi(t); q_0,p,\zeta\right] &:= \int_{I}V(\left|\Phi(x, t) +\zeta(x)\right|;q_0,p) \, dx,
\\
V(f; q_0, p) &= \frac{1}{2p + 2} f^{2p+2} - \frac{1}{2}q_0^{2p}f^2
+ \frac{p}{2p+2} q_0^{2p+2}.
\end{aligned}
\end{equation}
For the local in time $H^2_0(I)$ solutions of Lemma \ref{locex1}, the energy functional $\mathcal{E}$ is conserved:
\begin{proposition}\label{conen1}
Suppose that the hypothesis of Lemma \ref{locex1} holds true and let $\Phi_0 \in H^2_0(I)$. Then, the unique local solution $\Phi(t) \in H^2_0(I)$ of  equation \eqref{Phi-p} supplemented with the initial condition $\Phi(x, 0) = \Phi_0(x)$ and the boundary conditions \eqref{Dvbcn} conserves the energy functional $\mathcal{E}$ defined by~\eqref{Eneg1}, i.e.
\begin{equation}\label{conen2}
\mathcal{E}\left[\Phi(t); q_0,p,\zeta\right] 
= 
\mathcal{E}\left[\Phi_0; q_0,p,\zeta\right],
\quad
t\in [0, T_{\max}). 
\end{equation}	
\end{proposition}
\begin{proof}
The proof is similar to the one of Proposition \ref{conL2}. In particular, thanks to the regularity of the local solution, \eqref{conen2} can be derived by multiplying \eqref{Phi-p} by $(\overline{\Phi+\zeta})_x$ and integrating over $I$, this time keeping the real parts of the resulting expression.  
\end{proof}

\begin{remark}\label{def-r}
In the defocusing regime, the negative sign in front of $\mathcal G$ in \eqref{Eneg1} changes to a positive one and hence, since $V: [0, \infty) \to [0, \infty)$ (this can be shown via calculus techniques), thanks to the conservation law \eqref{conen2} the first term in \eqref{Eneg1} can be controlled by the value of $\mathcal E$ at $t=0$. Then, by continuation in time, $T_{\max}$ can be extended to infinity for all initial data in $H_0^2(I)$ with the improvement in comparison to \cite{Defoc}, that the estimates are uniform in time. On the other hand, due to the negative sign in \eqref{Eneg1}, this argument cannot be employed in the focusing case  considered here. Instead, we also need the uniform bounds of  Proposition \ref{globalL^2} at the $L^2$ level.
\end{remark}

We remark that, as shown in \cite{Defoc}, Proposition \eqref{conen1} is also valid when $I=\mathbb R$ and $\zeta$ is in the general Zhidkov space $X^m(\mathbb R) := 
\left\{ \zeta\in L^\infty(\mathbb R): \p^j \zeta \in L^2(\mathbb R), \ j=1, \ldots, m\right\}$ for some $m\geq 1$. 
For simplicity, we shall hereafter denote the functionals $\mathcal{E}\left[\Phi(t); q_0,p,\zeta\right]$ and $\mathcal G\left[\Phi(t); q_0,p,\zeta\right]$ by $\mathcal{E}\left[\Phi\right]$ and $\mathcal G\left[\Phi\right]$ respectively. 
The next result can be proved via estimates that are very similar to those involved in the proof of Theorem \ref{nzbc-lwp-t} (see also Proposition 2.3 in \cite{Defoc}).
\begin{lemma}\label{Gest}
Let $\varphi,\psi \in H_0^1(I)$ and $p\in\mathbb{N}$.  Then, there exist a constant  $C_1\big(\left\|\zeta\right\|_{L^{\infty}(I)}\big)=\mathcal{O}\big(\left\|\zeta\right\|_{L^{\infty}(I)}^{2p+2}\big)$ such that the functional $\mathcal G$ given by \eqref{defG} satisfies the inequality 
\begin{equation}\label{gineq2}
\big| \mathcal G\left[\varphi\right] - \mathcal G\left[\psi\right] \big|
\leq 
C_1\int_{I} 
\left(\left|\varphi\right|^{2p+1} + \left|\psi\right|^{2p+1} + \big|\left|\zeta\right| - q_0\big| \right) \left|\varphi - \psi \right| dx.
\end{equation}			
\end{lemma}
We are now ready to proceed to the proof of global existence of solutions at the level of  the $H^1_0(I)$ norm. 
\begin{theorem}
\label{gex}
Suppose that the hypothesis of Lemma \ref{locex1} holds true and let $\Phi_0 \in H^2_0(I)$. Then, under appropriate smallness conditions when $p=2$ (see \eqref{condsL} below), the unique local solution $\Phi(t) \in H^2_0(I)$ of  equation \eqref{Phi-p} supplemented with the initial condition $\Phi(x, 0) = \Phi_0(x)$ and the boundary conditions~\eqref{Dvbcn}  exists globally in $H_0^1(I)$. In particular,
\begin{enumerate}[label=\textnormal{(\roman*)}, leftmargin=8mm, topsep=2mm, itemsep=1mm]
\item If $1\leq p<2$, then $\Phi(t)$ is uniformly bounded in $H^1_0(I)$, unconditionally with respect to the size of the $L^2(I)$ norm of the initial data.
\item If $p=2$, then there exists a constant  
$C\big(p, \left\|\zeta\right\|_{L^{\infty}(I)}\big)=\mathcal{O}\big(\left\|\zeta\right\|_{L^{\infty}(I)}^{2p+2}\big)$  such that if
\begin{equation}\label{condsL}
L<\frac{1}{\left\|\zeta\right\|_{L^{\infty}(I)}}  \left(\frac{1}{4C}\right)^{\frac{2}{p+2}},
\quad
\left\|\Phi_0\right\|^2_{L^{2}(I)} <  \left(\frac{1}{4C}\right)^{\frac{2}{p+2}} - L\left\|\zeta\right\|_{L^{\infty}(I)}
\end{equation}
then $\Phi(t)$ is uniformly bounded in $H^1_0(I)$. 
\end{enumerate}
\end{theorem}

\begin{proof} 
If $T_{\max}$ in Lemma \ref{locex1} is infinite, then we are done. Otherwise, if $T_{\max}$ is finite, then in order for the solution to exist only locally in $H_0^1(I)$ 
it must be that  $\lim_{t\to T_{\max}} \no{\Phi(t)}_{H^1(I)} = \infty$.  However, below we show that $\sup_{t\in [0, T_{\max})} \no{\Phi(t)}_{H^1(I)} < C(\Phi_0, \zeta, q_0, p, L) <\infty$, reaching a contradiction. %
Of course, it could be that $\lim_{t\to T_{\max}} \no{\Phi_{xx}(t)}_{L^2(I)} = \infty$, which is why we do not claim global existence in $H_0^2(I)$, i.e. we do not claim that $T_{\max} = \infty$ in Lemma \ref{locex1} but only at the level of the $H_0^1(I)$ norm.

We begin by noting that the norm $\left\| \Phi(t) \right\|_{L^2(I)}$ is uniformly bounded in $t$ thanks to \eqref{modL24A}. Thus, we only need to consider $\left\| \Phi_x(t) \right\|_{L^2(I)}$. 
Suppressing the dependence on $t$, by the conservation law~\eqref{conen2} 
\begin{equation*}
\frac{1}{2} \left\|\Phi_x + \zeta' \right\|^2_{L^2(I)} = \frac{1}{2}\left\|{\Phi_0}_x + \zeta' \right\|^2_{L^2(I)}+\mathcal G\left[\Phi\right]- \mathcal G\left[\Phi_0\right].	
\end{equation*}	
Then, employing inequality \eqref{gineq2} with $\varphi=\Phi$ and $\psi=\Phi_0$, we obtain
\begin{equation}\label{gineq5}
	\frac{1}{2} \left\|\Phi_x + \zeta' \right\|^2_{L^2(I)}
\leq \frac{1}{2}\left\|{\Phi_0}_x + \zeta'\right\|^2_{L^2(I)}+ C_1\int_{I}\Big(|\Phi|^{2p+1}+|\Phi_0|^{2p+1}+\big| |\zeta|-q_0\big|\Big) \left|\Phi-\Phi_0\right|dx.
\end{equation}	
By the triangle inequality, the second term on the right hand-side of \eqref{gineq5} can be bounded by
\begin{equation}\label{gineq6}
\begin{aligned}
&\quad
\int_{I}|\Phi|^{2p+2}dx + \int_{I}|\Phi|^{2p+1}|\Phi_0|dx + \int_{I}|\Phi_0|^{2p+1}|\Phi|dx
\\
&+ \int_{I}|\Phi_0|^{2p+2}dx+ \int_{I} \big| |\zeta|-q_0\big| |\Phi|dx + \int_{I} \big| |\zeta|-q_0\big| |\Phi_0|dx.
\end{aligned} 
\end{equation}

For the first term in \eqref{gineq6}, we employ the Gagliardo-Nirenberg inequality \eqref{GN} with $2p+2$ in place of $p$, $j=0$, $q=r=2$ and $m=1$ (so that $\theta=\frac{p}{2p+2}$) to infer
\begin{equation}\label{gineq7}
C_1 \int_{I} |\Phi|^{2p+2}dx \leq C_2\left\|\Phi\right\|^{p+2}_{L^2(I)} \left\|\Phi_x\right\|_{L^2(I)}^p
\end{equation}
where the constant $C_2$ depends only on $\zeta$ and $p$.  Note that the constant in the Gagliardo-Nirenberg inequality is independent of the domain $I$, i.e. it is independent of $L$. 
For the second term in \eqref{gineq6}, we first apply H\"{o}lder's inequality with $p'=\frac{2p+2}{2p+1}>1$ and $q'=2p+2$ (so that $\frac{1}{p'}+\frac{1}{q'}=1$) and then Young's product inequality for the same choice of $p'$ and $q'$ to obtain 
$$
\int_{I}|\Phi|^{2p+1}|\Phi_0|dx
\leq \left\|\Phi\right\|_{L^{2p+2}(I)}^{2p+1}\left\|\Phi_0\right\|_{L^{2p+2}(I)}
\leq \frac{2p+1}{2p+2} \left\|\Phi\right\|_{L^{2p+2}(I)}^{2p+2}+ \frac{1}{2p+2} \left\|\Phi_0\right\|_{L^{2p+2}(I)}^{2p+2}
$$
so that via \eqref{gineq7} we find
\begin{equation}\label{gineq8}
C_1 \int_{I}|\Phi|^{2p+1}|\Phi_0|dx
\leq C_3\left\|\Phi\right\|^{p+2}_{L^2(I)}\left\|\Phi_x\right\|_{L^2(I)}^p+C_4\left\|\Phi_0\right\|_{L^{2p+2}(I)}^{2p+2}. 
\end{equation}
We note that $C_3, C_4=\mathcal{O}\big(\left\|\zeta\right\|_{L^{\infty}(I)}^{2p+2}\big)$ similarly to $C_1$. 
For the third term in~\eqref{gineq6}, the Cauchy-Schwarz inequality yields
\begin{equation}\label{gineq9}
\int_{I}|\Phi_0|^{2p+1}|\Phi| dx \leq \left\|\Phi\right\|_{L^2(I)}\left\|\Phi_0\right\|^{2p+1}_{L^{4p+2}(I)}.	
\end{equation}
The fourth term in \eqref{gineq6} is analogous to the first one. For  the fifth term we recall that according to~\eqref{zhi2b} we have $\lim_{x\rightarrow\pm L}(|\zeta|-q_0)=0$ and $|\zeta|-q_0\in L^2(I)$ and so by the Cauchy-Schwarz inequality 
\begin{equation}\label{bdb}
\int_{I} \big| |\zeta|-q_0\big| |\Phi|dx \leq \big\| |\zeta|-q_0\big\|_{L^2(I)}\left\|\Phi\right\|_{L^2(I)},
\end{equation}
with the sixth and final term in \eqref{gineq6}  admitting an analogous estimate.

In view of \eqref{gineq6} and of the estimates \eqref{gineq7}-\eqref{bdb}, inequality \eqref{gineq5} becomes 
\begin{equation}\label{gineq10}
\begin{aligned}
\frac{1}{2}\left\|\Phi_x + \zeta'\right\|^2_{L^2(I)}
&\leq \frac{1}{2}\left\|{\Phi_0}_x + \zeta' \right\|^2_{L^2(I)}
+C_4\left\|\Phi_0\right\|_{L^{2p+2}(I)}^{2p+2}
+C_1\left\|\Phi\right\|_{L^2(I)}\left\|\Phi_0\right\|^{2p+1}_{L^{4p+2}(I)}
\\
&\quad
+ C_2\left\|\Phi_0\right\|^{p+2}_{L^2(I)} \left\|{\Phi_0}_x\right\|_{L^2(I)}^p
+C_5\left\|\Phi\right\|^{p+2}_{L^2(I)}\left\|\Phi_x\right\|_{L^2(I)}^p
\\
&\quad
+C_1 \big\| |\zeta|-q_0\big\|_{L^2(I)}\|\Phi\|_{L^2(I)}+C_1\big\| |\zeta|-q_0\big\|_{L^2(I)}\|\Phi_0\|_{L^2(I)}
\end{aligned}
\end{equation}
where $C_5=\mathcal{O}\big(\left\|\zeta\right\|_{L^{\infty}(I)}^{2p+2}\big)$  similarly to $C_1$.  

At this point, the uniform in time estimate \eqref{modL24A} comes into play (recall that this estimate is not valid on $\mathbb R$, which is why we only claim global existence on the finite interval $I$). Letting 
$$
R^2 = R^2\big(\left\|\Phi_0\right\|_{L^2(I)}, \left\|\zeta\right\|_{L^{\infty}(I)}, L\big)  := \left\|\Phi_0\right\|^2_{L^2(I)} + L\left\|\zeta\right\|_{L^{\infty}(I)},
$$
we combine \eqref{modL24A} with \eqref{gineq10} to deduce
\begin{equation*}
\begin{aligned}
\sup_{t\in [0, T_{\max})} \left\|\Phi_x(t)+\zeta'\right\|^2_{L^2(I)}
&\leq
\left\|{\Phi_0}_x+\zeta'\right\|^2_{L^2(I)}
+2C_4\left\|\Phi_0\right\|_{L^{2p+2}(I)}^{2p+2}
+2C_1 R \left\|\Phi_0\right\|^{2p+1}_{L^{4p+2}(I)}
\\
&\quad
+2C_2\left\|\Phi_0\right\|^{p+2}_{L^2(I)} \left\|{\Phi_0}_x\right\|_{L^2(I)}^p
+2C_5 R^{p+2} \sup_{t\in [0, T_{\max})}  \left\|\Phi_x(t)\right\|_{L^2(I)}^p
\\
&\quad
+2C_1 \big\| |\zeta|-q_0\big\|_{L^2(I)}  \big(R+\left\|\Phi_0\right\|_{L^2(I)}\big).
\end{aligned}
\end{equation*}
Thus, in view of the inequality $\left\|\Phi_x\right\|_{L^2(I)}^2 \leq 2\left\|\Phi_x+\zeta'\right\|_{L^2(I)}^2+2\left\|\zeta'\right\|_{L^2(I)}^2$,  
\begin{equation}\label{gineq11}
\begin{aligned}
\sup_{t\in [0, T_{\max})} \left\|\Phi_x(t)\right\|^2_{L^2(I)}
&\leq
2\left\|{\Phi_0}_x+\zeta'\right\|^2_{L^2(I)}
+
4C_4\left\|\Phi_0\right\|_{L^{2p+2}(I)}^{2p+2}
+
4C_1 R \left\|\Phi_0\right\|^{2p+1}_{L^{4p+2}(I)}
\\
&\quad
+
4C_2\left\|\Phi_0\right\|^{p+2}_{L^2(I)} \left\|{\Phi_0}_x\right\|_{L^2(I)}^p
+
4C_1 \big\| |\zeta|-q_0\big\|_{L^2(I)}  \big(R+\left\|\Phi_0\right\|_{L^2(I)}\big)
\\
&\quad
+2\left\|\zeta'\right\|_{L^2(I)}^2
+
4C_5 R^{p+2} \sup_{t\in [0, T_{\max})}  \left\|\Phi_x(t)\right\|_{L^2(I)}^p.
\end{aligned}
\end{equation}
Note that the various Lebesgue norms on the right-hand side can be controlled by the $L^\infty(I)$ norm (and hence by the $H_0^1(I)$ norm) due to the crucial fact that we are working on the finite domain $I$.

In the subcritical case $1 \leq p < 2$, dividing~\eqref{gineq11} by $\sup_{t\in [0, T_{\max})}  \left\|\Phi_x(t)\right\|_{L^2(I)}^p$ we obtain
\begin{align}\label{gineq11-div}
\sup_{t\in [0, T_{\max})} \left\|\Phi_x(t)\right\|_{L^2(I)}^{2-p} 
&\leq
4C_5 R^{p+2} 
+
\frac{1}{\sup_{t\in [0, T_{\max})}  \left\|\Phi_x(t)\right\|_{L^2(I)}^p}
\Big[2\left\|{\Phi_0}_x+\zeta'\right\|^2_{L^2(I)}
\nn\\
&\quad
+
4C_4\left\|\Phi_0\right\|_{L^{2p+2}(I)}^{2p+2}
+
4C_1 R \left\|\Phi_0\right\|^{2p+1}_{L^{4p+2}(I)}
+
4C_2\left\|\Phi_0\right\|^{p+2}_{L^2(I)} \left\|{\Phi_0}_x\right\|_{L^2(I)}^p
\nn\\
&\quad
+
4C_1 \big\| |\zeta|-q_0\big\|_{L^2(I)}  \big(R+\left\|\Phi_0\right\|_{L^2(I)}\big)
+2\left\|\zeta'\right\|_{L^2(I)}^2
\Big],
\end{align}
which shows that $\sup_{t\in [0, T_{\max})}  \left\|\Phi_x(t)\right\|_{L^2(I)} < \infty$ and hence implies global existence in $H^1_0(I)$. Indeed, if $\sup_{t\in [0, T_{\max})}  \left\|\Phi_x(t)\right\|_{L^2(I)} = \infty$ then, due to the fact that  $2-p>0$ and $p>0$, \eqref{gineq11-div} leads to the contradiction $\infty \leq 4C_5 R^{p+2}$. Thus, part (i) of the theorem has been proved.

In the critical case $p=2$, \eqref{gineq11} can be rearranged to yield the uniform  bound 
\begin{equation}\label{gineq11-crit}
\begin{aligned}
&\quad
\left(1-4C_5 R^{4}\right) \sup_{t\in [0, T_{\max})} \left\|\Phi_x(t)\right\|^2_{L^2(I)}
\\
&\leq
2\left\|{\Phi_0}_x+\zeta'\right\|^2_{L^2(I)}
+
4C_4\left\|\Phi_0\right\|_{L^{6}(I)}^{6}
+
4C_1 R \left\|\Phi_0\right\|^{5}_{L^{10}(I)}
+
4C_2\left\|\Phi_0\right\|^{4}_{L^2(I)} \left\|{\Phi_0}_x\right\|_{L^2(I)}^2
\\
&\quad
+
4C_1 \big\| |\zeta|-q_0\big\|_{L^2(I)}  \big(R+\left\|\Phi_0\right\|_{L^2(I)}\big)
+2\left\|\zeta'\right\|_{L^2(I)}^2
\end{aligned}
\end{equation}
and so we once again deduce $\sup_{t\in [0, T_{\max})}  \left\|\Phi_x(t)\right\|_{L^2(I)} < \infty$, but this time provided that $1 - 4C_5 R^{p+2} > 0$, which is precisely the second condition in \eqref{condsL} (note that this condition is realizable thanks to the first condition in \eqref{condsL}). This implies global well-posedness in $H_0^1(I)$ and hence completes the proof of part (ii) of the theorem with $C=C_5$ (which depends only on $p$ and $\left\|\zeta\right\|_{L^{\infty}(I)}$).
\end{proof}

We remark that if $p>2$ then neither of the two arguments used in the proof of Theorem \ref{gex} works.
Theorem \ref{gex} has the following implication concerning the proximity of solutions to the integrable and non-integrable NLS equations in the case of a finite domain.
\begin{theorem}[Theorem \ref{Theorem:stability} for the finite interval]
\label{gex2}
Given $0<\ve<1$, suppose that the hypothesis of Theorem \ref{nzbc-t} with the relevant norms over $\mathbb{R}$ replaced with their counterparts over  the finite interval $I = (-L, L)$ holds true. If $L=\mathcal{O}\left(\frac 1\ve\right)$, then the closeness estimates of Theorem \ref{nzbc-t} are valid on $I$ for arbitrary $0<T_f<\infty$.
\end{theorem}

\begin{proof}
Since $\left\|\zeta\right\|_{L^{\infty}(I)}=\mathcal{O}(\varepsilon)$ and $L=\mathcal{O}\left(\frac 1\ve\right)$ by hypothesis, the second condition in \eqref{condsL} is satisfied for sufficiently small initial data such that $\left\|\Phi_0\right\|_{L^{2}(I)}=\mathcal{O}(\varepsilon)$ and thus we can invoke the global existence result of Theorem  \ref{gex}.
In particular, we may assume that $\left\|\Phi_0\right\|_{H^{1}_0(I)}=\mathcal{O}(\varepsilon)$ so that  the form of the right-hand side of  estimates \eqref{gineq11-div} and \eqref{gineq11-crit} guarantees (at least) $\left\|\Phi (t)\right\|_{H^{1}_0(I)}=\mathcal{O}(\varepsilon)$. 
Then, assuming also the condition \eqref{dH1nv} on the $H^1(I)$ distance between the initial data of the integrable and non-integrable models, we can derive the closeness estimates exactly as in the proof of Theorem \ref{nzbc-t}. Now, however, thanks to the global existence result of Theorem \ref{gex}, we can replace $T_c$ by any arbitrary finite time $0<T_f<\infty$.
\end{proof}

\vspace*{2mm}
\noindent
\textbf{The periodic problem.}
We turn our attention to the periodic Cauchy problem, namely to the case where the integrable and non-integrable NLS equations \eqref{NLS} and  \eqref{noninNLS}  are supplemented
with the periodic boundary conditions 
\begin{equation}
u(-L, t)=u(L, t), \quad U(-L, t)=U(L, t), \quad t\geq 0.\label{eq:pbcfinite}
\end{equation}
Importantly, the counterpart of Theorem \ref{zbc-wp-t} for the periodic conditions \eqref{eq:pbcfinite}, which guarantees global solutions at the level of the $H^1$ norm  with the same estimate as \eqref{boundg},  is proved in Theorem~2.1 of \cite{lrs1988}. 
With this global existence result in hand, the analogue of Theorem \ref{Theorem:stability} on the proximity of solutions to the integrable and non-integrable NLS equations in the case of the periodic boundary conditions \eqref{eq:pbcfinite} can be stated as follows.

\begin{theorem}[Theorem \ref{Theorem:stability} for the periodic Cauchy problem]
\label{THpbcfinite}
For $p\geq 1$ and $\nu, \mu, \gamma>0$, consider the integrable and non-integrable focusing NLS equations \eqref{NLS} and \eqref{noninNLS} with $x\in I = (-L, L)$, the initial conditions $u(x, 0) = u_0(x)$, $U(x, 0) = U_0(x)$, and the periodic boundary conditions \eqref{eq:pbcfinite}.  
\begin{enumerate}[label=\textnormal{(\roman*)}, leftmargin=7.5mm]
\item \underline{$L^2$ closeness}: Let $0<\varepsilon<1$ and suppose that the initial data satisfy
\begin{gather}
\label{d0finite}
\left\| u_0-U_0 \right\|_{L^2(I)} \le C \varepsilon^3,
\\
\label{d1finite}
\left\| u_0 \right\|_{H^1(I)} \le c_0 \, \varepsilon,
\ 
\left\| U_0 \right\|_{H^1(I)}\leq  C_0 \, \varepsilon,
\end{gather}
for some constants $c_0, C_0, C>0$.	Then, for arbitrary finite $0<T_f<\infty$, there exists a constant $\widetilde C=C(\mu, \gamma, c_0, C_0, C, T_f)$ such that  	
\begin{equation}\label{eq:bound1finite}
\sup_{t\in [0, T_f]} \left\| u(t)-U(t) \right\|_{L^2(I)} \le \widetilde C \varepsilon^3.
\end{equation}
\item \underline{$H^1$ and $L^{\infty}$ closeness}: Suppose that the initial data satisfy \eqref{d1finite} and the stronger condition (instead of \eqref{d0finite})
\begin{equation}\label{dH1finite}
\left\| u_0-U_0 \right\|_{H^1(I)} \le C_1 \varepsilon^3
\end{equation}
for some constant $C_1$. Then, there exists a constant $\widetilde{C}_1$ with similar dependencies as $\widetilde C$ in \eqref{bound1} such that
\begin{equation}\label{boundBfinite-h1}
\sup_{t\in [0, T_f]}  \left\| u(t)-U(t) \right\|_{H^1(I)}\le \widetilde{C}_1 \varepsilon^3.
\end{equation}
and, in turn, 
\begin{equation}\label{boundBfinite}
\sup_{t\in [0, T_f]}  \left\| u(t)-U(t) \right\|_{L^{\infty}(I)}\le \widetilde C_2 \varepsilon^3.
\end{equation}
\end{enumerate}
\end{theorem}

\begin{proof} 
The argument is entirely analogous to the one used for the proof of Theorem \ref{TH1} in the non-periodic case. In particular, expressing periodic functions $\varphi \in L^2(I)$ in the form of the Fourier series 
$\varphi(x) = \displaystyle \sum_{\xi\in\mathbb Z} e^{i\frac{\pi}{L} \xi x} \, \widetilde{\varphi}(\xi)$
where
$\widetilde{\varphi}(\xi, t)= \displaystyle \frac{1}{2L}\int_{-L}^L e^{-i\frac{\pi}{L} \xi x} \, \varphi(x) dx$, 
we find that the difference of solutions $\Delta = u-U$ to the integrable and non-integrable NLS equations satisfies
\begin{equation}\label{fs-delta}
\widetilde{\Delta}(\xi,t)=  e^{- i \nu \left(\frac{\pi}{L} \xi\right)^2 t} \widetilde{\Delta}(\xi,0)-i\int_0^t  e^{- i \nu \left(\frac{\pi}{L} \xi\right)^2(t-\tau)} \widetilde{N}(\xi,\tau)d\tau,
\end{equation}
where $\widetilde{N}(\xi, t)$ denotes the spatial Fourier series of the nonlinearity function $N(x, t)$ defined by~\eqref{eqd1}.
Thus, since $\left\| \Delta(t) \right\|_{L^2(I)} = \sqrt{2L} \, \big\| \widetilde{\Delta}(t) \big\|_{\ell^2(\mathbb Z)}$ by Parseval's theorem, using Minkowski's integral inequality and  the fact that $e^{i\nu\left(\frac{\pi}{L} \xi\right)^2 t}$ is unitary, we have
\begin{equation}
\left\|\Delta(t)\right\|_{L^2(I)} 
\leq
\left\|\Delta(0)\right\|_{L^2(I)} + t \sup_{\tau \in [0, t]} \left\|N(\tau) \right\|_{L^2(I)}.
\end{equation}
Then, similarly to the steps that led to estimate \eqref{F-est-0}, we employ \eqref{F-prop} and the analogue of the embedding \eqref{eqd3} on $I$ for $q=6$ and also for $q=2(p+1)$ to obtain
\begin{align*}
\left\|N(t)\right\|_{L^2(I)}^2
&\leq
2\mu^2 \left\| u(t) \right\|_{L^6(I)}^6 + 2\gamma^2 K^2 \left\| U(t) \right\|_{L^{2(2p+1)}(I)}^{2(2p+1)}
\nn\\
&\leq
2\mu^2 \left\| u(t) \right\|_{H^1(I)}^6 + 2 \gamma^2 K^2 \left\| U(t) \right\|_{H^1(I)}^{2(2p+1)}
\end{align*}
from which we infer
\begin{equation*}
\left\|\Delta(t)\right\|_{L^2(I)}
\leq
\left\|\Delta(0)\right\|_{L^2(I)} 
+ 
A \sup_{\tau \in [0, t]} \left(\left\| u(\tau) \right\|_{H^1(I)}^3 + \left\| U(\tau) \right\|_{H^1(I)}^{2p+1} \right) t
\end{equation*}
with $A = \sqrt 2\max\{|\mu|, |\gamma| K\}$ as before. 
Therefore, for  initial data $u_0$, $U_0$ satisfying \eqref{d0finite} and~\eqref{d1finite}, using the periodic analogue of Theorem \ref{zbc-wp-t} and  the solution size estimate \eqref{boundg} (see Theorem 2.1 on page 662 of \cite{lrs1988}), we deduce the desired $L^2$ bound \eqref{eq:bound1finite}, namely
\begin{equation*}
\left\|\Delta(t)\right\|_{L^2(I)}
\leq
C \varepsilon^3
+ 
A \big(M^3 c_0^3 \, \varepsilon^3 + M^{2p+1} C_0^{2p+1} \, \varepsilon^{2p+1} \big) t
\leq
\widetilde C \varepsilon^3, \quad t\in [0, T_f],
\end{equation*}
with $\widetilde C = \max\big\{C,  A M^3 c_0^3 T_f, A M^{2p+1} C_0^{2p+1} T_f\big\}$. 

In order to prove \eqref{boundBfinite-h1}, we combine the definition of the $H^1(I)$ norm with \eqref{fs-delta} to infer
\begin{equation*}
\left\|\Delta(t)\right\|_{H^1(I)}^2
\leq 
2\sum_{\xi \in \mathbb Z} \left(1+\xi^2\right) \big|\widetilde{\Delta}(\xi, 0) \big|^2
+
2\sum_{\xi \in \mathbb Z} \left(1+\xi^2\right) \left( \int_0^t \big|\widetilde{N}(\xi, \tau) \big| d\tau  \right)^2.
\end{equation*}
Hence, by Minkowski's integral inequality, 
\begin{equation*}
\left\|\Delta(t)\right\|_{H^1(I)}^2
\leq 
2\left\|\Delta(0)\right\|_{H^1(I)} + 2t^2 \sup_{\tau \in [0, t]} \left\|N(\tau) \right\|_{H^1(I)}
\end{equation*}
and proceeding as for \eqref{F-est-0} and \eqref{F'-est} we eventually obtain the analogue of \eqref{bound2}.
Finally, the $L^\infty$ estimate \eqref{boundBfinite} follows from the $H^1$ estimate \eqref{boundBfinite-h1} via the Sobolev embedding theorem.
\end{proof}

\vspace*{2mm}
\noindent
\textbf{Numerical studies for the case of nonzero boundary conditions.}
We now present the results of numerical studies in the case of nonzero boundary conditions in order to illustrate (i) our theoretical results of Section \ref{nzbc-s} on the nonzero boundary conditions on the whole line and (ii) the preceding theoretical results of the current section concerning the approximations of the problem on the infinite line by finite intervals. The numerical studies 
are motivated by the numerical investigations of \cite{bm2017} and \cite{blmt2018}. 
We consider the simplest case where $\zeta(x)=q_0>0$. We will treat two examples of initial conditions decaying on the nonzero background, namely exponentially and algebraically decaying, the latter stemming from the Peregrine soliton. 

\textit{I. Exponentially decaying initial data on the top of the nonzero background $q_0$}. 
We study the dynamics of the NLS equations emerging from the initial condition of the form 
\begin{equation}
\label{NVinc}
    u(x,0)=q_0(1+i\,\mathrm{sech} x).
\end{equation}
In all the numerical results we set $\mu=\gamma=1$ for the nonlinearity parameters of the integrable and non-integrable NLS equations, and $\nu=1$, for the linear dispersion parameter.  For the amplitude of the background we set $q_0=0.25$ in order to comply with the smallness conditions on the initial data of Theorem \ref{nzbc-t}.
Based on the analysis presented in Section \ref{nzbc-s} and Theorem \ref{nzbc-t}, we may expect closeness between the solutions of the integrable and the non-integrable NLS for short times. Thus, we start with the presentation of the numerical results for a short time interval $t\in [0, 10]$.  Figure \ref{fig5} depicts contour plots of the spatiotemporal evolution of the density of the NLS equations. The left panel depicts the dynamics for the integrable NLS, the central panel for the non-integrable NLS with the subcritical power nonlinearity~\eqref{NLSP} with $p=3/2$ (quartic case), and the right panel for the non-integrable NLS~\eqref{Sat2} with saturable nonlinearity. The patterns look almost indistinguishable and we expect this to be verified by the evolution of norms of the distance function. This is indeed the case, as it is shown in Figure \ref{fig6}. Due to estimate \eqref{bound1nv} of Theorem \ref{nzbc-t}, we examine the evolution of the norm 
\begin{equation}\label{newnorm}	
\|\widetilde{\Delta}(t)\|_{\mathcal{X}}:=\no{e^{-i\mu q_0^2 t} u(t) - e^{-i\gamma F(q_0^2)t} U(t)}_{\mathcal{X}}.
\end{equation}
For all norms we observe an excellent agreement with the predictions of Theorem \ref{nzbc-t} regarding their linear growth. On the other hand, we find again the moderate increase of $\|\widetilde{\Delta}(t)\|_{L^{\infty}}$ and the larger growth of $\|\widetilde{\Delta}(t)\|_{L^{2}}$ and $\|\widetilde{\Delta}(t)\|_{H^1}$. Again the saturable model exhibits dynamics which are closer to the integrable one than the quartic NLS, while both non-integrable systems are closer to the integrable in the sense of the pointwise topology.  The larger, although still mediocre deviation in the $L^2$ and $H^1$ norms, is a first indication  that in the non-integrable dynamics there may be smaller, finer structures or oscillations present in the solutions that are not captured by the large-scale pattern. These smaller scales could be associated with faster oscillations or sharper gradients in the solutions of the non-integrable NLS. The existence of such scales can become more eminent  over longer time intervals.

\begin{figure}[ht!]
	\begin{center}
		   \includegraphics[width=.33\textwidth]{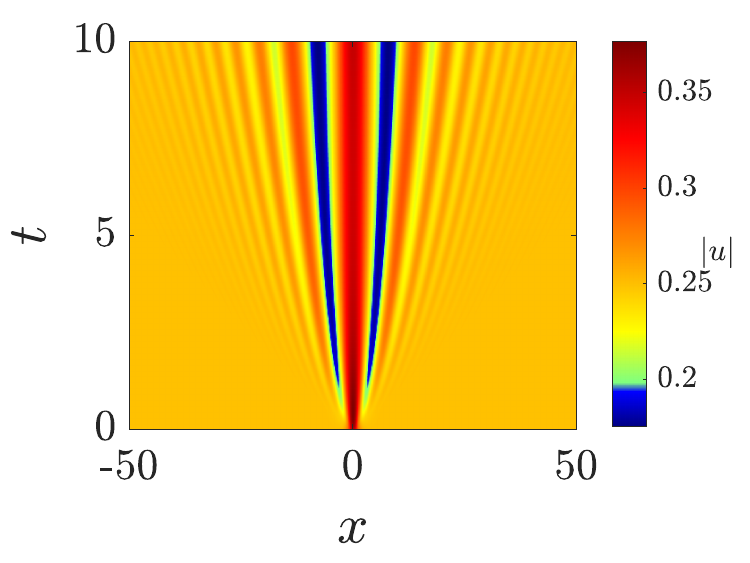}
			\hspace{-0.1cm}\includegraphics[width=.33\textwidth]{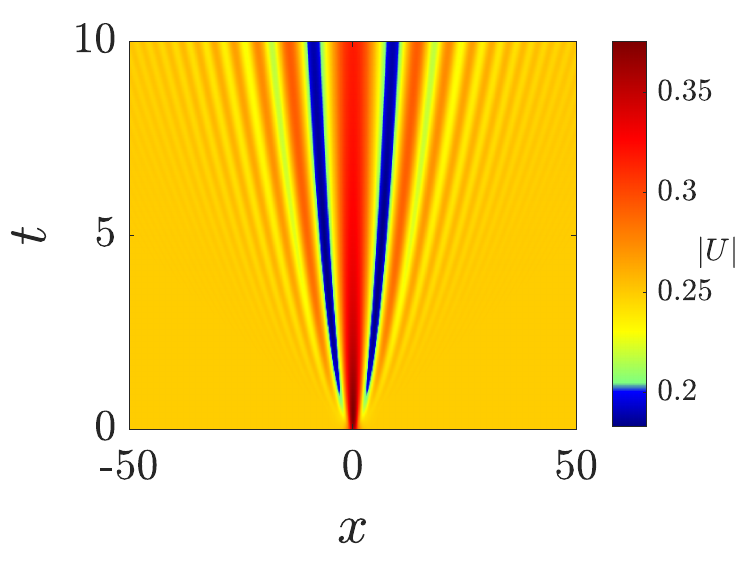}
			\hspace{-0.1cm}\includegraphics[width=.33\textwidth]{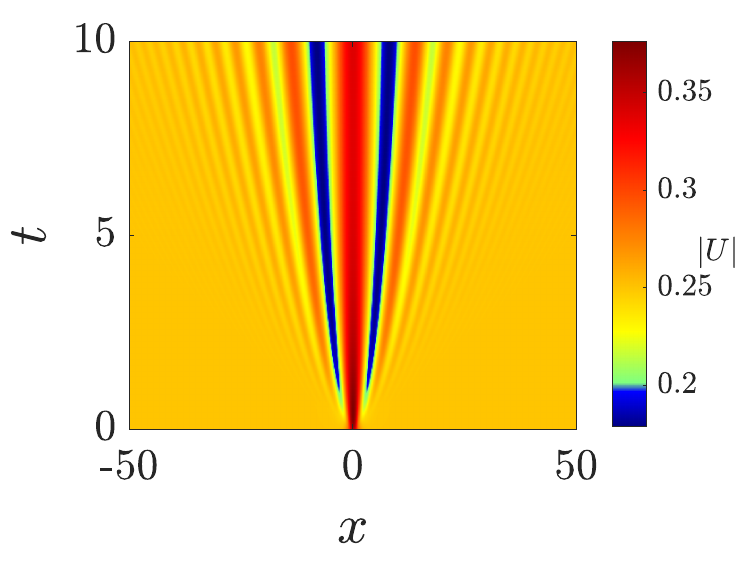}
	\end{center}
	\caption{Contour plots of the spatiotemporal evolution of the initial condition \eqref{NVinc} with $q_0=0.25$. Left: Integrable NLS $p=1$. Center: Non-integrable NLS \eqref{NLSP} with power nonlinearity in the subcritical case $p=3/2$. Right: Non-integrable NLS \eqref{Sat2} with saturable nonlinearity. 
	}
	\label{fig5}
\end{figure}

\begin{figure}[ht!]
	\begin{center}
   \includegraphics[width=.33\textwidth]{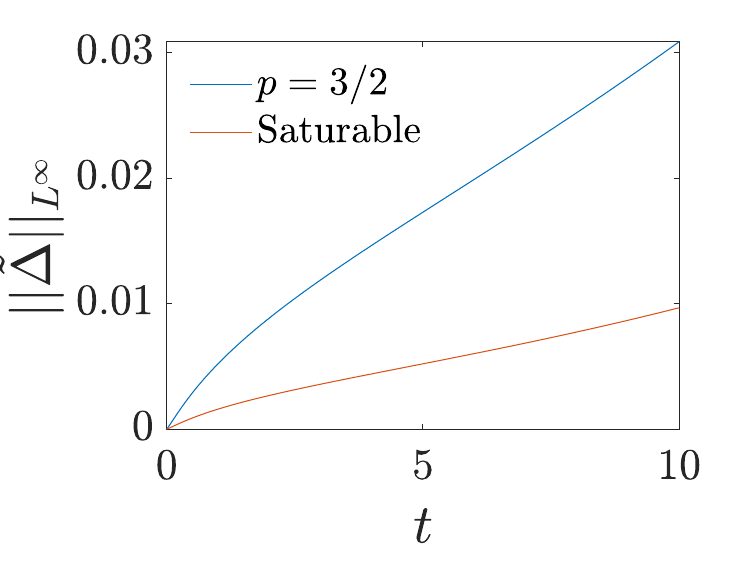}
			\hspace{-0.1cm}\includegraphics[width=.33\textwidth]{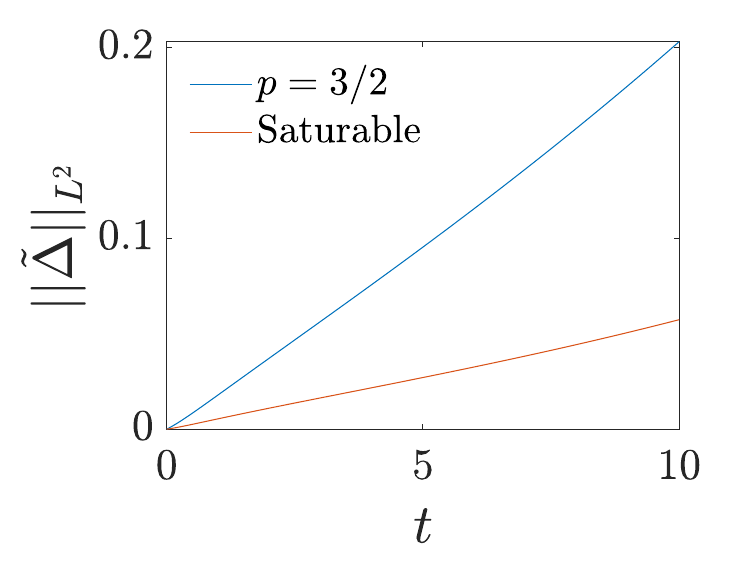}
			\hspace{-0.1cm}\includegraphics[width=.33\textwidth]{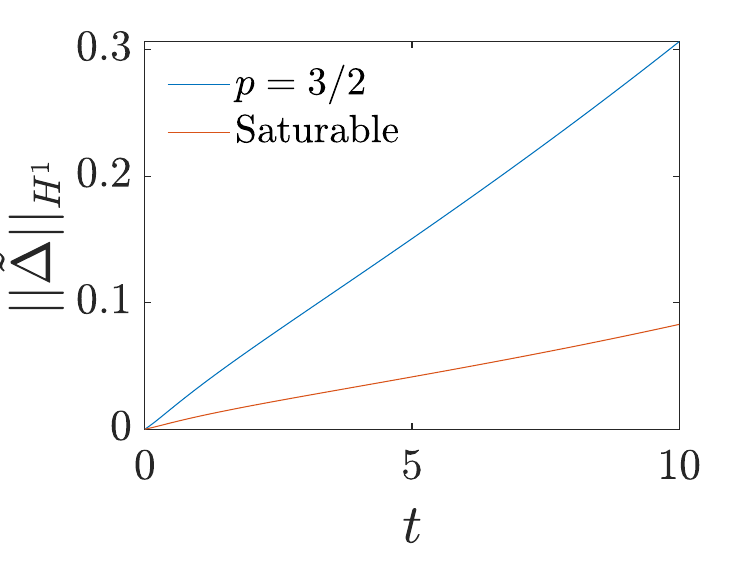}
	\end{center}
	\caption{Evolution of the distance norm 
 \eqref{newnorm} for the dynamics shown in Figure~\ref{fig5}. Left: $\|\widetilde{\Delta}(t)\|_{L^{\infty}}$. Center: $\|\widetilde{\Delta}(t)\|_{L^{2}}$. Right: $\|\widetilde{\Delta}(t)\|_{H^1}$.
	}
	\label{fig6}
\end{figure}

\begin{figure}[ht!]
	\begin{center}
   \includegraphics[width=.33\textwidth]{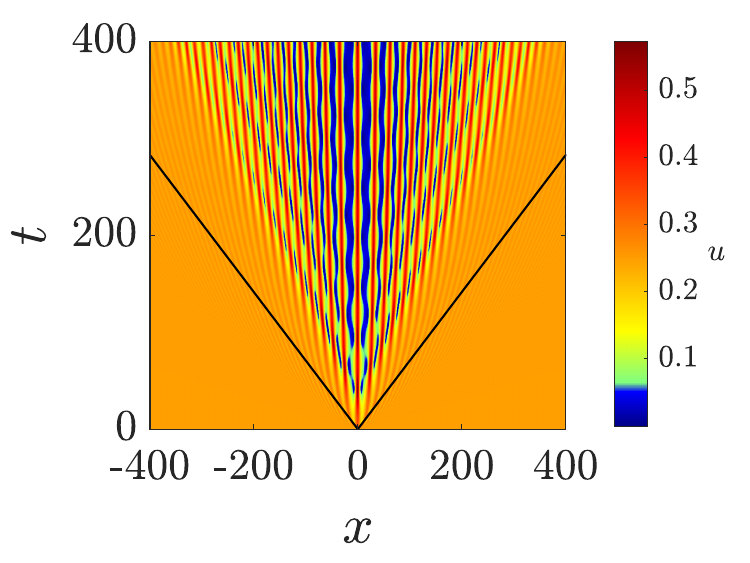}
			\hspace{-0.1cm}\includegraphics[width=.33\textwidth]{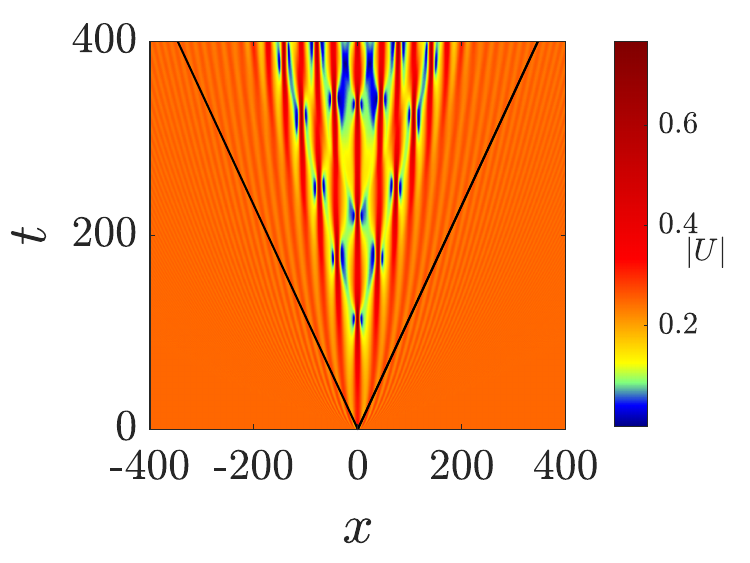}
   \includegraphics[width=.33\textwidth]{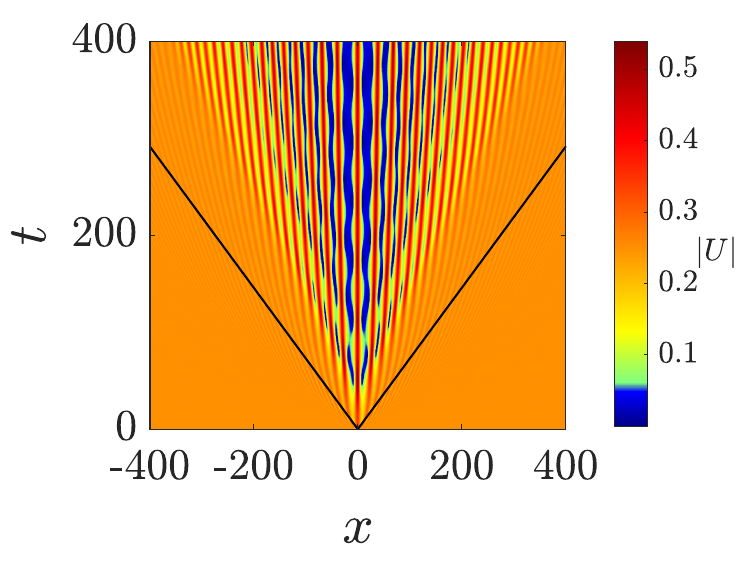}
	\end{center}
	\caption{Contour plots of the spatiotemporal evolution of the initial condition~\eqref{NVinc} with $q_0=0.25$ for longer times $t\in [0,400]$. Left: Integrable NLS $p=1$. Center: Non-integrable NLS \eqref{NLSP} with power nonlinearity  in the subcritical case $p=3/2$. Right: Non-integrable NLS \eqref{Sat2} with saturable nonlinearity. 
	}
	\label{fig7}
\end{figure}

\begin{figure}[ht!]
	\begin{center}
   \includegraphics[width=.33\textwidth]{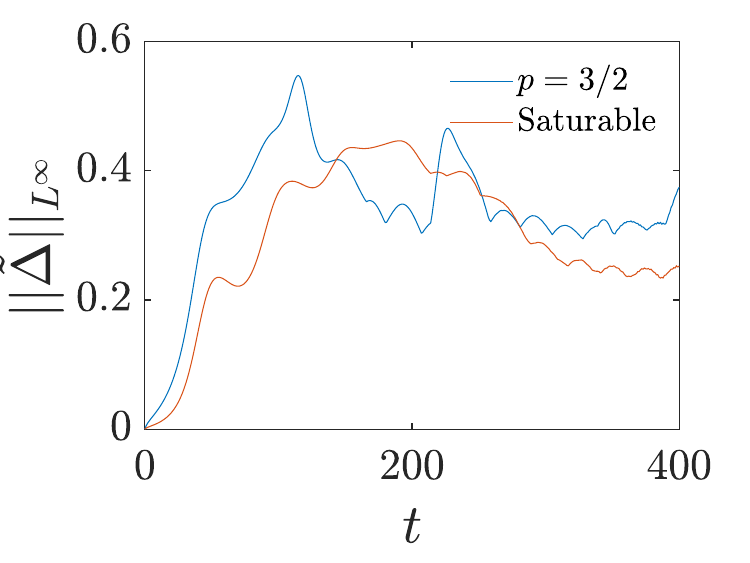}
			\hspace{-0.1cm}\includegraphics[width=.33\textwidth]{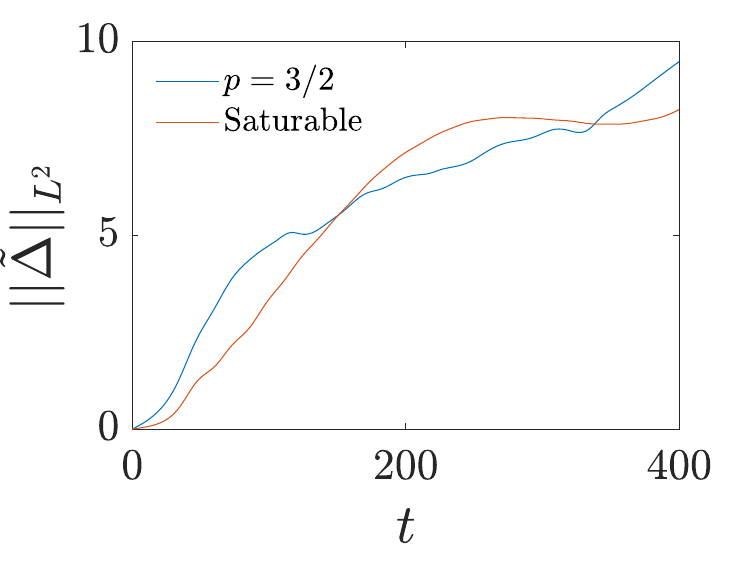}
			\hspace{-0.1cm}\includegraphics[width=.33\textwidth]{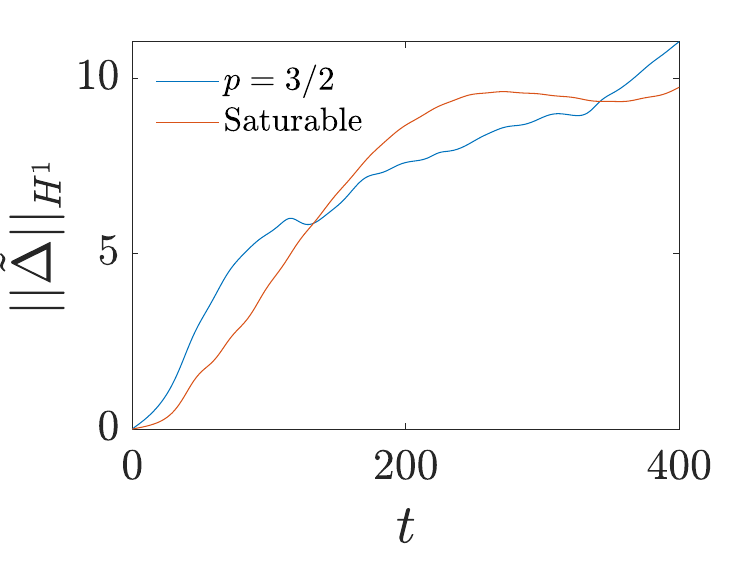}
	\end{center}
	\caption{Evolution of the distance norm 
 \eqref{newnorm} for the dynamics of Figure \ref{fig7} over $t\in [0,400]$. Left: $\|\widetilde{\Delta}(t)\|_{L^{\infty}}$. Center: $\|\widetilde{\Delta}(t)\|_{L^{2}}$. Right: $\|\widetilde{\Delta}(t)\|_{H^1}$. 
	}
	\label{fig8}
\end{figure}

\begin{figure}[ht!]
	\begin{center}
   \includegraphics[width=.33\textwidth]{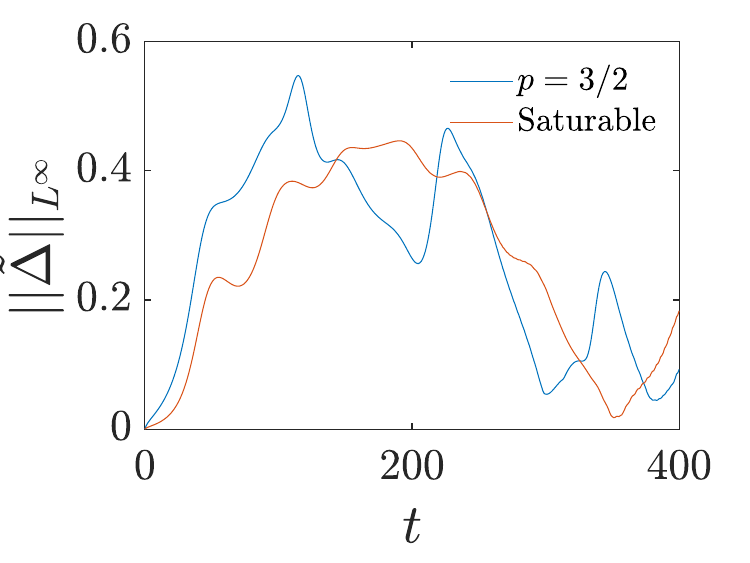}
			\hspace{-0.1cm}\includegraphics[width=.33\textwidth]{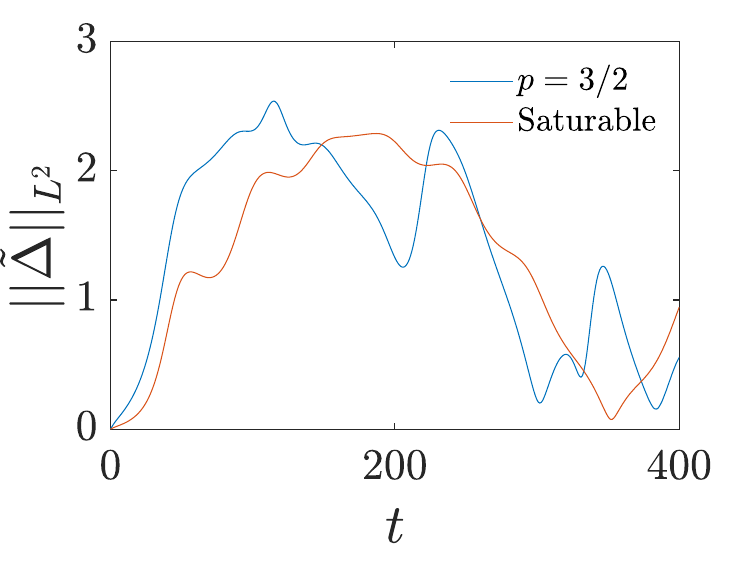}
			\hspace{-0.1cm}\includegraphics[width=.33\textwidth]{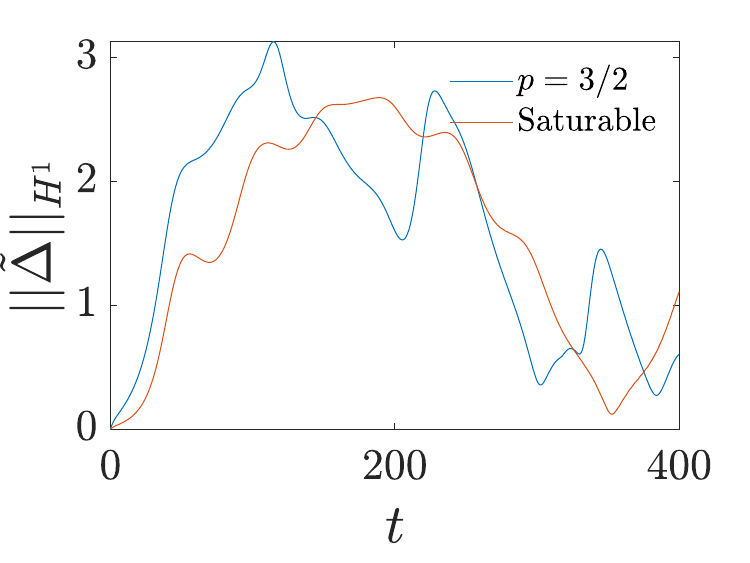}
	\end{center}
	\caption{Evolution of the distance norm 
 \eqref{newnorm} for the dynamics of Figure~\ref{fig7} over $t\in [0,400]$, but with the norm evaluated over a spatial interval of length $\mathcal O(1/\ve)$, namely $x\in (-1/\ve, 1/\ve) = (-4,4)$. Left: $\|\widetilde{\Delta}(t)\|_{L^{\infty}(-4,4)}$. Center: $\|\widetilde{\Delta}(t)\|_{L^{2}(-4,4)}$. Right: $\|\widetilde{\Delta}(t)\|_{H^1(-4,4)}$.
	}
	\label{fig9}
\end{figure}

\begin{figure}[ht!]
	\begin{center}
   \includegraphics[width=.33\textwidth]{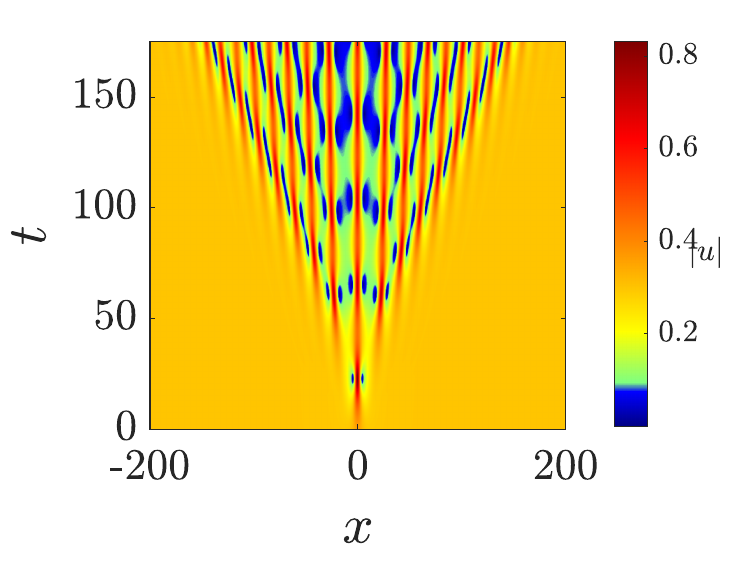}
			\hspace{-0.1cm}\includegraphics[width=.33\textwidth]{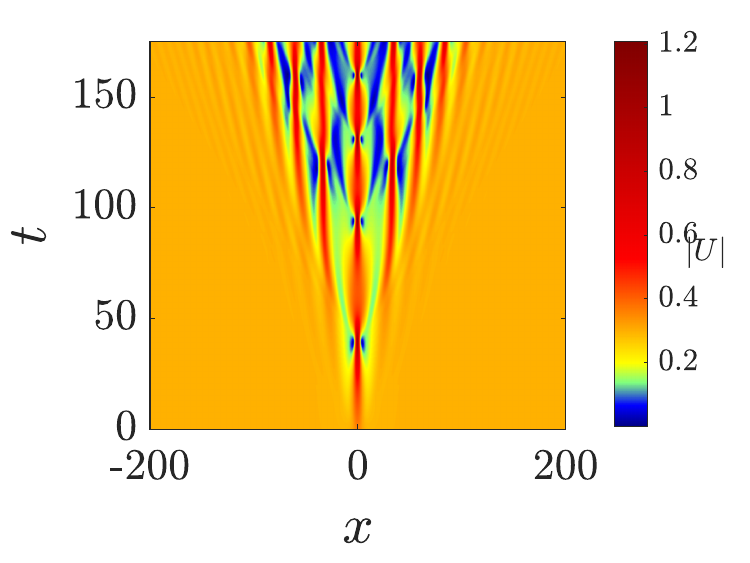}
			\hspace{-0.1cm}\includegraphics[width=.33\textwidth]{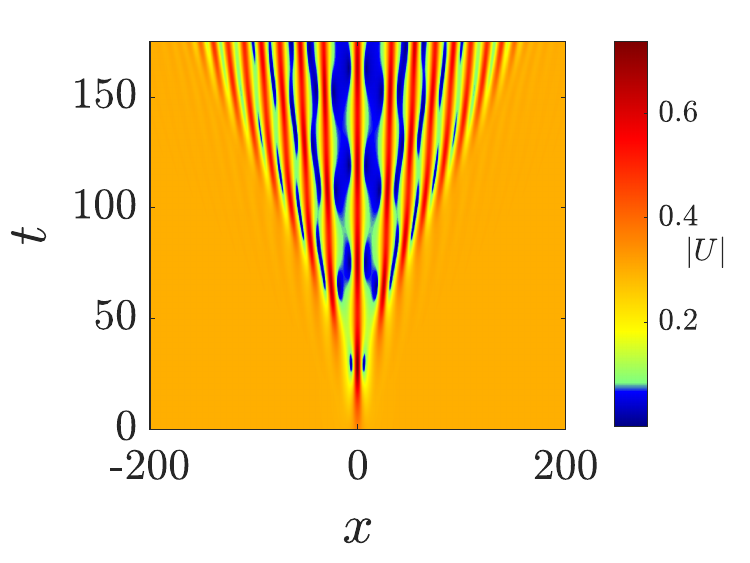}
	\end{center}
	\caption{Contour plots of the spatiotemporal evolution of the initial condition \eqref{psinc} with $\delta=\nu=1$, $\gamma=\mu=1$, $q_0=0.3$ and $T_0=-25$ for $t\in [-50,100]$. Left: Integrable NLS $p=1$. Center: Non-integrable NLS \eqref{NLSP} with power nonlinearity in the subcritical case $p=3/2$. Right: Non-integrable NLS \eqref{Sat2} with saturable nonlinearity.
	}
	\label{fig11}
\end{figure}
To investigate further these issues, we extended considerably the time-horizon of the above numerical investigations. Figure \ref{fig7} extends the study presented in Figure \ref{fig5} to a time span $t\in [0,400]$. The results of Figure \ref{fig7} verify that the evolution observed for $t\in [0,10]$ is the initial stage of the modulational instability dynamics discussed in detail in \cite{bm2017} and \cite{blmt2018}. In each panel, the bounding lines represent the linear caustics separating the $x$-$t$ plane into two types of regions: a left far-field region and a right far-field region, where the solution meets the condition at infinity (up to a phase shift), and a central region in which the asymptotic behavior is described by slowly modulated periodic oscillations. Although at the level of the nonlinear stage of modulational instability the non-integrable dynamics shares to a great extent the basic features of the ones of the integrable counterpart, there are still some important findings revealed. Comparing the pattern of the integrable NLS (shown in the left panel) and the corresponding one for the saturable nonlinearity (right panel), we observe that they have more in common than the pattern for the quartic nonlinearity (central panel).  The long-time behavior of the saturable model confirms also in the case of nonzero boundary conditions, that it is the more structural stable model in reference to the integrable one than its power nonlinearity counterpart.  For the quartic NLS, the oscillatory pattern within the caustics region is markedly  different  as it exhibits a considerably larger period of the internal modulated wave than in the other systems. On the other hand, the maximum amplitude of the internal oscillations seems to be similar in all systems indicating that the similarity of all patterns at large amplitudes should be manifested in the $L^{\infty}$ closeness while the differences of the patterns in their finer structures should be reflected in more pronounced  deviations in the $L^2$ and $H^1$ norms.

\begin{figure}[tbp!]
	\begin{center}
		\begin{tabular}{cc}
   \includegraphics[width=.45\textwidth]{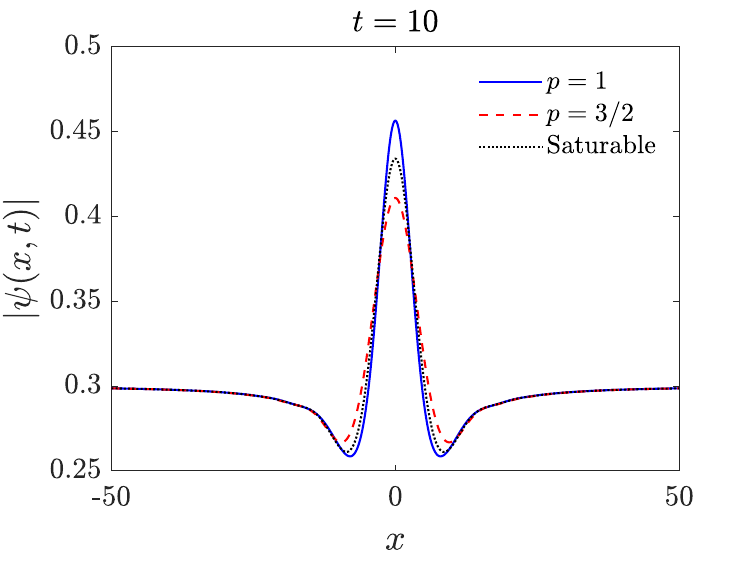}
   \includegraphics[width=.45\textwidth]{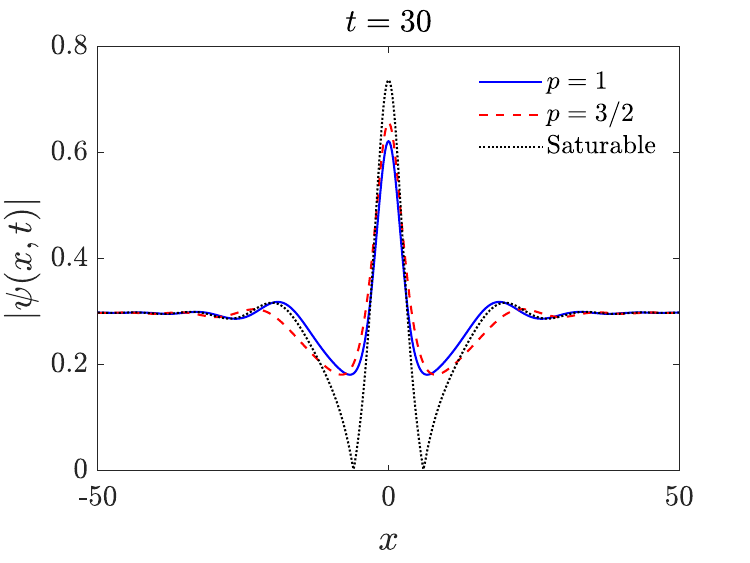}\\
   \includegraphics[width=.45\textwidth]{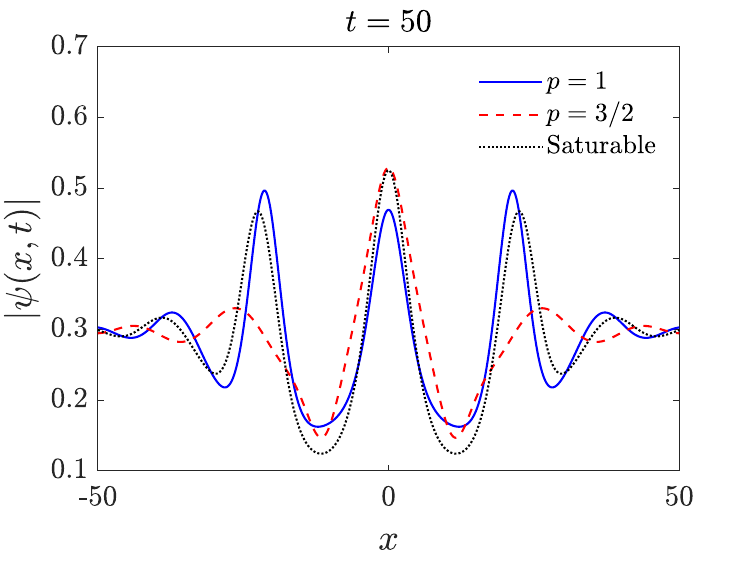}
\includegraphics[width=.45\textwidth]{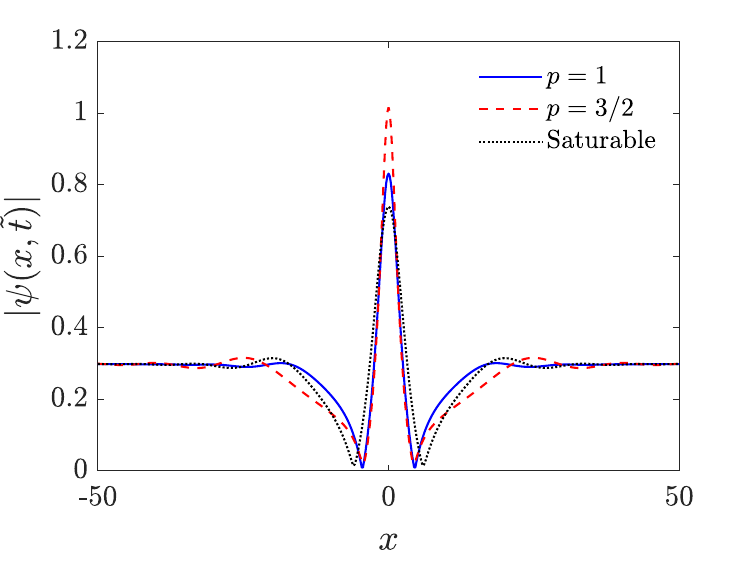}
		\end{tabular}
	\end{center}
	\caption{Snapshots of the evolution of the initial condition \eqref{psinc} with $\delta=\nu=1$, $\gamma=\mu=1$, $q_0=0.3$ and $T_0=-25$, portraying the  comparison between the profiles of the solutions of the integrable NLS with $p=1$, the non-integrable NLS with power nonlinearity $p=3/2$, and the non-integrable NLS with saturable nonlinearity.  The bottom right panel depicts the comparison of the profiles of the first event in each of the panels of Figure \ref{fig11}, which occurs at $\tilde t=25$ for the integrable case $p=1$,   $\tilde t\sim 30$ for the power nonlinearity $p=3/2$, and  $\tilde t\sim 25$ for the saturable case. 
	}
	\label{fig12}
\end{figure}

In connection with these predictions, we illustrate in Figure \ref{fig8} the evolution of the distance $\|\widetilde{\Delta}(t)\|_{\mathcal{X}}$ associated with the dynamics shown in Figure \ref{fig7}. From the central panel, showing the evolution of $\|\widetilde{\Delta}(t)\|_{L^{\infty}}$ it is evident  that the non-integrable modulational instability dynamics are remarkably close to the integrable ones in the sense of the pointwise topology, as $\|\widetilde{\Delta}(t)\|_{L^{\infty}}$ initially grows linearly but then becomes bounded and at most of order $\sim\mathcal{O}(10^{-1})$. Concerning  $\|\widetilde{\Delta}(t)\|_{L^{2}}$ (central panel)  and $\|\widetilde{\Delta}(t)\|_{H^1}$ (right panel), we observe an almost linear growth for $t\in [0,100]$ and a large deviation in compliance with the differences of the oscillatory behavior seen in Figure \ref{fig7}. Revisiting the latter in the saturable and integrable cases, there important differences between their oscillations close to the bounding caustics. These  differences explain the deviation of the systems in the $L^2$ and $H^1$ topology. On the other hand, the patterns of the integrable and the saturable NLS look very similar in the core, close to $x=0$.    This observation suggests an  agreement with the theoretical claim of Theorem \ref{gex} for the bounded domain approximation, that the dynamics should be more similar when considered on finite intervals of length $L\sim\mathcal{O}\left(\frac 1\ve\right)$.   The fact, that when restricted to an interval  of length $\sim\mathcal{O}\left(\frac 1\ve\right)$, the similarity is significantly more noticeable than on the full length interval $L$, is confirmed in Figure~\ref{fig9}. 

With our choice of $\left\|\zeta\right\|_{L^{\infty}}=q_0=0.25$ we monitor the behavior of the distances in the restricted interval for $x\in[-4,4]$ in accordance with   Theorems \ref{gex} and \ref{gex2}. The plots feature the remarkable decrease of the deviation from the integrable dynamics particularly for the $\|\widetilde{\Delta}(t)\|_{L^{2}}$   and $\|\widetilde{\Delta}(t)\|_{H^1}$ norms; in contrast with the plots of Figure \ref{fig8}, the deviation in these norms seems is bounded and at most of order $\mathcal{O}(1)$, suggesting that on large spatial scales, most of the deviation stems from the finer oscillating structures far from the core of the pattern.  
\begin{figure}[ht!]
	\begin{center}
  \includegraphics[width=.33\textwidth]{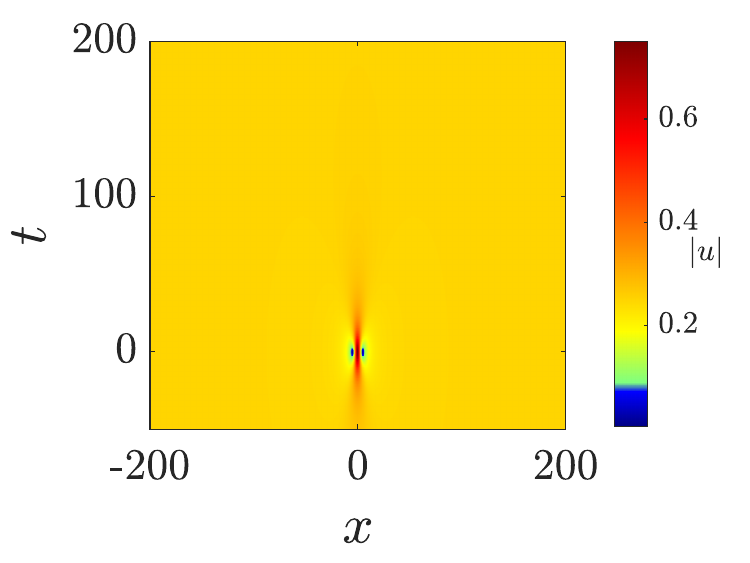}
			\hspace*{-0.1cm}\includegraphics[width=.33\textwidth]{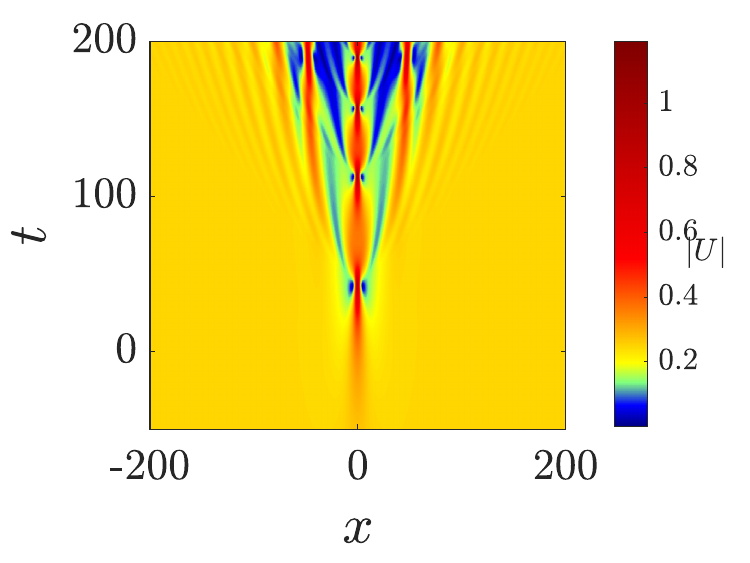}
			\hspace*{-0.1cm}
   \includegraphics[width=.33\textwidth]{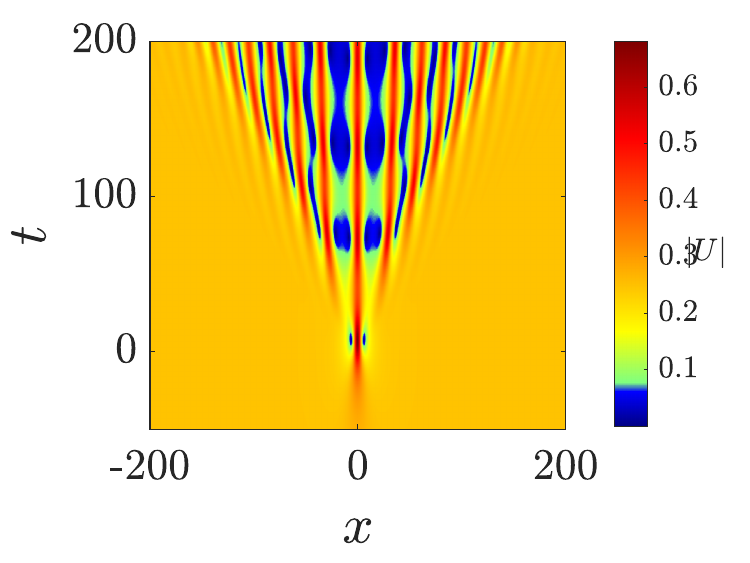}
	\end{center}
	\caption{Contour plots of the spatiotemporal evolution of the initial condition \eqref{psinc} with $\delta=2$, $\nu=1$, $\gamma=\mu=1$, $q_0=0.25$ and $T_0=50$ for $t\in [-50,200]$. Left: Integrable NLS $p=1$. Center: Non-integrable NLS \eqref{NLSP} with power nonlinearity  in the subcritical case $p=3/2$. Right: Non-integrable NLS \eqref{Sat2} with saturable nonlinearity.
	}
	\label{fig13}
\end{figure}

\begin{figure}[ht!]
	\begin{center}
  \includegraphics[width=.33\textwidth]{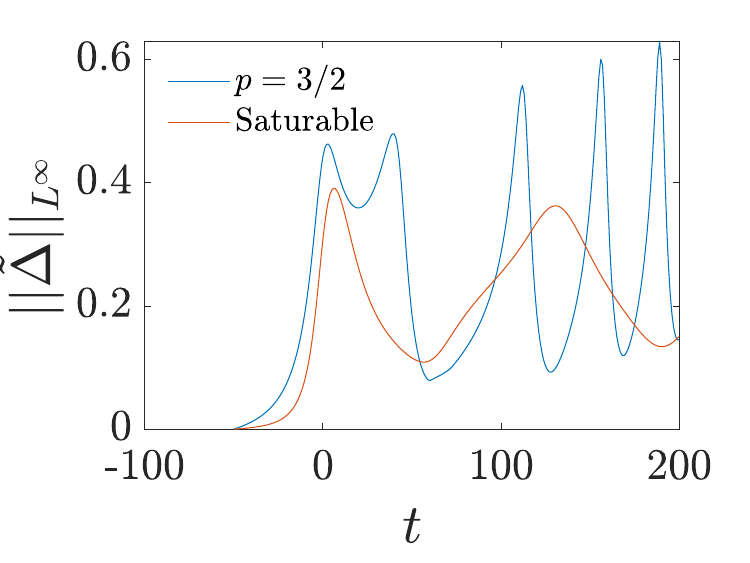}
			\hspace{-0.1cm}\includegraphics[width=.33\textwidth]{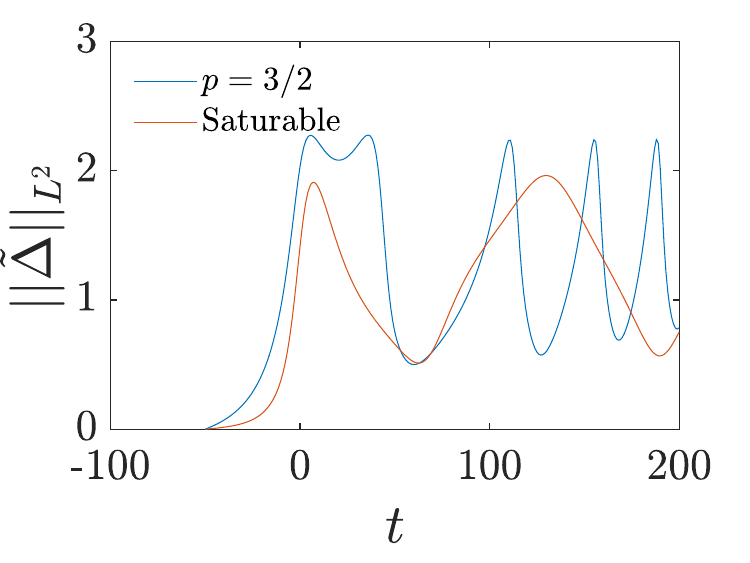}
			\hspace{-0.1cm}\includegraphics[width=.33\textwidth]{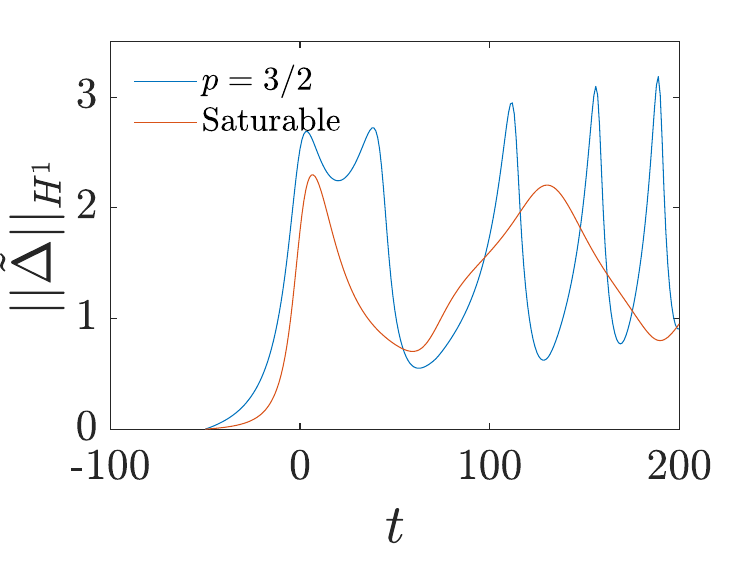}
	\end{center}
	\caption{Evolution of the distance norm  $\|\widetilde{\Delta}(t)\|_{\mathcal{X}}$ defined in
 \eqref{newnorm} for the dynamics of the initial condition \eqref{psinc} with the parameters of Figure \ref{fig13}, $x\in[-4,4]$ and $t\in [-50,250]$. Left: $\|\widetilde{\Delta}(t)\|_{L^{\infty}}$. Center: $\|\widetilde{\Delta}(t)\|_{L^{2}}$. Right: $\|\widetilde{\Delta}(t)\|_{H^1}$.
	}
	\label{fig10}
\end{figure}
\begin{figure}[tbh!]
	\begin{center}
		\begin{tabular}{cc}
\hspace{-0cm}\includegraphics[width=.48\textwidth]{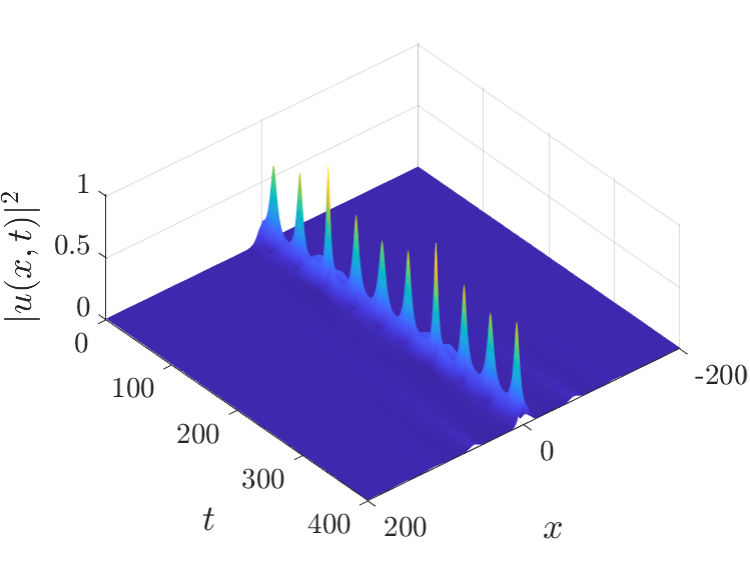}
\hspace{0cm}\includegraphics[width=.48\textwidth]{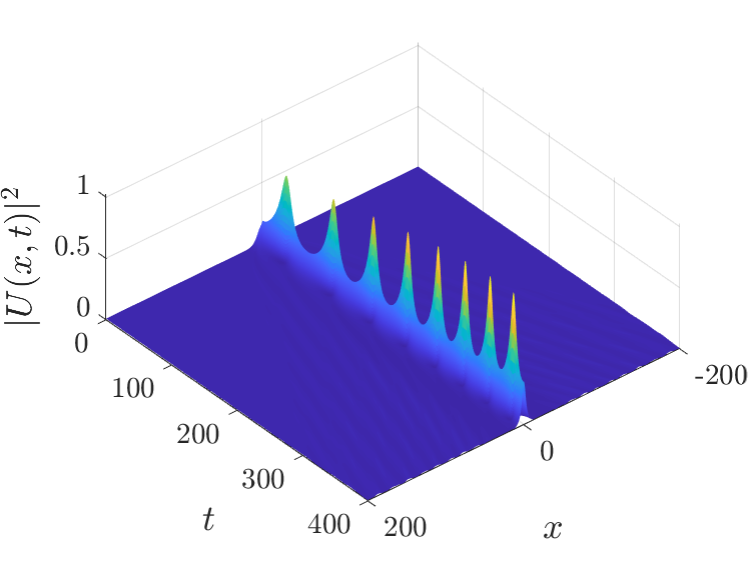}
\\
\hspace{-0.cm}\includegraphics[width=.48\textwidth]{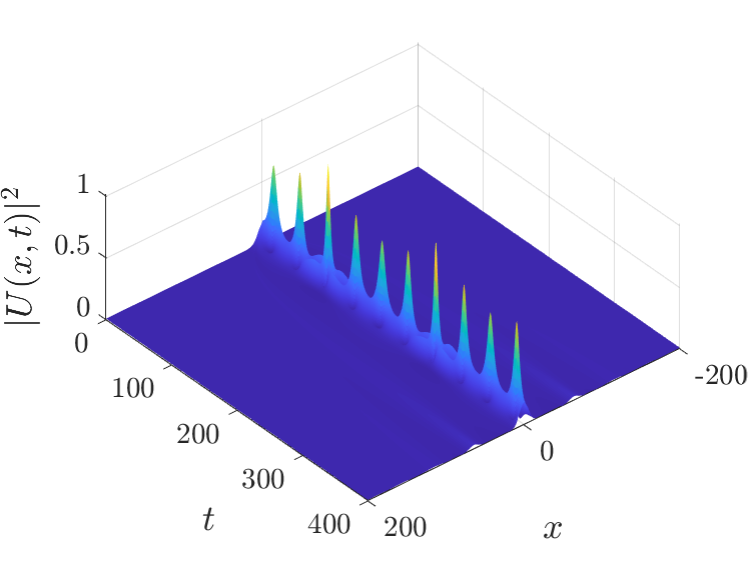}
\hspace{0cm}\includegraphics[width=.41\textwidth]{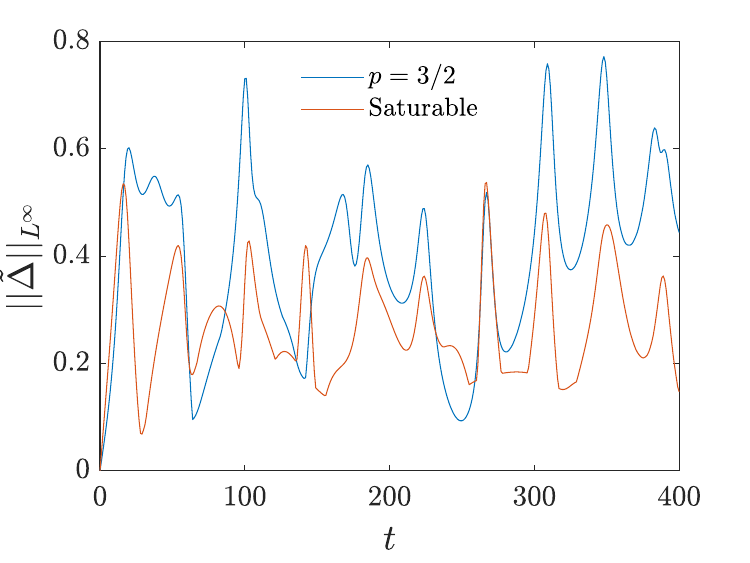}
		\end{tabular}
	\end{center}
	\caption{Simulation of the spatiotemporal evolution of the density of the solutions to the integrable and non-integrable NLS equations with the nonzero boundary conditions \eqref{uU-bc} over the non-constant background $\zeta(x)$ of the form \eqref{last1} and the initial conditions \eqref{last2} (corresponding to the modified NLS equations \eqref{cNLS} and \eqref{pNLS}). Parameters $A_1=0.3$, $\beta=0.01$, $q_0=0.1$, $A_2=0.1$.   Top left panel: Integrable (cubic) NLS. Top right panel: Quartic NLS ($p=3/2$). Bottom left panel: Saturable NLS. 
Bottom right panel: Evolution of the $L^\infty$ norm of the distance function $\widetilde \Delta$ for $x\in [-200, 200]$ and $t\in [0, 400]$.
More details in the text.
}
	\label{fig14}
\end{figure}
\textit{II. Algebraically decaying initial data on  top of the nonzero background $q_0$.} We conclude with the presentation of a numerical study concerning the dynamics emerging from  a quadratically decaying initial condition on  top of the finite background $q_0>0$. For this purpose, we will use initial conditions defined by the Peregrine soliton solution-the famous rational solution of the integrable NLS
\begin{equation}
    \label{NLSdelta}
i u_t + \frac{\delta}{2} u_{xx}+\mu|u|^{2}u=0,    
\end{equation}
given by the formula
		\begin{equation}
		u_{\text{PS}}(x,t;\delta;\mu;T_0;q_0)=q_0\left\{1-\frac{4\left[1+\frac{2i(t+T_0)}{\Lambda}\right]}{1+\frac{4x^2}{K_0^2}+\frac{4(t+T_0)^2}{\Lambda^2}}\right\}e^{\frac{i(t+T_0)}{\Lambda}}.
		\label{sprw}
		\end{equation}
The parameters in \eqref{sprw} are
		$\Lambda=\frac{1}{\mu\,q_0^2}$, $K_0=\sqrt{\delta\,\Lambda}$, and $T_0$ is a time-translation. 
Comparing the dispersion coefficients of the  NLS equations \eqref{NLSP} and \eqref{NLSdelta}, we see that they coincide when $\delta=2\nu$. With this observation we distinguish between two examples. 
\\[2mm]
\noindent
(a) $\delta\neq 2\nu$. We consider the parameters $\delta=\nu=1$, $\gamma=\mu=1$, $q_0=0.3$ and $T_0=-25$, in order to be compliant with the smallness conditions on the initial data when considering the initial condition 
\begin{equation}
\label{psinc}
    u(x,0)=u_{\text{PS}}(x,0;\delta;\mu;T_0;q_0),
\end{equation}
for the above set of parameters. This is an example where \eqref{sprw} is not an analytical solution of the NLS \eqref{NLSP} for its integrable case $p=1$. 

The theoretical results of \cite{bm2017} and \cite{blmt2018} establish  very similar spatiotemporal behavior of the integrable NLS as in  the case I, and this fact is confirmed in the  left panel of Figure \ref{fig11}. Comparing this panel with the central one which corresponds to the case $p=3/2$ and the right one which corresponds to the saturable NLS, we again observe that the saturable dynamics is closer to the integrable one than its power nonlinearity counterpart. The first events of large amplitude  occurring in all contour plots suggest also the presence of features reminiscent of Peregrine rogue waves (PRWs) governed  by the specific type of the initial condition.  To highlight this feature, as well as the similarities and differences of the solutions, we present  in Figure \ref{fig12}  snapshots of the profiles of the solutions for specific times. The similarities to PRWs can be observed for times around $t=25$ where the function \eqref{sprw} attains its maximum as shown in the fourth panel ($t=\tilde t \sim 30$ for the quartic and $t=\tilde t \sim 25$ for the saturable). For larger times ($t=50$, third panel), the differences of the profiles are reflected in the larger deviation of the norms of the distance function between the solutions.  The proximity around the center $x_0=0$ predicted by Theorem \ref{gex2} is also illustrated by the snapshots exhibiting the similarity of the central spikes.
\\[2mm]
\noindent
(b) $\delta=2\nu$. For $\mu=\gamma=1$, we choose this time $\delta=2$, $q_0=0.25$ and $T_0=50$ in the formula \eqref{sprw} and $\nu=1$ for the NLS \eqref{NLSP}. This time  the formula \eqref{sprw} defines an analytical solution of the NLS \eqref{NLSP} for its integrable case $p=1$. Using the initial condition \eqref{psinc} for this choice of parameters we observe crucial  differences in comparison with  the dynamics of the case \textit{a}, as it can be seen in Figure \ref{fig13}, for long times.  In the left panel, which shows the dynamics for the integrable case $p=1$, the PRW attains it maximum at $t=0$, since we have triggered initially, the  exact analytical PRW solution and  the MI effects predicted by \cite{bm2017,blmt2018} should occur much later.  This is not the case for the non-integrable models whose dynamics are represented in the central panel (power nonlinearity with $p=3/2$) and the right panel (saturable), respectively.  Figure \ref{fig10}, showing the time evolution of the norms of the distances for $x\in [-4,4]$, further confirms for the initial condition \eqref{psinc}, that when considering  sufficiently small spatial scales and small time intervals, one observes diminished deviations between  non-integrable and integrable dynamics  in all of the distances $\|\widetilde{\Delta}(t)\|_{\mathcal{X}}$. Particularly, for $t\in [-50,50]$, the behavior of the norms for the saturable NLS suggests the proximity of the first RW occurring close to $t=0$, as in the integrable case $p=1$. For the quartic nonlinearity, the first RW appears much later at $t\sim 50$, in compliance with the larger deviation of norms. However, both patterns and norms show that for $t\in [-50,50]$, the dynamics of all systems are similar. This is an illustration that Theorems \ref{nzbc-t} and \ref{gex2} may capture important nonlinear effects, particularly when these occur within short time scales.
\textit{III. Dynamics for non-constant $\zeta (x)$}. We conclude with a brief presentation of numerical results concerning  the case of the nonzero boundary conditions \eqref{nvDBC}, this time with a non-constant $\zeta(x)$. We consider a simple example of a non-constant background 
\begin{eqnarray}
\label{last1}
\zeta(x)=A_1\exp(-\beta x^2)+q_0, \quad A_1, \beta, q_0>0,
\end{eqnarray}
so that $\zeta(x) \to q_0 >0$ at an exponential rate as $|x|\rightarrow\infty$. For $\zeta(x)$ given by \eqref{last1}, we  solve  numerically the finite interval problems approximating the ones on the infinite line for the modified NLS equations~\eqref{cNLS} and \eqref{pNLS} supplemented with the nonzero boundary conditions \eqref{nvDBC}. The corresponding initial conditions are
\begin{eqnarray}
\label{last2}
\phi(x,0)=\Phi(x,0)=iA_2\mathrm{sech}x, \quad A_2>0,
\end{eqnarray}
which approximately satisfy the  zero Dirichlet boundary conditions \eqref{Dvbcn} for large $L$.    We return to the solutions $u$ and $U$ of the original equations \eqref{NLS} and \eqref{noninNLS} with the nonzero boundary conditions \eqref{uU-bc} via the transformations \eqref{zhi4} and \eqref{zhi6}. 
The top left panel of Figure \ref{fig14} shows the spatiotemporal evolution of the density $|u(x,t)|^2$  for the cubic NLS equation, and the top right and bottom left panels for the quartic ($p=3/2$) and the saturable nonlinearities respectively. The bottom right panel depicts the evolution of the norm $\| \widetilde \Delta(t)\|_{L^\infty}$ for $x\in [-200, 200]$ and $t\in [0, 400]$. The parameters are $A_1=0.3$, $\beta=0.01$ for the  non-constant background \eqref{last1} and $A_2=0.1$ for the initial condition \eqref{last2}; $\zeta(x)$ decays slowly to $q_0$. For this example, $\left\|\zeta-q_0\right\|_{L^2(\mathbb{R})}=1.06$(so we are in the limit of the theoretical assumption for the smallness condition for  $\left\|\zeta-q_0\right\|_{L^2(\mathbb{R})}$) and $\left\|\Phi(0)\right\|_{L^2(\mathbb{R})}=0.16$.   Although the long-time asymptotics are covered by the results of \cite{bm2017} (recall that the initial conditions of the original problems are of the form $u(x,0)=\phi(x,0)+\zeta(x)$), interesting dynamics are generated, as we observe the emergence of waveforms reminiscent to Kuznetzov-Ma  breathers.  For the emergence of such waveforms in the case of vanishing boundary conditions with Gaussian initial data, we refer to \cite{bs}. As it is expected by the analysis of the previous cases, the corresponding waveforms for the cubic and saturable NLS equations are very similar. They resemble  a two-period Kuznetzov-Ma breather alike waveform, with the peak of the larger amplitude reminiscent of the form of a second-order rogue wave. In the quartic case, after their initial stage, the dynamics stabilize to a single period Kuznetzov-Ma breather alike waveform. Despite the differences of the waveforms, this example still shows that the proximity analysis may justify consequently the  proximity of the dynamics between the integrable and the non-integrable models (as it is also shown by the evolution of $\|\widetilde{\Delta}(t)\|_{L^{\infty}}$ in the bottom right panel of Figure~\ref{fig14}). This is showcased by the stronger localization of the centered oscillations, similar to the one exhibited by Kuznetzov-Ma breathers.

\end{document}